\documentclass[a4paper,11pt]{article}
\usepackage[top=3.0cm,bottom=3.0cm,left=2.50cm,right=2.50cm,a4paper]{geometry}
\usepackage{siunitx}
\usepackage{color}
\usepackage[colorlinks]{hyperref}
\hypersetup{
    colorlinks=true, %% false ==> Boxed links; true ==> colored links
    linkcolor={blue!50!black},
    citecolor={green!60!black},
    urlcolor= {red!80!black}
}
\usepackage{amsmath,amsthm,stmaryrd}
\usepackage{amssymb}
\usepackage{upgreek}
\usepackage{mdframed,empheq}
\usepackage{float,moredefs}
\usepackage{graphicx,pgfplots}
\usepackage[]{caption}
\usepackage{enumitem,longtable}
\usetikzlibrary{fit,matrix,arrows,decorations.pathmorphing,patterns}
\usepgfplotslibrary{groupplots}
\usetikzlibrary{shapes,arrows,positioning,fit}
\usepackage[textsize=small,colorinlistoftodos]{todonotes}
\pgfplotsset{compat=newest}
\usepackage[]{algorithm,algorithmic}
\usepackage{mathbbol}
\usepackage{comment}

\numberwithin{equation}{section}
\usepackage{subfig}
\usepackage[]{todonotes}
\newcommand*\samethanks[1][\value{footnote}]{\footnotemark[#1]}
\newcommand{\RR}{\mathbb{R}}

\newcommand{\mathrmb}[1]{\boldsymbol{\mathrm{#1}}}

\newcommand{\xib}{\boldsymbol{\xi}}

\newcommand{\Rsun}{R_{\mathrm{S}}}
\newcommand{\ra}{r_{\mathrm{a}}}
\newcommand{\rs}{r_{\mathrm{s}}}

\newcommand{\bx}{\mathbf{x}}

\newcommand{\Ehe}{E_{\text{he}}}
\newcommand{\Atmo}{\texttt{AtmoI}}

\newcommand{\rb}{r_b}
\newcommand{\rmax} {r_{\max}}

\newcommand{\alatmo}{\alpha_{\mathrm{a}}}
\newcommand{\adatmo}{\gamma_{\mathrm{a}}}
\newcommand{\katmo}{k_{\mathrm{a}}}

\newcommand{\Lconjug}{\mathcal{L}_{\mathrm{cjg}}} % \mathcal{L}_{\ell}
\newcommand{\Lorigin}{\mathcal{L}_{\mathrm{org}}} % \mathfrak{L}_{\ell}

\newcommand{\Vscalar}{V_{\ell}^{\mathrm{scalar}}}

\newcommand{\ksf}{\mathsf{k}}
\newcommand{\Fl} {\mathsf{F}_{\ell}}
\newcommand{\Fz} {\mathsf{F}_{0}}

\newcommand{\alphap}   {\alpha_{p_0}}

\newcommand{\ii}       {\mathrm{i}}
\newcommand{\alphagamp}{\alpha_{\gamma p_0}}
\newcommand{\wronskian}{\mathcal{W}}
\newcommand{\Zrbc}     {\mathcal{Z}}

\newcommand{\ZrbcNL}      {\Zrbc_{\textsc{nl},\ell}}
\newcommand{\ZrbcNLa}     {\Zrbc_{\mathrm{a}\textsc{nl},\ell}}
\newcommand{\ZrbcNLaG}    {\Zrbc_{\mathrm{a}\textsc{nlg},\ell}}
\newcommand{\ZrbcHFzero}  {\Zrbc_{\textsc{HF0}}}
\newcommand{\ZrbcHFa}     {\Zrbc_{\textsc{HF1}}}
\newcommand{\ZrbcHFb}     {\Zrbc_{\textsc{HF2},\ell}}
\newcommand{\ZrbcHFGa}    {\Zrbc_{\textsc{HFG1}}}
\newcommand{\ZrbcHFGb}    {\Zrbc_{\textsc{HFG2},\ell}}
\newcommand{\ZrbcSAIzero} {\Zrbc_{\textsc{SAI0}}}
\newcommand{\ZrbcSAIa}    {\Zrbc_{\textsc{SAI1},\ell}}

\newcommand{\ZrbcSAIGzero}{\Zrbc_{\textsc{SAIG0}}}
\newcommand{\ZrbcSAIGa}   {\Zrbc_{\textsc{SAIG1},\ell}}
\newcommand{\ZrbcSAIGb}   {\Zrbc_{\textsc{SAIG2},\ell}}

\newlength{\plotwidth}  \newlength{\modelwidth}
\newlength{\plotheight} \newlength{\modelheight}

\newcommand{\datafile} {}

\newcommand{\datafileB}{}
\newcommand{\datafileC}{}
\newcommand{\modelfile}{}
\newcommand{\cbtitle}  {}
\newcommand{\modelfileA}{}
\newcommand{\cbtitleA}  {}
\newcommand{\modelfileB}{}
\newcommand{\cbtitleB}  {}
\newcommand{\modelfileC}{}
\newcommand{\cbtitleC}  {}
\newcommand{\myylabel}{}
\newcommand{\dataA}{} \newcommand{\legendA}{}
\newcommand{\dataB}{} \newcommand{\legendB}{}
\newcommand{\dataC}{} \newcommand{\legendC}{}
\newcommand{\dataD}{} \newcommand{\legendD}{}
\newcommand{\dataE}{} \newcommand{\legendE}{}
\newcommand{\dataF}{} \newcommand{\legendF}{}
\newcommand{\dataG}{} \newcommand{\legendG}{}
\newcommand{\dataH}{} \newcommand{\legendH}{}

\newcommand{\myfreq}{}
\newcommand{\myfreqA}{}
\newcommand{\myfreqB}{}
\newcommand{\myfreqC}{}
\newcommand{\xtickloc}{}

\newcommand{\GPP}{G^{\mathrm{PP}}_\ell}
\newcommand{\GBP}{G^{\mathrm{BP}}_\ell}
\newcommand{\GPB}{G^{\mathrm{PB}}_\ell}
\newcommand{\GBB}{G^{\mathrm{BBreg}}_\ell}

\newcommand{\err}   {\mathfrak{E}}

\newcommand{\errpt} {\mathfrak{e}}
\makeatletter
\pgfdeclarepatternformonly[\LineSpace]{my north east lines}
                                      {\pgfqpoint{-1pt}{-1pt}}
                                      {\pgfqpoint{\LineSpace}{\LineSpace}}{\pgfqpoint{\LineSpace}
                                      {\LineSpace}}%                                      
{
    \pgfsetcolor{\tikz@pattern@color} %new code
    \pgfsetlinewidth{0.4pt}
    \pgfpathmoveto{\pgfqpoint{0pt}{0pt}}
    \pgfpathlineto{\pgfqpoint{\LineSpace + 0.1pt}{\LineSpace + 0.1pt}}
    \pgfusepath{stroke}
}
\pgfdeclarepatternformonly[\LineSpace]{my north west lines}
                                      {\pgfqpoint{-1pt}{-1pt}} %% continuous lines
                                      {\pgfqpoint{\LineSpace}{\LineSpace}}{\pgfqpoint{\LineSpace}
                                      {\LineSpace}}%
{
    \pgfsetcolor{\tikz@pattern@color} %new code
    \pgfsetlinewidth{0.4pt}
    \pgfpathmoveto{\pgfqpoint{0pt}{0pt}}
    \pgfpathlineto{\pgfqpoint{\LineSpace + 0.1pt}{\LineSpace + 0.1pt}}
    \pgfusepath{stroke}
}
\makeatother %new code
\newdimen\LineSpace
\tikzset{
    line space/.code={\LineSpace=#1},
    line space=3pt
}
\newtheorem{lemma}      {Lemma}
\newtheorem{proposition}{Proposition}
\newtheorem{remark}     {Remark}

\newtheorem{corollary} {Corollary}
\usepackage[capitalize,nameinlink]{cleveref}
\crefname{section}   {Section}   {Sections}
\crefname{subsection}{Subsection}{Subsections}
\Crefname{section}   {Section}   {Sections}
\Crefname{subsection}{Subsection}{Subsections}
\Crefname{figure}    {Figure}    {Figures}

\crefformat{equation}{\textup{#2(#1)#3}}
\crefrangeformat{equation}{\textup{#3(#1)#4--#5(#2)#6}}
\crefmultiformat{equation}{\textup{#2(#1)#3}}{ and \textup{#2(#1)#3}}
{, \textup{#2(#1)#3}}{, and \textup{#2(#1)#3}}
\crefrangemultiformat{equation}{\textup{#3(#1)#4--#5(#2)#6}}%
{ and \textup{#3(#1)#4--#5(#2)#6}}{, \textup{#3(#1)#4--#5(#2)#6}}{, and \textup{#3(#1)#4--#5(#2)#6}}
\Crefformat{equation}{#2Equation~\textup{(#1)}#3}
\Crefrangeformat{equation}{Equations~\textup{#3(#1)#4--#5(#2)#6}}
\Crefmultiformat{equation}{Equations~\textup{#2(#1)#3}}{ and \textup{#2(#1)#3}}
{, \textup{#2(#1)#3}}{, and \textup{#2(#1)#3}}
\Crefrangemultiformat{equation}{Equations~\textup{#3(#1)#4--#5(#2)#6}}%
{ and \textup{#3(#1)#4--#5(#2)#6}}{, \textup{#3(#1)#4--#5(#2)#6}}{, and \textup{#3(#1)#4--#5(#2)#6}}

\crefdefaultlabelformat{#2\textup{#1}#3}

\crefname{proposition}{Proposition}{Propositions}
\Crefname{proposition}{Proposition}{Propositions}
\crefname{definition} {Definition} {Definitions}
\Crefname{definition} {Definition} {Definitions}
\crefname{theorem}    {Theorem}    {Theorems}
\Crefname{theorem}    {Theorem}    {Theorems}
\crefname{remark}     {Remark}     {Remarks}
\Crefname{remark}     {Remark}     {Remarks}
\crefname{assumption} {Assumption} {Assumptions}
\Crefname{assumption} {Assumption} {Assumptions}
\crefname{corollary}  {Corollary} {Corollaries}
\Crefname{corollary}  {Corollary} {Corollaries}

\title{Assembling algorithm for Green's tensors 
      and absorbing boundary conditions 
      for Galbrun's equation in radial symmetry}
\author{
Ha Pham\thanks{Project-Team Makutu, Inria, University of Pau and Pays de l'Adour, TotalEnergies, CNRS UMR 5142, France.}
\and
Florian Faucher\samethanks[1]
\and
Damien Fournier\thanks{Max-Planck-Institut f\"ur 
                       Sonnensystemforschung, Justus-von-Liebig-Weg 3, 
                       37077 G\"ottingen, Germany.}
\and
H\'el\`ene Barucq\samethanks[1]
\and
Laurent Gizon\samethanks[2]
             \thanks{Institut f\"ur Astrophysik und Geophysik, 
                     Georg-August-Universit\"at G\"ottingen, 
                     Friedrich-Hund-Platz 1, 37077 G\"ottingen, Germany.}
}

\date{}

\usetikzlibrary{fpu}
\begin{document}
\maketitle

\begin{abstract}

Solar oscillations can be modeled by Galbrun's equation which 
describes Lagrangian wave displacement in a self-gravitating 
stratified medium.
For spherically symmetric backgrounds, we construct 
an algorithm to compute efficiently and accurately 
the coefficients of the Green's tensor of the time-harmonic equation in vector 
spherical harmonic basis.  
With only two resolutions, our algorithm provides values of the kernels 
for all heights of source  and receiver, and prescribes analytically the 
singularities of the kernels. 
We also derive absorbing boundary conditions (ABC) to 
model wave propagation in the atmosphere above 
the cut-off frequency.
The construction of ABC, which contains varying gravity terms, 
is rendered difficult by the complex behavior of the solar potential 
in low atmosphere and for frequencies below the Lamb frequency.
We carry out extensive numerical investigations to compare and 
evaluate the efficiency of the ABCs in capturing outgoing solutions. 
Finally, as an application towards helioseismology, we 
compute synthetic solar power spectra that contain pressure modes 
as well as internal-gravity (g-) and surface-gravity (f-) ridges which are 
missing in simpler approximations of the wave equation.
For purpose of validation, the location of the ridges in the synthetic power 
spectra are compared with observed solar modes.

\end{abstract}

% \tableofcontents \newpage
% ===================================================
\section{Introduction}
\label{into::sec}
% ===================================================

% \modif{Understanding the propagation of waves in stars is a key ingredient for astereoseismology, a field that aims at inferring the internal structure of stars from observed time series of brightness variations. The equations of stellar oscillations have been derived in \cite{lynden1967stability} and are now routinely used to determine the theoretical eigenfrequencies of stars in some standard codes such as ADIPLS \cite{christensen2008adipls} and GYRE \cite{Townsend2013}. The computed eigenvalues are then compared to the observed power spectra to determine the properties of stars, see, e.g., the review \cite{Aerts2021}. A similar approach can be conducted to learn about the structure of planets, e.g., \cite{Dewberry2021} for Saturn.}
% \modif{The equations of stellar oscillations that represent the 3D Lagrangian displacement $\boldsymbol{\xi}$ on top of an adiabatic static background have been derived by Lynden-Bell and Ostriker in \cite{lynden1967stability}.}

Understanding the propagation of waves in stars is a key ingredient for astereoseismology, a field that aims at inferring the internal structure of stars from observed time series of brightness variations. A classical choice to model stellar oscillations is given by the equations derived by Lynden-Bell and Ostriker in \cite{lynden1967stability} which describe the 3D Lagrangian displacement $\boldsymbol{\xi}$ on top of a stationary self-gravitating 
background in hydrostatic equilibrium. These equations are now routinely solved to compute theoretical eigenfrequencies and free oscillations of stars which are then compared, using techniques of global seismology, to observed power spectra in order to constrain the internal properties of stars,
%, which include the Sun, our closest star,
see e.g., \cite{Aerts2021} for a review. These same equations are also employed in the seismology of certain types of planets, e.g. in ring seismology of Saturn, \cite{Dewberry2021}.
For the Sun, our closest star, the study of its internal structure is known as helioseismology, \cite{Christensen_Dalsgaard_2002} and is mostly done by analyzing its gravity and acoustic oscillations. The effect of magnetic fields (as well as rotation and radiative transfer) on these waves is small\footnote{This simplification is considered in the standard
codes computing eigenfrequencies for stars such as ADIPLS \cite{christensen2008adipls} and GYRE \cite{Townsend2013}.}
and can be treated a posteriori as a perturbation from a reference (non-magnetic) background (see, e.g., the review \cite{Aerts2021}).

The main difference between helioseismology and asteroseismology is the type (and quantity) of observations which consist of not only disk-integrated brightness variations but also spatially resolved Doppler velocity maps. Consequently, while employing the same modeling equations of \cite{lynden1967stability}, one can aim at recovering more structural information than for distant stars with techniques of local seismology.
% In global helioseismology, the spectral structure of the wave operator plays an important role and, similarly to other stars, the eigenfrequencies of oscillations are used \cite{Basu2016,christensen2014lecture,Gough1993}. 
While the eigenfrequencies of the wave operators are the key elements in global helioseismology \cite{Basu2016,christensen2014lecture,Gough1993}, local helioseismology requires the Green's tensor to recover more details, such as 3D perturbations, cf. \cite{gizon2005local}.
In fact, the importance of Green's kernels holds in a larger family of seismic techniques, in particular in passive imaging which relies on the correlation of random wavefields, cf. \cite{gouedard2008cross}, and to which local helioseismology belongs.

% However, for the Sun, more details can be inferred, e.g. recovering 3D perturbations, by using techniques of local helioseismology which requires Green's tensor, cf. \cite{gizon2005local}.
% This technique is considered for Earth, oceanography, Volcanoes e.g. it is needed to 
% compute Born sensitivity kernels \cite{boning2016sensitivity,bhattacharya2020general}.
% \modif{
% The spectral structure of  the wave operator
% % $\boldsymbol{\mathcal{L}}$ \dam{not defined yet, should we just say the wave operator?}
% plays an important role in global helioseismology which works with
% eigenfrequencies of oscillations \cite{Basu2016},
% while the Green's tensor $\mathbb{G}(\mathbf{x},\mathbf{y})$ 
% is a fundamental quantity in local 
% helioseismology \cite{gizon2005local,christensen2014lecture,Gough1993}, 
% e.g. it is needed to 
% compute Born sensitivity kernels \cite{boning2016sensitivity,bhattacharya2020general}.}
 %Since solar data are dominated by acoustic waves, most works in helioseimology employ acoustic waves in the solar convection zone.
%In particular, the Green's tensor $\mathbb{G}(\mathbf{x},\mathbf{y})$ and the spectral structure of  $\boldsymbol{\mathcal{L}}$ play an important role in helioseismology, cf. e.g \cite{gizon2005local,christensen2014lecture,Gough1993} and are the main  ingredients to compute the Born sensitivity kernels \cite{boning2016sensitivity,bhattacharya2020general}.

The Lagrangian formulation of the stellar oscillation equations is also adopted in other communities such as aeroacoustics and hydroacoustics to model sound propagation in a moving fluid. Neglecting the perturbations to the gravitational potential (Cowling's approximation) \cite{Cowling1941}, the equations of stellar oscillations reduce to one called Galbrun's equation also derived by Galbrun \cite{galbrun1931propagation} in the context of aeroacoustics.
Here we refer to \cite{maeder202090} for a review of the history of this equation in these contexts. Although there exist several formulations (e.g., linearized Euler formulation) in acoustics, due to its advantages, e.g., in the derivation of interface jump conditions \cite{treyssede2009jump} and energy conservation's laws \cite{brazier2011derivation}, Galbrun's equation receives extensive attention, with a large body of literature dedicated to studying its theoretical properties as well as numerical resolution, cf. \cite{maeder202090}. Furthermore, it provides a unifying framework to describe wave propagation in the presence of solid-fluid interactions and is thus employed to model waves in the fluid part of an Earth containing oceans, cf. \cite{gharti2023spectral}.

% \modif{\cite{galbrun1931propagation} which is often also used in aeroacoustics, see the review \cite{maeder202090}. While less known than its Eulerian counterpart,  Galbrun's equation can be discretized efficiently and on stable way using finite elements \cite{Nennig2011} and is often used in applications, such as aeroacoustics \cite{Peyret2001, treyssede2003mixed} or to learn about the structure of planets, such as the Earth \cite{DahlenTromp}, or .}

%{However, care should be taken in choosing the finite element basis to avoid spurious unphysical solutions \cite{Dietzsch2014}.}

% The original Galbrun's equation is obtained by Lynden-Bell and Ostriker 
% \cite{lynden1967stability} in astrophysics and by Galbrun
% \cite{galbrun1931propagation} in aeroacoustics, cf. \cite{maeder202090}. 
%Ignoring flow and rotation, the simplified

In this work, we consider the time-harmonic Galbrun's equation, under spherical symmetry and without flows, which models the displacement $\xib$, on top of a background medium characterized by its pressure $p_0$, density $\rho_0$, adiabatic index $\gamma$, and gravitational 
potential $\phi_0$.  Our working equation is written in a coordinate system scaled by $\Rsun$, the radius of the Sun, and we introduce attenuation 
via the parameter $\Gamma>0$,
\begin{equation} \label{main_eqn_Galbrun}
\begin{aligned}
 \boldsymbol{\mathcal{L}} \, \xib =  \mathbf{f}\, , \hspace*{0.2cm} \text{in } \mathbb{R}^3 \,, \hspace*{0.5cm} &\text{where}\,\hspace*{0.2cm}
\Rsun^2\, \boldsymbol{\mathcal{L}} \, \xib
          := -  \rho_0 \,\Rsun^2 (\omega^2+2\mathrm{i}\omega\Gamma)\boldsymbol{\xi}
                         -\nabla \big( \gamma\, p_0 \nabla\cdot \boldsymbol{\xi} \big)    \\
   &\hspace*{.9cm} + (\nabla p_0) (\nabla \cdot\xib) -  \nabla [ (\xib \cdot \nabla) p_0 ] + (\xib\cdot \nabla) \nabla p_0   \, + \,  \rho_0 (\xib \cdot \nabla) \nabla \phi_0  \,.
\end{aligned}
\end{equation}
We assume that the real-valued background coefficients only depend on 
the radial coordinate and are given by model \texttt{S}-\Atmo~constructed 
in \cite{BackgroundRep}.
%Additionally, unlike the more familiar problem  \cref{BVP::def} which is defined in the Sun, problem \cref{main_eqn_Galbrun} is considered in the whole $\RR^3$ with coefficients given by model \texttt{S}-\Atmo~constructed in \cite{BackgroundRep}.
This background model combines the standard solar model \texttt{S} 
\cite{christensen1996current} defined in the interior of the Sun
with an isothermal atmospheric model \Atmo~which retains the main 
features of that employed for 
scalar wave equation in \cite{fournier2017atmospheric,barucq2018atmospheric,barucq2020outgoing},
that is an exponentially decaying density and a constant sound speed. 
For equation \cref{main_eqn_Galbrun}, we present an algorithm to compute efficiently
the components of the Green's tensor denoted by $\mathbb{G}(\mathbf{x},\mathbf{y})$ in vector spherical harmonic basis, called directional kernels (Topic 1). 
We also investigate absorbing boundary conditions to approximate 
outgoing-at-infinity solutions (Topic 2).
The theoretical existence and characterization of the outgoing directional kernels were established in \cite{barucq2021outgoing}. 
For well-posedness results of Galbrun's equation in $\RR^3$, 
we refer to \cite{halla2022treatment}. The Hybridizable 
 Discontinuous Galerkin (HDG) method is employed 
for  all numerical experiments in this article, and we refer to \cite{barucq:hal-02423882} for detailed implementation of the method to the modal equations of \cref{main_eqn_Galbrun}.

%Precursor to this work is the numerical investigation 
% in \cite{barucq:hal-02423882}, to which we refer for details 
% on the implementation of the Hybridizable 
% Discontinuous Galerkin (HDG) method employed 
% for all  numerical experiments in this article. 

% \modif{Galbrun's equation is one of the key equations employed in asteroseismology and helioseismology to describe wave oscillations. In particular two routinely used codes, ADIPLS \cite{christensen2008adipls} and GYRE \cite{Townsend2013} are solving the eigenvalue problem associated to \cref{main_eqn_Galbrun}. The computed eigenvalues are then compared to the observed power spectra to determine the properties of stars, see, e.g., the review \cite{Aerts2021}. A similar approach can be conducted to learn about the structure of planets, e.g., \cite{Dewberry2021} for Saturn. Another domain of application of Galbrun's equation concerns aeroacoustics, see the review \cite{maeder202090}. While less known than its Eulerian counterpart,  Galbrun's equation can be discretized efficiently and on stable way using finite elements \cite{Nennig2011} and is often used in applications, see, e.g., \cite{Peyret2001, treyssede2003mixed}. However, care should be taken in choosing the finite element basis to avoid spurious unphysical solutions \cite{Dietzsch2014}.
% }
   
These two topics were also addressed for a scalar equation  proposed in \cite{gizon2017computational}
which approximates \cref{main_eqn_Galbrun} in ignoring gravity effects. While simpler, this scalar equation already presents computational challenges due to the strong solar stratification and the unboundness of the domain. Regarding the second point, Dirichlet boundary conditions which are usually employed are not adapted as they create spurious reflections for waves above the atmospheric cut-off frequency $\sim 5.2\si{\milli\Hz}$. Waves above this frequency 
are not trapped and can propagate into the low atmosphere, as observed 
and predicted by theory \cite{kumar1990observed,jefferies1998high}.  Radiation boundary conditions are thus necessary for proper modeling and have been derived, for the scalar equation, in \cite{barucq2018atmospheric} and analyzed in a solar context in \cite{fournier2017atmospheric}. An efficient algorithm to compute the Green's function of this scalar problem with appropriate boundary conditions has been proposed in \cite{barucq2020efficient}.
This Green's function has shown its utility in important applications in helioseismology: it improves the capability to image scatterers (sunspots) on the far side of the Sun \cite{Yang2023}, a step forward for space weather applications; it allows imaging large-scale flows in the solar convection zone such as the meridional flow \cite{Gizon2020}, an essential ingredient to understand the solar dynamo.
Despite its attractiveness, being simpler and retaining most of the physics of solar acoustic modes, the scalar problem presents limitations: with gravity effects ignored, 
the surface gravity mode (f-mode) cannot be resolved, cf. \cite[Figure 10]{gizon2017computational} which shows a comparison between observed and simulated power spectra; and the modeling of acoustic modes is impacted.
Consequently, the resolution of \cref{main_eqn_Galbrun} is required to provide a more realistic representation of wave propagation in the solar interior.

We now highlight the main features and novelty of the two stated topics, addressed for the vector equation \cref{main_eqn_Galbrun}.

\noindent \textbf{Topic 1.} The first main result concerns the computation of the
Green's tensor $\mathbb{G}$. This is reduced, in spherical symmetry, to 
computing\footnote{Here we are only concerned with methods which compute the Green's 
kernel directly, in contrast to the spectral decomposition approach,
in which the kernel is obtained through 
its spectral decomposition in terms of eigenfunctions (for zero surface pressure condition), 
cf. \cite[Appendix B]{boning2016sensitivity}.} 
the coefficients of the components of $\mathbb{G}$ in  
vector spherical harmonic (VSH) tensor basis, called directional kernels.
Our goal is to obtain efficiently and accurately these directional kernels, 
which contain different levels of regularity, for all heights of source and receiver.
This is achieved, in an approach called `assembling', 
described by \cref{compoG_reg:prop,algorithm:assemble}. 
\cref{compoG_reg:prop} provides explicit expressions of these kernels in terms 
of two specific solutions (and their derivatives) 
to the modal ODE with no source. Implementing these results, 
\cref{algorithm:assemble} can compute numerically all kernels, 
in addition to prescribing analytically their singularities, 
at the cost of only two simulations.

This work extends to the Galbrun's equation the result 
of \cite[Algorithm 3.2]{barucq2020efficient} 
for the scalar equation which has demonstrated, for this case, 
the advantages of the `assembling' approach compared to 
direct resolution with a Dirac-source.
The benefits of \cref{algorithm:assemble} go beyond those demonstrated 
in \cite{barucq2020efficient}, cf. \cref{subsection:assembling-algorithm}.
For Galbrun's equation, analytical prescription of the regularity of 
the directional kernels, which forms part of \cref{compoG_reg:prop}, 
is necessary and provides the only viable option that is numerically stable 
for their computation.
A more direct attempt in which the coefficients are obtained by solving 
the modal equation with a Dirac source in a similar vein to \cite[Algorithm 3.1]{barucq2020outgoing}, 
would be ridden with difficulties, as pointed out in \cref{App1_vector::rmk}.
In another word, the assembling idea of \cref{algorithm:assemble} has 
proven its utmost utility with the vector equation.
% We note that our algorithm can work with any generic 
% boundary condition placed at the end of the computational 
% domain.
%We note that our algorithm can be applied to 
%both problems \cref{main_eqn_Galbrun,BVP::def}, 
%with \cref{main_eqn_Galbrun} requiring ABC 
%dealt with in Topic 2.

\medskip

\noindent \textbf{Topic 2.}
The unboundedness of the domain in \cref{main_eqn_Galbrun} necessitates constructing 
 boundary conditions and evaluating their accuracy in approximating
outgoing-at-infinity solutions theoretically constructed in \cite{barucq2021outgoing}.
Absorbing boundary conditions (ABC) for the scalar equation in spherical symmetry
were constructed  in  \cite{barucq:hal-02168467,barucq:hal-02423882} 
and numerically investigated in \cite{barucq2020efficient}; they comprise 
of a nonlocal boundary condition and its approximations which are local 
in $\ell(\ell+1)$ with harmonic degree $\ell$.
We will follow the spirit of these works to construct ABC
for $\boldsymbol{\mathcal{L}}$. 
However,the task here is more delicate, due to the complicated behavior of 
the modal Schr\"odinger operator of $\boldsymbol{\mathcal{L}}$, 
denoted by $-\partial_r^2+ V_\ell$, in comparison to that of the 
scalar equation, denoted by $-\partial_r^2+\Vscalar$, cf. \cref{resultJDE::subsec}. 
In addition to the propagative region\footnote{These regions exist for frequencies below the buoyancy frequencies, cf. \cite{christensen2014lecture} and \cite[Section 7.1]{barucq:hal-03406855}, which gives rise to gravity ridges in synthetic power spectrum.} that $-\partial_r^2+ V_\ell$ possesses 
in the interior of the Sun, % for our equation which is also defined in the atmosphere, 
an other propagative region is observed in the atmosphere 
for frequencies below the Lamb frequency, which corresponds 
to low frequencies $\omega$ and high-degree modes $\ell$.
     
While both of these features are absent in the scalar wave equation,
it is the propagative region in low atmosphere which renders 
difficult the construction of ABC. Due its presence,
an outgoing solution associated with $-\partial_r^2 +  V_\ell$ converges 
slowly to the oscillatory behavior at infinity prescribed 
by $e^{\ii \mathsf{k}_a r}$, with the slowest convergence observed
 for $(\omega,\ell)$ below the Lamb frequency curve.
While the behavior at infinity 
theoretically takes effect in very high atmosphere (i.e. $r \gg 1$),
artificial boundaries should not be 
placed\footnote{This is not only question of computational cost
but also of physical modeling.} 
after the photosphere, which corresponds to 
scaled height $r = 1.003$.
On the other hand, the outgoing solution associated with
$-\partial_r^2+\Vscalar$ 
 converges quickly to the infinity behavior, and its ABCs are imposed
at radius $r=1.0008$ in \cite{barucq2020outgoing,barucq2020efficient}.
Finally, for regions on which $V_\ell$ is similar in nature to $\Vscalar$, the construction
of ABC remains more technical, due to the lack of 
analytic solutions, and the fact that $V_\ell$, unlike $\Vscalar$, is 
not a polynomial in $\ell(\ell+1)$.

%    
%         The resulting local conditions 
%      are indeed different from those obtained form scalar, bearing two distinct features:
%      Firstly, there is  the presence of the 
%      gravity term which despite being an term of order $\mathsf{O}(r^{-3})$
%      play an important role in the low atmosphere modifying the generalized harmonic mode at level $r^{-2}$.
%      Secondy, the coefficients of $r^{-2}$ which are simply $\ell(\ell+1)$ for scalar operator,
%      is now replaced by `generalized mode' , cf. discussion in \cref{??}.

%% from discussion of Damien
% ... on the absorbing BC at 500 km above the surface and consider the nonlocal as reference. It is the closest to what we did in the scalar and then we can explain and show the differences with the radiation BC. Numerically we don't want to extend the domain too far away that was the initial argument for working on BC. But it is of course a bit tricky to work with a potential that is still varying in the atmosphere...
% not having to discretize the Dirac, solution for all the sources  and have an opening plot at the end with the difficulties such as the differences RBC / ABC in this case.

\medskip

The organization of the paper is as follows. 
After introducing necessary notations and auxiliary 
quantities in \cref{notation:sec}, 
we obtain the `assembling' formula in \cref{section:modeling},
cf. \cref{compoG_reg:prop,algorithm:assemble}. 
In \cref{section:rbc}, we discuss how nonlocal and local BCs are obtained for 
the vector equation, taking into account contribution of gravity and the 
complex behavior of the potential in the atmosphere. 
Numerical investigations of their accuracy in approximating  
outgoing-at-infinity solutions are carried out in  
\cref{section:numerics:original-and-conjugated,section:numerics:original-and-conjugated}.
These comprise of constructing reference solutions, and ascertaining the validity of 
the constructed ABC in low atmosphere. 
Additionally, these investigations also complete preliminary 
results obtained in \cite{barucq:hal-02423882} which compare between 
two formulations of the modal equation (so-called `original' and `conjugated' equations). 
In \cref{section:observables}, as an application towards helioseismology,
these results are employed to compute synthetic solar spectrum which displays gravity ridges
which are absent when the modeling is done with the scalar equation.
The nature of these ridges are validated by comparing
with eigenvalues computed with the software \texttt{gyre} \cite{Townsend2013},
and with the observed solar mode frequencies from \cite{Korzennik2013,Larson2015}.
All numerical implementations to compute kernels are realized with
the open-source software \texttt{hawen}, \cite{Hawen2020}. % \url{https://ffaucher.gitlab.io/hawen-website/}.

% =====================================================================
\section{Physical parameters}
\label{notation:sec}
% =====================================================================

%% ==================================================================
%\subsection{Scaled coordinate system} 
%% ==================================================================

Here we recall the main properties of the background 
solar model \texttt{S}-\Atmo~constructed in \cite{BackgroundRep} 
which imposes constant sound speed and exponentially 
decaying density in the atmosphere, referred to as the
\emph{Atmo} features. 
% For the scalar equation (e.g., \cite{fournier2017atmospheric,barucq2018atmospheric}), 
% these features can be imposed at the end of
% model \texttt{S} (\cite{christensen1996current}) 
% in $r=\rs$, which is just after the surface of the Sun 
% $\{r=1\}$ in scaled coordinates. 
% Note that scaled radius $r=1$ corresponds to the solar radius $\Rsun = \num{6.96e8} \si{\m}$.
To maintain regularity of the coefficients, 
a transition region denoted by $[\rs,\ra]$ is needed 
between the behaviour represented by model \texttt{S} 
and the Atmo features.
We work with values of $\rs$, $\ra$ from \cite{BackgroundRep},
\begin{equation}\label{rars::def}
   \rs\,=\,\num{1.0007126} \,,\quad \ra \,=\, \num{1.00073}\,.
\end{equation}
% \flo{While condition $\delta^{\mathrm{L}}_{\mathsf{p}}=0$ 
% is traditionnally placed at $r=\rs$}\todoflo{remove?}, 
Absorbing boundary conditions are placed at $r \geq \ra$. % to simulate wave propagating into the low atmosphere
To begin the description, we also need 
% the following quantity called 
the so-called \emph{inverse scale height} function 
$\alpha_\mathfrak{g}$ defined for $\mathcal{C}^1$ 
function $\mathfrak{g}(r)$,
\begin{equation}\label{scale_heights::def}
    \alpha_{\mathfrak{g}}(r) \, := \, -\dfrac{\mathfrak{g}'(r)}{ \mathfrak{g}(r)} \, .
\end{equation}

\subsection{Properties of model \texttt{S}-\Atmo}

We cite here the properties of the model 
\texttt{S}-\Atmo~which follows 
\cite[Assumption 1]{barucq2021outgoing}, with $r$ denoting the radial variable.
\begin{enumerate}
\item All background coefficients are radial, 
      which include
      density $\rho_0$, adiabatic index $\gamma$, attenuation $\Gamma(\omega,\cdot)$, pressure $p_0$, 
      and sound speed $c_0$. Additionally, $\rho_0, \gamma, \Gamma \in \mathcal{C}^2[0,\infty)$, 
       $p_0 \in \mathcal{C}^3[0,\infty)$, and
 \begin{subequations}\label{assump:g1}
   \begin{align}
     & \Gamma(\omega,r)\geq 0\,,\qquad 1<\,\gamma(r) \,<2, \qquad
     \rho_0(r) > 0, \qquad p_0(r)  > 0,\label{ad_exp_hypo}\\       
     & c_0^2\rho_0 = \gamma p_0  \hspace*{0.3cm}\text{(Adiabaticity)} , \hspace*{1.5cm} r\,c_0^{-1} \,\,\text{strictly increasing on } [0,\infty) \label{r_c_increase}\,.
\end{align}
\end{subequations}
%Here,  $r$ is the radial coordinate and $c_0$ is the sound speed.
\item  The \emph{gravitational potential} $\phi_0$ is defined  as the $L^2(\RR^3)$ solution to the Poisson equation,
\begin{equation}\label{gravpot::r}
 \dfrac{1}{r^2} \left( r^2 \phi_0'\right)' = 4\pi G \Rsun^2\rho_0,
 \end{equation}
 with gravitational constant
   G=\num{6.6743e-8}\si{\centi\meter\cubed\per\gram\per\second\squared}.
    Its radial derivative is given by, cf. \cite[Appendix G.3]{barucq:hal-02423882},
\begin{equation}
 \phi_0'(r)  =  \dfrac{4 \pi  G \Rsun^2}{r^2} \int_0^r \rho_0(s)     s^2 \mathrm{d} s\,.
\end{equation}

\item The hydrostatic equilibrium $p_0' = -\rho_0 \phi_0'$
is achieved only on $ [0,\rs]$ and deviation to this in 
the atmosphere is measured by
 quantity $\Ehe$ (unitless) which is non-positive and decreasing,
   \begin{equation} \Ehe(r)  := \dfrac{  \phi'_0(r) }{ c_0^{2}(r)}-  \dfrac{\alpha_{p_{0}}(r)}{\gamma(r)}\,, \hspace*{0.4cm} \text{with} \hspace*{0.2cm} \begin{cases}\Ehe=0 \hspace*{0.2cm}\text{on }   [0,\rs]\,, \\
 \Ehe \leq 0 \text{ and } \Ehe  \text{ decreasing on } [0,\infty)\end{cases} \,.\label{eq:hydrostatic}
\end{equation}
\item
  For $r\geq r_a$, parameters $\Gamma$,
   $\gamma$ and $c_0$ are constant, while $\rho_0$
  is exponentially decreasing, i.e., 
  \begin{equation}\label{specialfeaturesatmo}
    \Gamma(\omega,r)= \Gamma_{\mathrm{a}}(\omega),
    \quad \gamma(r) = \adatmo, 
    \quad c_0(r)   = c_\mathrm{a}, 
    \quad \rho_0(r) = \rho_0(\ra)\,  e^{-\alatmo (r - \ra)}\,,
  \end{equation}  
   with constants  $\Gamma_{\mathrm{a}}$, $\adatmo$, $c_\mathrm{a}$, $\alatmo$. 
Their values employed in this work are obtained from \cite{BackgroundRep}, 
\begin{equation}
 \adatmo = \num{1.6401}   \,, \hspace*{0.5cm}
 \alpha_{p_0}=\alpha_{\rho_0}=\alatmo = \num{6.6325e3} \,,\hspace*{0.5cm}
 \dfrac{c_\mathrm{a}}{\Rsun}   = \num{9.8608e-6}\si{\per\second}\,.
\end{equation}

\end{enumerate}

\subsection{List of main and auxiliary parameters}
\paragraph{Wavenumbers} The \emph{complex-frequency} $\sigma$  (in \si{\per\second}) and 
wavenumber $k_0$ (unitless) are defined in terms of 
the angular frequency $\omega$ and attenuation 
$\Gamma(\omega,r)$,
%\begin{equation} \label{eq:sigma-and-k0}
%\sigma(r, \,\omega) \,=\, \sqrt{ \omega^2 \, + \, 2\mathrm{i} \, \omega\, \Gamma(\omega,r) } \, , \qquad\qquad
%             k_0(r) \,= \, \Rsun \,\dfrac{\sigma(r)}{c_0(r)}\,.
%\end{equation}
\begin{equation} \label{eq:sigma-and-k0}
\sigma^2(r, \omega) =  \omega^2 +  2\mathrm{i} \,\omega \Gamma(\omega,r)  \, ,\hspace*{1cm}
             k_0^2(r) =  \Rsun^2 \dfrac{\sigma^2(r)}{c_0^2(r)}
                      =  \Rsun^2 \dfrac{\omega^2 +  2\mathrm{i} \,\omega \Gamma(\omega,r)}{c_0^2(r)}\,.
\end{equation}

\paragraph{Liouville change of variable} 
The following terms $\mathsf{F}_0(r;\omega)$ and $\mathsf{F}_\ell(r;\omega)$ play important roles in the definition of the coefficients of the modal ODEs,
\begin{equation}\label{sfFell::def} 
\mathsf{F}_0= k_0^2  r^2 -  r\Ehe
\,, \hspace*{1cm} \mathfrak{F}= c_0^2\rho_0 \mathsf{F}_0\,,\hspace*{1.5cm}  \mathsf{F}_{\ell}  = k_0^2  r^2- r\Ehe  -  \ell(\ell+1)\,.
\end{equation}
They also appear in the change of variable to transform  
the original ODE into Schr\"odinger form, together with coefficient
$\mathfrak{I}_{\ell}$ defined as,
  \begin{equation}  \mathfrak{I}_{\ell}  =   \dfrac{1}{r \,c_0\sqrt{\rho_0}} 
 \dfrac{\sqrt{\mathsf{F}_{\ell}}}{\sqrt{ \mathsf{F}_0}} = 
 \dfrac{1}{r\, c_0 \sqrt{\rho_0}} 
\sqrt{\sigma^2  - \mathcal{S}^2_\ell }\Big/ \sqrt{\sigma^2  -
 \dfrac{E_{\text{he}}}{r}   \dfrac{c_0^2}{\Rsun^2}     }\,. \label{unif_inth_ell} 
 \end{equation}
For $\ell > 0$, we define Lamb singularities as
% as the combination of $(r,\omega,\ell)$ the define 
the set of zeros of $\mathsf{F}_\ell$, i.e.,
\begin{equation}\label{Lambsing::def}
(r;\omega,\ell)  \text{ creates a Lamb singularity if}   \hspace*{0.3cm} \mathsf{F}_\ell(r;\omega) =0\,.
\end{equation}
% The above equation defines implicitly the  generalized Lamb 
% frequencies $\mathcal{S}_\ell$ \cref{Lamb_freq::def} corresponding 
% to a given $(r,\omega)$, 
The above equation defines implicitly the  generalized Lamb frequencies 
$(\ell,r) \mapsto \mathcal{S}_\ell(r)$ 
such that $\mathsf{F}_{\ell}(r; \mathcal{S}_\ell) = 0$;
its equivalent explicit expression is given in \cref{Lamb_freq::def}.
% corresponding to a given $(r,\omega)$, 
Similarly,  for a given $(\ell,\omega)$, \cref{Lambsing::def} defines 
the position $r^{\star}_{\ell,\omega}$ introduced in  \cite[Prop 6]{barucq2021outgoing}.

\paragraph{Lamb and buoyancy frequencies}
The \emph{generalized} Lamb and buoyancy frequency, denoted 
respectively as $\mathcal{S}_\ell$ and $\mathcal{N}$,
extend the original definitions of
the Lamb and buoyancy frequency
$S_\ell$ and $N$ to the atmosphere,
\begin{subequations}\label{physfreq::def}
\begin{align}
  S_{\ell}^2 & :=  \ell(\ell+1) \dfrac{c_0^2}{r^2 \Rsun^2} \,,   
  %\hspace*{0.5cm} 
  %\substack{\text{Lamb}\\ \text{ frequency}}, 
  \hspace*{2cm} 
  \mathcal{S}_\ell^2:= S_\ell^2 + \frac{c_0^2}{\Rsun^2} \dfrac{E_{\text{he}}}{r} \label{Lamb_freq::def}\,,\\
N^2 &:=  \dfrac{\phi_0'}{\Rsun^2} \left( \alpha_{\rho_0} - \dfrac{\alpha_{p_0}}{\gamma}\right)\,,
%              \hspace*{0.2cm} \substack{\text{Buoyancy}\\ 
% \text{ frequency}}\,, \hspace*{.8cm} 
\hspace*{1.2cm} 
\mathcal{N}^2 = N^2 + \dfrac{c_0^2}{\Rsun^2} \left( \Ehe' - \Ehe (\alpha_{\gamma p_0} - \dfrac{\alpha_{p_0}}{\gamma}) \right). 
              \label{Buoy_freq::def}
\end{align}
\end{subequations}
We have, cf. \cite[Section 3.1]{barucq:hal-03406855},
  \begin{equation}
    k_{\mathcal{N}}^2=\dfrac{\phi_0''}{c_0^2} -  
                            \left( \dfrac{\alphap}{\gamma} \right)'
                        +\dfrac{\alphap}{\gamma}
                         \left(\alphagamp - \dfrac{\alphap}{ \gamma}\right), \hspace*{0.5cm}
\text{with} \hspace*{0.3cm}
k_{\mathcal{N}}^2 := \Rsun^2\dfrac{\mathcal{N}^2}{c_0^2}.
\end{equation}

%
%
%
%\paragraph{Change of variable } in
%particular $\mathfrak{I}_{\ell}$ appears in the 
%change of variable between the original modal ODE 
%and its Schr\"odinger form.
%% -------------------------------------------------------
%\begin{subequations}
%\begin{align}
%\mathsf{F}_{\ell} & = k_0^2  r^2- r\Ehe  -  \ell(\ell+1)\, , \hspace*{0.5cm} \mathsf{F}_0= k_0^2  r^2 -  r\Ehe, \hspace*{0.3cm} \mathfrak{F}(r): = c_0^2(r) \rho_0(r) \mathsf{F}_0(r) \label{sfFell::def}  \\[1em]
% \mathfrak{I}_{\ell} & =   \dfrac{1}{r \,c_0\sqrt{\rho_0}} 
% \dfrac{\sqrt{\mathsf{F}_{\ell}}}{\sqrt{ \mathsf{F}_0}} = 
% \dfrac{1}{r\, c_0 \sqrt{\rho_0}} 
%\sqrt{\sigma^2  - S_\ell^2  - \tfrac{E_{\text{he}}}{r}  \tfrac{c_0^2}{\Rsun^2} }\Big/ \sqrt{\sigma^2  -
% \tfrac{E_{\text{he}}}{r}   \tfrac{c_0^2}{\Rsun^2}     }\,, \label{unif_inth_ell} 
% \end{align}
% \end{subequations}

\paragraph{Coefficients of Schr\"odinger potentials}
We introduce,
\begin{equation}\label{eq:auxiliary-functions-main}
\mathfrak{t}_{\ell}  = \dfrac{\mathsf{F}_0'}{\mathsf{F}_{\ell}\mathsf{F}_0}\,, \hspace*{0.6cm}
\mathfrak{w}_{\ell}    = \dfrac{r^2 (k_0^2 - k_{\mathcal{N}}^2) }{\mathsf{F}_0}
         +   \dfrac{2\mathsf{F}_0}{\ell(\ell+1)\mathsf{F}_{\ell}}
          +   \left( 1
         -  \dfrac{r\upeta}{2} \right) \left(
        r \mathfrak{t}_{\ell} - \dfrac{2}{\mathsf{F}_\ell}\right)\,
           -  \dfrac{r^2\mathfrak{t}'_{\ell}}{2}.
\end{equation}
The following quantities are used to define the modal ODE in Schr\"odinger form, 
\begin{subequations}\label{eq:auxiliary-sch-main}
\begin{align}
\omega^2_c &\,= \, \dfrac{c_0^2}{\Rsun^2}\left( \dfrac{\alpha_{\gamma p_0}^2}{4}  - \dfrac{\alpha'_{\gamma\,p_0}}{2} 
              +  \dfrac{\phi_0''}{c_0^2} \, - \, \dfrac{\upeta}{r}\, + \, \dfrac{2}{r^2} \right)\,, \hspace*{2cm} \substack{\text{local cut-off}\\\text{ angular frequency} }
                 \label{eq:cut-off-frequency} \\
k^2_{\mathrm{h}} &= \dfrac{\ell(\ell+1)}{r^2} \left(
                 1 + \dfrac{  r\Ehe - r^2 k_{\mathcal{N}}^2  }{ \mathsf{F}_0 }
                 - \dfrac{r^2\upeta \mathfrak{t}_{\ell}}{2} 
                 + r\mathfrak{t}_{\ell}
                 - \dfrac{r^2  \mathfrak{t}_{\ell}'}{2} 
                 +  \dfrac{\ell(\ell+1)(r \mathfrak{t}_{\ell})^2}{4} \right)\,, \hspace*{0.2cm} \substack{\text{horizontal}\\
                 \text{ wavenumber}} \label{eq:kh2} \\
\ksf^2                & \,=\, - \dfrac{\alpha^2_{\gamma\,p_0}}{4} + \dfrac{\alpha'_{\gamma\,p_0}}{2} +  k_0^2\,,    \hspace*{1.5cm}  \upeta \,=\, 2 \dfrac{\alphap}{\gamma} \,-\,\alphagamp\,,\\ 
\upnu_0^2&= \dfrac{9}{4}\,, \hspace*{1.cm} \upnu_{\ell}^2  = \dfrac{1}{4} +
                       \ell(\ell+1)\, \mathfrak{w}_{\ell} \,-\,  
                       r\upeta  \bigg(\dfrac{\mathsf{F}_0}{\mathsf{F}_{\ell}} - 1 \bigg)
                   +   \bigg(\dfrac{\ell(\ell+1)  r\mathfrak{t}_{\ell}}{2} \bigg)^2 \, , \hspace*{0.2cm} \ell>0     \,. \label{eq::nuell}  
%\mu_\ell^2 & \,=\, \dfrac{1}{4} \,+\, 2 \,+\, \ell(\ell+1) \,+\, 
%                   \dfrac{\ell(\ell+1)}{\kz^2}\dfrac{\alpharho}{\gamma}\bigg(\dfrac{\alpharho}{\gamma} \,-\,
%                   \alpharho\bigg) \, , 
\end{align}
\end{subequations}

\section{Computations of the Green's tensors in spherical symmetry}
\label{section:modeling}
% =====================================================================

The main goals of this section are to obtain \cref{compoG_reg:prop} 
and construct \cref{algorithm:assemble} which compute efficiently the directional kernels of 
the Green's tensor $\mathbb{G}^+$ in vector spherical 
harmonics (VSH) tensor basis, cf. \cref{formalexpanGTv0} for their definition.
%\begin{equation}\label{quantocompute_rep}
%G^{\mathrm{PP}}_\ell, \quad G^{\mathrm{PB}}_\ell, \quad
%G^{\mathrm{BP}}_\ell, \quad G^{\mathrm{BB}}_\ell\,.
%\end{equation}
For $\mathbf{f} \in \mathcal{C}^{\infty}_c(\mathbb{R}^3)^3$, the Green's tensor $\mathbb{G}^+$ 
gives the unique outgoing-at-infinity solution $\xib \in H^1(\mathbb{R}^3)^3$ to equation \cref{main_eqn_Galbrun}; 
specifically,
\begin{equation} \label{main_eqn_Galbrun_rep}
\xib = \Rsun^2 \left< \mathbb{G}^+ , \mathbf{f}\right> \hspace*{0.5cm} \text{satisfies} \hspace*{0.5cm} \boldsymbol{\mathcal{L}}\,\xib
 =  \mathbf{f}  \hspace*{0.2cm} \text{ in } \hspace*{0.2cm} \mathbb{R}^3.
 \end{equation} 
We refer to $\mathbb{G}^+$ as the outgoing-at-infinity 
kernel of $\boldsymbol{\mathcal{L}}$. 
In previous work \cite{barucq2021outgoing}, we have solved \cref{main_eqn_Galbrun_rep} 
in VSH basis, and obtained for each harmonic degree $\ell$ a modal ODE in unknown the 
coefficient of the radial displacement, cf. \cref{origodesys}. Existence and characterization 
of its outgoing kernel is given in \cite{barucq2021outgoing}, as well as a 
representation,  called the `gluing' formula \cref{G+_phitdphi:def}, in terms of 
two homogeneous (no source) solutions, one regular-at-zero and the other
outgoing-at-infinity. These provide the main ingredients, recalled
in \cref{resultJDE::subsec,OGK::subsec}, for the results obtained 
in \cref{coeff3Dgreens::subsec,subsection:assembling-algorithm}.

%
%\begin{remark} \textcolor{red}{to be put somewhere?}
%Our approach to compute the modal Green's kernel is comparable 
%      to \cite{bhattacharya2020general,mandal2017finite}, in the sense that 
%      it is computed directly and not via a spectral expansion. 
%      However, we carry out our numerical investigation using a 
%      second-order ODE (derived from \cref{main_eqn_Galbrun_unscaled})
%      which only involves the displacement $\xib$, instead of a first-order ODE system.
%     \end{remark}

% -----------------------------------------------------------
\subsection{Modal Ordinary Differential Equations (ODEs)}\label{resultJDE::subsec}
% -----------------------------------------------------------

\paragraph{Vector spherical harmonics} We follow the definition and notation convention employed in
\cite{monk2003finite,colton2013inverse} for the scalar 
and vector spherical harmonics (VSH). The latter are constructed from the scalar spherical harmonics $\mathrm{Y}^m_\ell$ together with tangential gradient $  \nabla_{\mathbb{S}^2} :=     \mathbf{e}_{\theta}\partial_{\theta} 
                             +\tfrac{1}{\sin\theta}\mathbf{e}_{\phi}  \partial_{\phi}  $. They form an orthonormal basis for $L^2(\RR^3)^3$ and are given by, cf. \cite[Equation (9.56)]{monk2003finite}
or \cite[Definition 3.336]{martin2006multiple},
\begin{equation}\label{vec_sph_har::def}
\begin{aligned}
& \mathbf{P}^{m}_{\ell} (\widehat{\mathbf{x}}) = \mathrm{Y}^{m}_{\ell}(\widehat{\mathbf{x}})\, \mathbf{e}_r ,
   \quad \ell=0,1,\ldots; \\[0.3em]
& \mathbf{B}^{m}_{\ell}  (\widehat{\mathbf{x}})\, = \,\dfrac{\nabla_{\mathbb{S}^2}  \mathrm{Y}^{m}_{\ell} }{\sqrt{\ell(\ell+1)}}\,, \qquad
 \mathbf{C}^{m}_{\ell} (\widehat{\mathbf{x}}) \, = 
  - \dfrac{\mathbf{e}_r \times \nabla_{\mathbb{S}^2}  \mathrm{Y}^{m}_{\ell}}{\sqrt{\ell(\ell+1)}}\,   , \quad \ell = 1,2,\dots \,.
\end{aligned}
\end{equation}
We introduce the summation notation,
\begin{equation}
\sum_{(\ell,m)}:= \sum_{\ell=0}^{\infty} \sum_{m=-\ell}^{\ell}\,, \hspace*{1cm}  \sideset{}{'}\sum_{(\ell,m)} := \sum_{\ell=1}^{\infty} \sum_{m=-\ell}^{\ell}\,.
\end{equation}

% -----------------------------------------------------------
\paragraph{Modal ODEs}
% -----------------------------------------------------------
We decompose the unknown $\xib$ and source vector $\mathbf{f}$ 
of \cref{main_eqn_Galbrun_rep} in VSH basis as
\begin{equation}\label{coeffVHS_f_xi}
\begin{aligned}
& \hspace*{2cm} \mathbf{f} =    \mathbf{f}_{\perp}   + \mathbf{f}_{\mathrm{h}}  + \mathbf{f}_{\times}
\,, \hspace*{1cm}\text{and} \hspace*{0.7cm} \xib =    \xib_{\perp}   + \xib_{\mathrm{h}}  + \xib_{\times} \,,\\
\text{with} \hspace*{0.3cm} \mathbf{f}_{\perp} &= \sum_{(\ell,m)} f^{m}_{\ell}(r) \, \mathrmb{P}^{m}_{\ell}(\widehat{\mathbf{x}}) 
\,, \hspace*{0.5cm}  \mathbf{f}_{\mathrm{h}} = \sideset{}{'}\sum_{(\ell,m)} g^{m}_{\ell}(r)  \, \mathrmb{B}^{m}_{\ell}(\widehat{\mathbf{x}})\,, \hspace*{0.5cm} \mathbf{f}_{\times}= \sideset{}{'}\sum_{(\ell,m)}  h^{m}_{\ell}(r) \, \mathrmb{C}^{m}_{\ell}(\widehat{\mathbf{x}})\,,\\
\xib_{\perp} &= \sum_{(\ell,m)} a^{m}_{\ell}(r) \, \mathrmb{P}^{m}_{\ell}(\widehat{\mathbf{x}}) 
\,, \hspace*{0.5cm}  \xib_{\mathrm{h}} = \sideset{}{'}\sum_{(\ell,m)} b^{m}_{\ell}(r)  \, \mathrmb{B}^{m}_{\ell}(\widehat{\mathbf{x}})\,, \hspace*{0.5cm} \xib_{\times}= \sideset{}{'}\sum_{(\ell,m)}  c^{m}_{\ell}(r) \, \mathrmb{C}^{m}_{\ell}(\widehat{\mathbf{x}})\,.
\end{aligned}
\end{equation}
Their coefficients in VSH basis satisfy at each level $(\ell,m)$, \cite{barucq2021outgoing},
\begin{subequations} \label{origodesys}
\begin{empheq}[left={ \empheqlbrace\,\,}]{align} 
& \left(\hat{q}_1 \, (r\partial_r)^2 \, +\,  \hat{q}_2\,  r\partial_r \, + \, \hat{q}_3 \right) a^m_\ell \,= \,  \mathfrak{f}^m_{\ell}\,,    \label{origmodalODE}\\ 
& \dfrac{b^m_{\ell}}{\sqrt{\ell(\ell+1)}} 
  \,=\,\textcolor{black}{-}\dfrac{r}{\mathsf{F}_{\ell}}\partial_r a^m_{\ell} \,\textcolor{black}{-}\,\left(\dfrac{2}{r} 
  \,\,-\,\, \dfrac{\alpha_{p_0}}{ \gamma} \right)\dfrac{r}{\mathsf{F}_{\ell}} a^m_{\ell}
  \,\, \textcolor{black}{-} \, \, \dfrac{r^2}{\mathsf{F}_{\ell}}\dfrac{\Rsun^2 }{\gamma\, p_0} \dfrac{g^m_{\ell}}   
 {\sqrt{\ell(\ell+1)}} \,, \label{intb_green}\\[0.5em]
& c^m_{\ell}  \,= \,
 -\Rsun^2\dfrac{r^2}{\rho_0\,  c_0^2\, \mathsf{F}_0}\,  h^{m}_{\ell}\,.\label{intc_green}
\end{empheq}
\end{subequations}
In the above equations, $\mathsf{F}_0$, $\mathsf{F}_\ell$ are defined in \cref{sfFell::def}, and $\mathfrak{f}^m_{\ell}$ is 
defined in  terms of $(f^m_\ell, g^m_\ell)$ by
%\begin{equation}\label{origmodalODErhs}
%\begin{aligned}
%\mathfrak{f}^m_{\ell} \, & :=\, \Rsun^2\,  \ell(\ell+1)\dfrac{r\tfrac{\alpha_{p_0}}{ \gamma} 
%               -r\alpha_{\gamma p_0} -1}{\mathsf{F}_{\ell}}\, \dfrac{g^m_{\ell} }{\gamma\, p_0\, \sqrt{\ell(\ell+1)}}\\
%               %%
%& \hspace*{3cm}\,+\,  \Rsun^2\, \dfrac{\ell(\ell+1)}{r}\partial_r\, \left(\dfrac{r^2}{\mathsf{F}_{\ell}} \,
%\dfrac{ 1}{\gamma\, p_0}\dfrac{g^m_{\ell}}{ \sqrt{\ell(\ell+1)}}\right)
%\,+\, \Rsun^2\,\dfrac{f^m_{\ell}}{\gamma \, p_0} \, .
%\end{aligned}
%\end{equation} 
\begin{equation}\label{origmodalODErhs}
\dfrac{\mathfrak{f}^m_{\ell}}{\Rsun^2} := \dfrac{\sqrt{ \ell(\ell+1)} (r\tfrac{\alpha_{p_0}}{ \gamma} 
               -r\alpha_{\gamma p_0} -1)g^m_{\ell} }{\mathsf{F}_{\ell} \gamma\, p_0\,}
 +  \dfrac{\sqrt{\ell(\ell+1)}}{r}\partial_r\, \left(\dfrac{r^2g^m_{\ell}}{\mathsf{F}_{\ell}\gamma\, p_0} \right)
+\,\dfrac{f^m_{\ell}}{\gamma \, p_0} \, .
\end{equation}
Coefficients 
$\hat{q}_1$, $\hat{q}_2$, $\hat{q}_3$ are defined with the quantities 
listed in \cref{eq:auxiliary-functions-main} as follows: for $\ell > 0$,
%\begin{subequations}\label{coeffq::def}
%\begin{align}
%\hat{q}_1 \,&=  -\dfrac{\mathsf{F}_0}{\mathsf{F}_{\ell}}\,;\hspace*{2cm}
%%%------------------------------------------------
%\hat{q}_2 \, = \,\left( \alpha_{\gamma p_0} - \dfrac{2}{r} \right) \dfrac{\mathsf{F}_0}{\mathsf{F}_{\ell}}
%\, + \, \ell(\ell+1) \dfrac{\mathsf{F}_0'}{(\mathsf{F}_{\ell})^2}  \,;\\[0.8em]
%%%--------------------------
%\hat{q}_3\, &= \,  \left(-  k_0^2 + \dfrac{\phi_0''}{c_0^2} + \dfrac{2}{r^2}
%+ \dfrac{ 2 (\alpha_{\gamma p_0} - \tfrac{\alpha_{p_0}}{\gamma})}{r}\right)\dfrac{\mathsf{F}_0}{\mathsf{F}_{\ell}} \, +\,  \ell(\ell+1) \left( \dfrac{2}{r} - \dfrac{\alpha_{p_0}}{\gamma}\right)  \dfrac{\mathsf{F}_0'}{(\mathsf{F}_{\ell})^2}  \,\nonumber\\[0.3em]
%%
%& \hspace*{2.5cm} +  \dfrac{\ell(\ell+1)}{\mathsf{F}_{\ell}}
%\left(k_0^2\, -\, \dfrac{\phi_0''}{c_0^2}\,+\, \left(\dfrac{\alpha_{p_0}}{\gamma}\right)'
%+\, \,\dfrac{\alpha_{p_0}}{\gamma}\,\left(-\alpha_{\gamma p_0}\, +\,\dfrac{\alpha_{p_0}}{ \gamma}  \right)\right)\,. \label{tqell::def}
%\end{align}
%\end{subequations}
\begin{subequations}\label{coeffq::def}
\begin{align}
 \hspace*{0.5cm} \hat{q}_1 &= - \dfrac{\mathsf{F}_0}{r^2\mathsf{F}_{\ell}}\,;\hspace*{1cm}
%%------------------------------------------------
\hat{q}_2  = \left(r \alpha_{\gamma p_0} - 1 \right) \dfrac{\mathsf{F}_0}{r^2\mathsf{F}_{\ell}}
\, + \, \ell(\ell+1) \dfrac{\mathsf{F}_0'}{r(\mathsf{F}_{\ell})^2}  \,;\\[0.8em]
%%--------------------------
 \hat{q}_3 &=   \left[-  k_0^2 + \dfrac{\phi_0''}{c_0^2} + \dfrac{2}{r^2}
+ 2 \left(\dfrac{\alpha_{\gamma p_0}}{r} - \dfrac{\alpha_{p_0}}{\gamma r}\right)\right]\dfrac{\mathsf{F}_0}{\mathsf{F}_{\ell}}  +  \left( \dfrac{2}{r} - \dfrac{\alpha_{p_0}}{\gamma}\right)  \dfrac{\ell(\ell+1) \mathsf{F}_0'}{(\mathsf{F}_{\ell})^2} 
 +  \dfrac{\ell(\ell+1)(k_0^2 -k_{\mathcal{N}}^2)}{\mathsf{F}_{\ell}}\,. \label{tqell::def}
\end{align}
\end{subequations}
\begin{equation}
\text{For } \ell=0 :  \hspace*{0.5cm} \hat{q}_1 = -\dfrac{1}{r^2}\,,\hspace*{0.3cm}  \hat{q}_2 = \dfrac{\alpha_{\gamma p_0}}{r} - \dfrac{1}{r^2}, \hspace*{0.3cm}
 r^2\hat{q}_3  = -  k_0^2 + \dfrac{\phi_0''}{c_0^2} + \dfrac{2}{r^2}
+ 2 \left(\dfrac{\alpha_{\gamma p_0}}{r} - \dfrac{\alpha_{p_0}}{\gamma r}\right) \,.
\end{equation}

\paragraph{Modal Schr\"odinger ODE}
% -----------------------------------------------------------
Equation \cref{origmodalODE} is written in Schr\"odinger form as follows,
\begin{equation}\label{conjmodalODErhs}
\text{with} \hspace*{0.1cm} \tilde{a}^m_{\ell} = \dfrac{a^m_\ell}{\mathfrak{I}_{\ell}}: 
 \hspace*{0.7cm} \cref{origmodalODE} \hspace*{0.2cm}\Leftrightarrow  \hspace*{0.2cm} \left( - \partial_r^2  + V_{\ell}(r)\right) \tilde{a}^m_{\ell}
  =  - \dfrac{\mathsf{F}_\ell}{\mathfrak{I}_\ell\, \mathsf{F}_0} \mathfrak{f}^m_\ell\,.
\end{equation}
Here $\mathfrak{I}_{\ell}$ is given in \cref{unif_inth_ell}. 
The potential $V_\ell$ is defined with auxiliary quantities in \cref{eq:auxiliary-sch-main} as,
\begin{equation}\label{Vell::def}
V_0(r):=   - \dfrac{\big(\sigma^2(r) - \omega^2_c(r)\big)\Rsun^2}{c_0(r)^2}\,, \hspace*{0.5cm} \text{and } \hspace*{0.5cm}
V_\ell(r) = V_0(r) +   k^2_{\mathrm{h}}(r)\,,\hspace*{0.2cm} \text{ for } \ell>0.
\end{equation}
It also has an equivalent expression useful for analysis in scattering theory, cf. \cite{barucq2021outgoing},
\begin{equation} V_{\ell}(r) =   
  -\ksf(r)^2 + \dfrac{\phi_0''(r)}{c_0(r)^2}- \dfrac{\upeta(r)}{r}
  + \dfrac{\upnu_\ell(r)^2 - \frac{1}{4}}{r^2} . \end{equation}

\begin{remark}\label{cov_Dal::rmk} 
The change of unknown \cref{conjmodalODErhs} also 
appears in \cite[Equations (7.5) and (7.6)]{christensen2014lecture} 
and is also used to obtain a Schr\"odinger equation in 
variable $\xi_r$. %% i.e. from equation (7.3) in (7.7) \todoflo{in?} of \cite{christensen2014lecture}. 
The term $\mathfrak{I}_\ell$ is the square of 
the expression given in \cite[Equation (7.5)]{christensen2014lecture}.
\hfill $ \diamond$
\end{remark}

We write the original modal operators in \cref{origmodalODE} 
and conjugate modal operator in \cref{conjmodalODErhs} as
\begin{equation}\label{notation_modalop}
\Lconjug  \,:=\,  - \partial_r^2  + V_{\ell}(r)
\,, \qquad 
\Lorigin  \,:=\,  \hat{q}_1 \, (r \partial_r)^2 \,+\,  \hat{q}_2\,  r\partial_r \, + \, \hat{q}_3\,.
\end{equation}
Theoretical construction of outgoing-at-infinity solutions 
was obtained in \cite{barucq2021outgoing} with the Schr\"odinger form $\Lconjug$, cf. \cref{OGK::subsec}, and naturally, the construction of ABC will also be carried out with this operator, cf. \cref{section:rbc}.
We will investigate the numerical robustness of both 
forms of the equation in \cref{section:numerics:original-and-conjugated}.

\begin{remark}\label{oldcoeffs::rmk}
In \cite{barucq2021outgoing} (e.g. Proposition 2) and \cite[Section 3.1]{barucq:hal-03406855}, $\Lorigin$ is written in an equivalent form
$\Lorigin = \tilde{q}_1 \partial_r^2 + \tilde{q}_1 \partial_r + \tilde{q}_3$ with coefficients (with a change of notation),
\begin{equation}\label{oldcoeffs}
\begin{aligned}
\ell > 0 : \hspace*{0.3cm} & \hspace*{0.3cm} \tilde{q}_1 \, = \, -\dfrac{\mathsf{F}_0}{\mathsf{F}_\ell} \,, \hspace*{0.5cm} \tilde{q}_2 \, = \,\left( \alpha_{\gamma p_0} - \dfrac{2}{r} \right) \dfrac{\mathsf{F}_0}{\mathsf{F}_{\ell}}
+  \ell(\ell+1) \dfrac{\mathsf{F}_0'}{(\mathsf{F}_{\ell})^2};\\
  \ell =0 :\hspace*{0.3cm} & \hspace*{0.3cm} \tilde{q}_1 = -1 \,, \hspace*{0.5cm} \tilde{q}_2 = \alpha_{\gamma p_0} - \dfrac{2}{r}.
  \end{aligned}
\end{equation}
\end{remark}

% -----------------------------------------------------------
\paragraph{Formal decomposition of Green's kernel}
% -----------------------------------------------------------
The Green's kernel $\mathbb{G}^+$ in \cref{main_eqn_Galbrun_rep}
is decomposed into parts which give the radial, 
tangential-$\mathbf{B}$, and tangential-$\mathbf{C}$ components
of the solution,
\begin{equation}
\mathbb{G}^+ = \mathbb{G}^+_{\perp}  \,  + \, \mathbb{G}^+_{\mathrm{h}} \, + \, \mathbb{G}^+_{\times}.
\end{equation}
Result \cref{origodesys}, which states that coefficients $b^m_\ell$ are determined
by $a^m_\ell$ and  $(f^m_\ell,g^m_\ell)$, 
and that $c^m_\ell$ is decoupled from $a^m_\ell$ and $b^m_\ell$, can be phrased as,
\begin{equation}\label{formalexpanGTv00}
\xib^{\perp} =  \Rsun^2 \left< \mathbb{G}^+_{\perp} \, , \, \mathbf{f}^{\perp} + \mathbf{f}^\mathrm{h}\right>\,, \hspace*{0.3cm}   \xib^{\perp} =  \Rsun^2 \left< \mathbb{G}^+_{\mathrm{h}}\, ,\, \mathbf{f}^{\perp} + \mathbf{f}^\mathrm{h}\right> \,, 
 \hspace*{0.3cm} \xib^{\times} =  -   \Rsun^2\dfrac{\lvert \mathbf{x}\rvert^2}{\textcolor{black}{ 
   \mathfrak{F}(\lvert \mathbf{x}\rvert)}} \mathbf{f}^{\times}(\mathbf{x}) \,.
\end{equation}
Here, $ \left< \cdot , \, \cdot \right>$ is the distributional pairing on $(0,\infty)$.
We note that $\mathbb{G}^+_{\times}$ only has component  which is a Dirac distribution, along direction $\mathbf{C}^m_\ell$.
The rest of the work will only concern $\mathbb{G}^+_{\perp}$ and $\mathbb{G}^+_{\mathrm{h}}$. Their coefficients in formal VSH decomposition, denoted\footnote{These subscripts have the connotation of (row,column) subscripts of a matrix, with `P' and `B' being the indices of the basis elements. The latter convention follows from the notation of the basis vectors listed in \cref{vec_sph_har::def}. The matrix interpretation is reflected in \cref{abGPPB}}. by $G^{\mathrm{PP}}_\ell$, $ G^{\mathrm{PB}}_\ell$, $
G^{\mathrm{BP}}_\ell$, $ G^{\mathrm{BB}}_\ell$ 
are computed in \cref{coeff3Dgreens::subsec},
\begin{equation}\label{formalexpanGTv0}
\begin{aligned}
\mathbb{G}^+_{\perp}(\mathbf{x},\mathbf{y})
&= \sum_{(\ell,m)} 
  \textcolor{black}{G^{\mathrm{PP}}_{\ell}(\lvert\mathbf{x}\rvert,\lvert\mathbf{y}\rvert)} \,
\mathrmb{P}^m_{\ell}(\widehat{\mathbf{x}})  \otimes\overline{\mathrmb{P}^m_{\ell}(\widehat{\mathbf{y}})} 
\,+\, \sideset{}{'}\sum_{(\ell,m)} 
\textcolor{black}{G^{\mathrm{BP}}_{\ell}(\lvert\mathbf{x}\rvert,\lvert\mathbf{y}\rvert)} \,\, 
\mathrmb{B}^m_{\ell}(\widehat{\mathbf{x}})  \otimes\overline{\mathrmb{P}^m_{\ell}(\widehat{\mathbf{y}})}\,;\\ 
%%----------------------------------------
\mathbb{G}^+_{\mathrm{h}}(\mathbf{x},\mathbf{y})&= 
 \sideset{}{'}\sum_{(\ell,m)} 
\textcolor{black}{G^{\mathrm{PB}}_{\ell}(\lvert\mathbf{x}\rvert,\lvert\mathbf{y}\rvert)} \, \,
\mathrmb{P}^m_{\ell}(\widehat{\mathbf{x}})  \otimes\overline{\mathrmb{B}^m_{\ell}(\widehat{\mathbf{y}})}
 \, + \, \sideset{}{'}\sum_{(\ell,m)} 
\textcolor{black}{G^{\mathrm{BB}}_{\ell}(\lvert\mathbf{x}\rvert,\lvert\mathbf{y}\rvert)} \, 
\mathrmb{B}^m_{\ell}(\widehat{\mathbf{x}})  \otimes\overline{\mathrmb{B}^m_{\ell}(\widehat{\mathbf{y}})} 
\,. 
 \end{aligned}
 \end{equation}
 
%Here and in the rest of the work, we denote by 
%$\mathcal{D}_+ := \mathcal{C}^{\infty}_c([0,\infty))$ 
%the space of smooth functions having compact support on 
%$[0,\infty)$, and $\mathcal{D}'_+$ its dual. 

% -----------------------------------------------------------
\subsection{Outgoing modal Green's kernels and assembling formula}\label{OGK::subsec}
% -----------------------------------------------------------
%From the analysis of \cite{barucq2021outgoing}, we have that 
%the regularity at zero is prescribed by indicial exponent
%\begin{equation}
%\lambda_0^+= 1 \,,  \hspace*{1cm} \lambda_\ell^+ = \ell-1 \,, \hspace*{0.3cm} \text{for } \ell > 0\,. 
%\end{equation}
%while oscillatory behavior at infinity is characterized by $\mathsf{k}_{\mathrm{a}}^2$ the energy level of $-V_\ell$ at infinity
%\begin{equation} \mathsf{k}_{\mathrm{a}}^2 = \lim_{r\rightarrow \infty} -V_\ell\,,\end{equation}
%and a phase $\psi$  which is  a solution of the eikonal equation with right-hand-side obtained by ignoring the `short-range' term $\tfrac{\upnu_\ell(r)^2 - \frac{1}{4}}{r^2} $ in $V_\ell$, (cf. \cite[Eq. (??)]{barucq2021outgoing}), 
%\begin{equation}\label{phasefcn::def}
% \rvert\nabla \psi\rvert^2 = -\ksf^2 + \dfrac{\phi_0''}{c_0^2}- \dfrac{\upeta}{r+1}\,,
% \hspace*{0.2cm} 
%    \text{satisfying } \hspace*{0.2cm} \psi(r,\mathsf{k}_{\mathrm{a}}) 
%  = r \mathsf{k}_{\mathrm{a}} ( 1 + \mathsf{o}(1)),  \hspace*{0.2cm}\text{as } \hspace*{0.2cm}r\rightarrow \infty\,.
%\end{equation}

From the analysis of \cite{barucq2021outgoing}, we have that 
the regularity at zero is prescribed by indicial exponent
\begin{equation}
\lambda_0^+= 1 \,, \quad \text{for $\ell=0$}\,, 
\hspace*{2cm} \lambda_\ell^+ = \ell-1 \,,\quad  \text{for } \ell > 0,
\end{equation}
while oscillatory behavior at infinity is characterized by $\mathsf{k}_{\mathrm{a}}^2$, 
the energy level of $-V_\ell$ at infinity, and a phase $\psi$, with,
\begin{equation}\label{phasefcn::def} \mathsf{k}_{\mathrm{a}}^2 = \lim_{r\rightarrow \infty} -V_\ell\,, \hspace*{0.5cm}
\text{and} \hspace*{0.3cm} \psi(r,\mathsf{k}_{\mathrm{a}}) 
  = r \mathsf{k}_{\mathrm{a}} ( 1 + \mathsf{o}(1)),  \hspace*{0.2cm}\text{as } \hspace*{0.2cm}r\rightarrow \infty\,.\end{equation}
The phase $\psi$ is a solution to 
the eikonal equation whose right-hand side is obtained 
from $V_\ell$ by ignoring the `short-range' term 
$\tfrac{\upnu_\ell(r)^2 - \frac{1}{4}}{r^2} $, cf. \cite{barucq2021outgoing}.

We restate the result of \cite[Section 7]{barucq2021outgoing} in terms of the Green's kernel of $\Lorigin$. 
Under appropriate assumption\footnote{These include  
positive attenuation and away from scattering threshold frequency, and non-vanishing Wronskian, 
cf. \cite[Equations (7.7) and (7.9)]{barucq2021outgoing}.}, 
the coefficients $a^m_{\ell}$ of $\xib \in L^2(\mathbb{R}^3)^3$ are given by,
\begin{equation} \label{inta_green}
 a^m_{\ell}(r) \,= \, \int_0^{\infty} G^+_{\ell}(r,s) \, \mathfrak{f}^m_{\ell}(s)\, ds\,,
\end{equation}
where $G^+_\ell$  called the outgoing Green's kernel for $\Lorigin$ is the unique regular-at-0 and outgoing-at-infinity distributional solution to, 
%\begin{equation} \label{eq:Gplus-main}
%\mathfrak{L}_{\ell} \, G^+_{\ell}
% \, =  \, \delta(r-s)\,.
%\end{equation}
\begin{subequations}\label{eq:Gplus-main}
\begin{align}
 & \hspace*{1cm}\Lorigin\, G^+_{\ell}\, = \,\delta(r-s) \,, \quad r\in (0, \infty)\\
 \text{satisfying} &\hspace*{0.5cm} G^+_{\ell}(r,s) \dfrac{\mathfrak{I}_{\ell}(s)}{\mathfrak{I}_{\ell}(r)}   \, \dfrac{\mathsf{F}_0(s)}{\mathsf{F}_{\ell}(s)} \, = \, 
 e^{\mathrm{i} \, \psi(r,\mathsf{k}_a)} \, \big(1 \, + \, \mathsf{o}(1)\big) 
  \, ,  \quad \text{as }   r\rightarrow \infty\,, \label{outinfcond::def}\\[0.5em]
\text{and} & \hspace*{0.5cm}  G^+_{\ell}(r) \, = \, r^{\lambda^+_{\ell}} \, \big(1 \, + \, \mathsf{o}(1) \big)\,, \quad \text{as } r\rightarrow 0\,.        \label{regzerocond::def} 
\end{align} \end{subequations}

We will need the following characterization of $G^+_\ell$ for our main results, 
cf. \cref{subsection:assembling-algorithm}. It is given in terms of two solutions,  
$\phi_{\ell}$ and $\phi^+_{\ell}$ which satisfy, respectively,
\begin{subequations} \label{eq:phi_phi_problem}
\begin{align}
\Lorigin\, \phi_{\ell} =  0 \,, \hspace*{0.2cm} \text{on }  (0, \infty)\,, 
\hspace*{0.3cm}&\text{satisfying}  \hspace*{0.3cm} \phi_{\ell}(r)  = r^{\lambda^+_{\ell}}  (1  +  \mathsf{o}(1) )\,, \hspace*{0.2cm} \text{as } r\rightarrow 0\,,  \label{regsol_origODE}\\
\Lorigin\, \phi^+_{\ell} =  0\,, \hspace*{0.2cm} \text{on }(0, \infty)\,, \hspace*{0.3cm}& \text{satisfying}  \hspace*{0.3cm} \mathfrak{I}_\ell(r)\, \phi^+_{\ell}(r)  =   e^{\mathrm{i} \, \psi(r,\mathsf{k}_a)} \left(1  +  \mathsf{o}(1)\right) 
\,, \hspace*{0.2cm} \text{as }   r\rightarrow \infty\,.\label{outgoingsol_origODE}
\end{align}
\end{subequations}
We also refer to \cite[Equation (3.90)]{barucq:hal-03406855} for more detailed discussion of these solutions. 
In terms of $\phi_\ell$ and $\phi^+_\ell$, 
the outgoing kernel $G^+_\ell$ is given by
\begin{equation}\label{G+_phitdphi:def}
G^+_{\ell}(r,s)\,  = \, 
\dfrac{ \mathrm{H}(s-r)\, \phi_{\ell}(r)\, \phi^+_{\ell}(s)
\, + \, \mathrm{H}(r-s)\, \phi_{\ell}(s)\, \phi^+_{\ell}(r)}{ \hat{q}_{\ell}(s)\, 
\mathcal{W}^+_\ell(s)}\,.
\end{equation}
Here $\hat{q}_\ell$ is defined in \cref{oldcoeffs}, $\mathrm{H}$ denotes the Heaviside function and $\mathcal{W}^+$  the Wronskian,
\begin{equation}\label{calW+::def}
\mathcal{W}^+_\ell(s) \,:= \, \mathcal{W}\{\phi_\ell(s),\phi^+_{\ell} (s)\}
\, = \, \phi_\ell(s) \, \phi^{+\prime}_\ell(s)
\, - \, \phi_\ell'(s) \, \phi^+_\ell(s)\,.
\end{equation}

\subsection{Coefficients of Green's tensor in VSH basis}
\label{coeff3Dgreens::subsec}
% ================================================================================
%
%
%\textcolor{red}{to rewrite this intro}
%In this section, we first rewrite the above results 
%in \cref{coeff_orig::prop} to make appear the coefficients 
%of the 3D kernel in VSH basis. 
%These quantities are 
%\begin{equation}\label{quantocompute_rep}
%G^{\mathrm{PP}}_\ell, \quad G^{\mathrm{PB}}_\ell, \quad
%G^{\mathrm{BP}}_\ell, \quad G^{\mathrm{BB}}_\ell\,.
%\end{equation}
%\todoflo{define PP, PB, etc.}
%How they are related to the 3D kernel is formally discussed at the end of the section. 
%Since the system of the ODE \cref{origodesys} of the coefficients of the unknown $\xib$
%can be reduced to a single ODE (\cref{conjmodalODErhs}, or equivalently \cref{origmodalODE}) in terms of $a^m_\ell$ the coefficients in the radial direction, it follows that the components of $\mathbb{G}^+$ will be determined in terms of that in the radial direction ($G^{\mathrm{PP}}_\ell$) and its derivatives.
%This is the purpose of \cref{coeff_orig::prop}.
%In our second result, given by \cref{compoG_reg:prop}, 
%we go further and obtain a `gluing' formula for these components, in which they will be expressed directly in terms of the homogeneous (i.e. without source) solutions $\phi^+_\ell$ and $\phi_\ell$ of $\mathfrak{L}_\ell$.
%This result is the key ingredient for the computation algorithm in \cref{subsection:assembling-algorithm}. % \cref{Numgendis::sec}. 

Starting from representation \cref{G+_phitdphi:def} for the 
outgoing kernel $G^+_\ell$ of $\Lorigin$, we now proceed to 
compute the directional kernels,
$G^{\mathrm{PP}}_\ell$, $ G^{\mathrm{PB}}_\ell$, 
$G^{\mathrm{BP}}_\ell$, $ G^{\mathrm{BB}}_\ell$
introduced in \cref{formalexpanGTv0},
in terms of solutions $\phi^+_\ell$  \cref{outgoingsol_origODE} and $\phi_\ell$  \cref{regsol_origODE}. These results form the statement of \cref{compoG_reg:prop} and require the following auxiliary kernels,
\begin{subequations}\label{primker}
\begin{align}
\mathsf{G}^+_{\ell}(\textcolor{black}{r},s)
\,&:=   \mathrm{H}(s-\textcolor{black}{r}) \, \textcolor{black}{\phi_\ell(r)} \,\phi_\ell^+(s) \,\,+ \,\,  \mathrm{H}(\textcolor{black}{r}-s) \, \phi_\ell(s)\,  \textcolor{black}{\phi_\ell^+(r)}\,;  \label{sfGell::def}  \\[0.5em]
\mathsf{T}^+_{\ell}(\textcolor{black}{r},s) &:=  \mathrm{H}(s-\textcolor{black}{r}) \, \textcolor{black}{r\partial_r \phi_\ell(r)} \, \phi_\ell^+(s)
\,\, + \,\, \mathrm{H}(\textcolor{black}{r}-s) \, \phi_\ell(s)\,  \textcolor{black}{r\partial_r \phi_\ell^+(r)}   \,; \label{sfTell::def}\\[0.5em]
\mathsf{Q}^+_{\ell}(\textcolor{black}{r},s) &:=  \mathrm{H}(s-\textcolor{black}{r}) \, \textcolor{black}{r\partial_r \phi_\ell(r)}\, s \partial_s\phi_\ell^+(s)
\,\, +\,\, \mathrm{H}(\textcolor{black}{r}-s)\,  s\partial_s \phi_\ell(s) \,  \textcolor{black}{r\partial_r \phi_\ell^+(r)}   \,. \label{sfQell::def}
\end{align}
\end{subequations}
In \cref{coeff_orig::prop}, we first establish relations between these directional kernels 
in terms of $G^{\mathrm{PP}}_\ell$ and its derivatives; 
the latter are computed in \cref{derGPP::cor}.
 We will also need the following quantity, defined in terms of the Wronskian $\mathcal{W}^+_{\ell}$ of \cref{calW+::def} % associated to $\{ \phi_{\ell}, \phi^+_{\ell}\}$, 
and $\frak{F}$ of \cref{sfFell::def},
\begin{equation} \label{eq:mathfrakp}
\mathfrak{p}(s)
\, :=\,   \dfrac{\mathsf{F}_{\ell}(s)}{\mathcal{W}^+_\ell(s)\, \mathfrak{F}(s)} \,.
\end{equation}
% Below, we write $ \left< \cdot , \, \cdot \right>$ to denote the distributional pairing on $(0,\infty)$.
% \todoflo{already below (3.15)}

%% ======================================================================
%\subsubsection{Green's kernels in VSH basis}
%% ======================================================================

\begin{proposition}\label{coeff_orig::prop} 
%The coefficient kernels, $G^{\mathrm{PP}}_\ell$, $ G^{\mathrm{PB}}_\ell$, 
%$G^{\mathrm{BP}}_\ell$, $ G^{\mathrm{BB}}_\ell$ introduced
% to \cref{formalexpanGTv00}--\cref{formalexpanGTv0} are related 
%to the coefficients $a^m_{\ell}$ and $b^m_{\ell}$ of radial  component $\xib_{\perp}$ and tangential component $\xib_\mathrm{h}$, cf. \cref{coeffVHS_f_xi}, of solution $\xib$ to \cref{main_eqn_Galbrun_rep}, which are  outgoing solutions to 
% \cref{inta_green} and \cref{intb_green}, as follows.
The directional kernels, $G^{\mathrm{PP}}_\ell$, $ G^{\mathrm{PB}}_\ell$, 
$G^{\mathrm{BP}}_\ell$, $ G^{\mathrm{BB}}_\ell$, introduced
in \cref{formalexpanGTv0}, are related
to the coefficients $a^m_{\ell}$ and $b^m_{\ell}$ of the solution $\xib$
to \cref{main_eqn_Galbrun_rep} as follows,
\begin{subequations}\label{abGPPB}
\begin{align}
a^m_{\ell}
 \,&=\, \Rsun^2 \int_0^{\infty} G^{\mathrm{P}\mathrm{P}}_{\ell}(r,s)  f^m_{\ell}(s)
 \,s^2\,  \mathrm{d} s
\,   + \, \Rsun^2 \int_0^{\infty} G^{\mathrm{P}\mathrm{B}}_{\ell}(r,s)  g^m_{\ell}(s)\, s^2\, \mathrm{d}s\,,\label{aGPPB}\\
b^m_{\ell}
 \,&=\,\Rsun^2 \int_0^{\infty} G^{\mathrm{B}\mathrm{P}}_{\ell}(r,s)  f^m_{\ell}(s)
 \,s^2\,  \mathrm{d} s
\,   + \, \Rsun^2  \left< G^{\mathrm{B}\mathrm{B}}_{\ell}(r,s) \, ,\,    g^m_{\ell}(s)\, s^2\, \right>.\label{bGPPB}
\end{align}
\end{subequations}
Recall that $a^m_\ell$ and $b^m_\ell$ are the coefficients 
of radial component $\xib_{\perp}$ and tangential component $\xib_\mathrm{h}$, 
cf. \cref{coeffVHS_f_xi}; additionally, $a^m_\ell$ are outgoing solutions to 
\cref{inta_green} and $b^m_\ell$ are defined in \cref{intb_green}.
The above kernels further satisfy the following relations to the outgoing kernel $G^+_\ell$ of $\Lorigin$ and the radial-radial components $G^{\mathrm{P}\mathrm{P}}_{\ell}$,
\begin{subequations}
\begin{align}
 \textcolor{black}{G^{\mathrm{P}\mathrm{P}}_{\ell}}(r,s)  &=  \dfrac{ G^+_{\ell}(r,s)}{c^2_0(s) \rho_0(s) s^2}\,; \label{GPP}\\
G^{\mathrm{P}\mathrm{B}}_{\ell}(r,s) 
&= -\dfrac{\sqrt{\ell(\ell+1)}}{\mathsf{F}_{\ell}(s)}
    \left[s\partial_s G^{\mathrm{PP}}_{\ell}+  G^{\mathrm{P}\mathrm{P}}_{\ell} \left( 2  - \dfrac{s\alpha_{p_0}(s)}{ \gamma(s)} \right)    \right] \label{GPB_GPP} ;\\
G^{\mathrm{BP}}_{\ell}(r,s)
   &=  \textcolor{black}{-}\dfrac{ \sqrt{\ell(\ell+1)}}{\mathsf{F}_{\ell}(r)}
\left[ r  \partial_rG^{\mathrm{P}\mathrm{P}}_{\ell}(r,s) 
 +   G^{\mathrm{PP}}_{\ell}(r,s)\left( 2 -r\dfrac{\alpha_{p_0}(r)}{\gamma(r)}\right) \right]\,; \label{GBP_GPP}\\
G^{\mathrm{BB}}_{\ell}(r,s)
&= \dfrac{ \ell(\ell+1)}{\mathsf{F}_{\ell}(r)\,\mathsf{F}_{\ell}(s)}
    \left[r\partial_r s\partial_s G^{\mathrm{PP}}_{\ell}\, + \,\left( 2  - \dfrac{s\alpha_{p_0}(s)}{ \gamma(s)} \right) r\partial_rG^{\mathrm{P}\mathrm{P}}_{\ell}    +   \left(  2 - \dfrac{r\alpha_{p_0}(r)}{\gamma(r)}\right)s\partial_s G^{\mathrm{PP}}_{\ell} \right.\nonumber \\
&    
 \left. \hspace*{0.2cm}   +  G^{\mathrm{P}\mathrm{P}}_{\ell} \left(  2  -\dfrac{r\alpha_{p_0}(r)}{\gamma(r)}\right)\left( 2  - \dfrac{s\alpha_{p_0}(s)}{ \gamma(s)} \right)    \right]   \,\textcolor{black}{-} \,  \dfrac{r^2}{\mathsf{F}_{\ell}(r)\gamma(r) p_0(r)} \dfrac{\delta(r-s)}{s^2}\,.  \label{GBB_GPP}
\end{align}
\end{subequations}

%\begin{subequations}
%\begin{align}
%G^{\mathrm{BB}}_{\ell}(r,s)
%&=\textcolor{black}{-}\dfrac{ \sqrt{\ell(\ell+1)}}{\mathsf{F}_{\ell}(r)}
%\left[ r  \partial_rG^{\mathrm{PB}}_{\ell}(r,s) 
%\, + \,  G^{\mathrm{PB}}_{\ell}(r,s)\left( 2 \,-\, r\dfrac{\alpha_{p_0}(r)}{\gamma(r)}\right) \right] \nonumber \\
%%
% & \hspace*{2cm}\, \textcolor{black}{-} \,  \dfrac{r^2}{\mathsf{F}_{\ell}(r)\, \gamma(r)\, p_0(r)} \dfrac{\delta(r-s)}{s^2}\\[0.6em]
%%%
%&=  \dfrac{ \ell(\ell+1)}{\mathsf{F}_{\ell}(r)\,\mathsf{F}_{\ell}(s)}
%    \left[r\partial_r s\partial_s G^{\mathrm{PP}}_{\ell}\, + \,\left( 2  - s\tfrac{\alpha_{p_0}(s)}{ \gamma(s)} \right) r\partial_rG^{\mathrm{P}\mathrm{P}}_{\ell}     \right.  \nonumber\\
%    & \hspace*{3cm} +  \left. \left(  2 -r \tfrac{\alpha_{p_0}(r)}{\gamma(r)}\right)s\partial_s G^{\mathrm{PP}}_{\ell}\, + \, G^{\mathrm{P}\mathrm{P}}_{\ell} \left(  2 -r \tfrac{\alpha_{p_0}(r)}{\gamma(r)}\right)\left( 2  - s\tfrac{\alpha_{p_0}(s)}{ \gamma(s)} \right)    \right] \nonumber \\
%    %%
%    & \hspace*{1cm} \, \textcolor{black}{-} \,  \dfrac{r^2}{\mathsf{F}_{\ell}(r)\, \gamma(r)\, p_0(r)} \dfrac{\delta(r-s)}{s^2}\,.    \label{GBB_GPP}
%\end{align}
%\end{subequations}

\end{proposition}

\begin{proof}

\textbf{Part 1: Kernel for the radial coefficients}
With $\mathfrak{f}^m_{\ell}$ comprised of three terms 
in its definition \cref{origmodalODErhs}, we rewrite \cref{inta_green}, 
upon ignoring the factor $\Rsun^2$, as
\begin{equation}
a^m_{\ell} \,= \, \int_0^{\infty} G^+_{\ell}(r,s) \, \mathfrak{f}^m_{\ell}(s)\, \mathrm{d}s 
=  I_1  + I_2 + I_3\,, 
\end{equation}
\begin{equation}
\begin{aligned}
\text{where} \hspace*{0.3cm} I_1 &:= \int_0^{\infty} G^+_{\ell}(r,s) \, \dfrac{f^m_{\ell}(s)}{\gamma(s) p_0(s)}\, \mathrm{d}s \,, \\
I_2 &:= \int_0^{\infty} G^+_{\ell}(r,s) \, \left(s\dfrac{\alpha_{p_0}(s)}{ \gamma(s)} 
               -s\alpha_{\gamma p_0}(s) -1 \right)\dfrac{\sqrt{\ell(\ell+1)}\, g^m_{\ell}(s)}{\mathsf{F}_{\ell}(s)\gamma(s) p_0(s)}\, \mathrm{d}s \,, \\
 %              %
  I_3 &:= \int_0^{\infty} G^+_{\ell}(r,s) \, \dfrac{\sqrt{\ell(\ell+1)}}{s} \partial_s \left( \dfrac{s^2 \, g^m_{\ell}(s)}{\mathsf{F}_{\ell}(s)\gamma(s) p_0(s)} \right)\, \mathrm{d}s  \, .
\end{aligned}
\end{equation}
Term $I_1$ gives readily the expression of $G^{\mathrm{PP}}_{\ell}$ in \cref{GPP}, 
by using adiabaticity \cref{r_c_increase}.
 We next consider $I_2$ and $I_3$. Replacing $G^+_\ell$ by \cref{GPP} we rewrite $I_2$ as 
\begin{equation}\label{temp_I2_gPP}
I_2 = \int_0^{\infty}
\left(s\dfrac{\alpha_{p_0}(s)}{ \gamma(s)} 
               -s\alpha_{\gamma p_0}(s) -1 \right)
       G^{\mathrm{P}\mathrm{P}}_{\ell}\,   g^m_{\ell}(s)
        \,  \dfrac{\sqrt{\ell(\ell+1)}}{\mathsf{F}_{\ell}(s) } \,    s^2\mathrm{d} s\,.
\end{equation}
For $I_3$, since we are working with $g^m_{\ell}$ of compact support, 
together with the fact that, 
$\mathcal{G}_{\ell}^+(0) = 0$,
we can carry out integration by parts in $I_3$, which gives,
\begin{equation}
\begin{aligned}
I_3 &= -\int_0^{\infty} \left( s\partial_s G^{\mathrm{PP}}_{\ell}\right)  \dfrac{ g^m_{\ell}(s) \sqrt{\ell(\ell+1)}}{\mathsf{F}_{\ell}(s)} \,  s^2 \mathrm{d} s 
   \,+ \,\int_0^{\infty}  G^{\mathrm{P}\mathrm{P}}_{\ell}  \alpha_{\gamma p_0}(s)  \dfrac{s\, g^m_{\ell}(s)}{\mathsf{F}_{\ell}(s)} \,
 \sqrt{\ell(\ell+1)}    \, s^2\, \mathrm{d} s \\
&\hspace*{2cm} - \int_0^{\infty} G^{\mathrm{P}\mathrm{P}}_{\ell}   \dfrac{g^m_{\ell}(s)}{\mathsf{F}_{\ell}(s)} \,  \sqrt{\ell(\ell+1)}     \, s^2\, \mathrm{d} s \, . 
\end{aligned}
\end{equation}
Putting together the above expression with \cref{temp_I2_gPP}, we obtain the expression \cref{GPB_GPP} for $G^{\mathrm{P}\mathrm{B}}_{\ell}$. 
%\begin{equation}
%\begin{aligned}
%G^{\mathrm{P}\mathrm{B}}_{\ell}(r,s)  &= 
%\sqrt{\ell(\ell+1)} \dfrac{s\tfrac{\alpha_{p_0}(s)}{ \gamma(s)} 
%               -2}{\mathsf{F}_{\ell}(s)}
%       G^{\mathrm{P}\mathrm{P}}_{\ell}\,   
%       - \sqrt{\ell(\ell+1)}\left( s\partial_s G^{\mathrm{PP}}_{\ell}\right)  \dfrac{ 1}{\mathsf{F}_{\ell}(s)}\\
%       %
%    &= -\dfrac{\sqrt{\ell(\ell+1)}}{\mathsf{F}_{\ell}(s)}
%    \left[s\partial_s G^{\mathrm{PP}}_{\ell}\, + \, G^{\mathrm{P}\mathrm{P}}_{\ell} \left( 2  - s\tfrac{\alpha_{p_0}(s)}{ \gamma(s)} \right)    \right]
%  \end{aligned}
%      \end{equation}

\paragraph{Part 2a : kernel for the horizontal coefficients}  
%We recall the definition of $\mathsf{b}^m_{\ell}$ from \cref{intb_green}
%\begin{equation*}
%\dfrac{\mathsf{b}^m_{\ell}}{\sqrt{\ell(\ell+1)}} \,= \,\textcolor{black}{-}\dfrac{r}{\mathsf{F}_{\ell}(r)}\partial_r a^m_{\ell} \,\textcolor{black}{-}\, \dfrac{2 \,-\, r\tfrac{\alpha_{p_0}(r)}{ \gamma(r)} }{\mathsf{F}_{\ell}(r)} a^m_{\ell}
%\,\textcolor{black}{-} \, \dfrac{r^2}{\mathsf{F}_{\ell}}\dfrac{\Rsun^2 }{\gamma\, p_0} \dfrac{g^m_{\ell}}{\sqrt{\ell(\ell+1)}} \,.
%\end{equation*}
We start with the expression \cref{intb_green} of $\mathsf{b}^m_{\ell}$. 
We use result \cref{aGPPB} to compute the derivative of $a^m_\ell$,
\begin{equation}
\partial_r a^m_{\ell} \, = \, \Rsun^2
\int_0^{\infty} (\partial_rG^{\mathrm{P}\mathrm{P}}_{\ell}(r,s) )\, f^m_{\ell}(s)
 \,s^2\,  \mathrm{d} s
\,   + \, \Rsun^2\int_0^{\infty}(\partial_r G^{\mathrm{P}\mathrm{B}}_{\ell}(r,s) ) \, g^m_{\ell}(s)\, s^2\, \mathrm{d}s \,.
\end{equation}
Substituting the above expression and \cref{aGPPB} into \cref{intb_green}, we obtain
\begin{equation}
\begin{aligned}
&\dfrac{b^m_{\ell}}{\Rsun^2} \,= \,\textcolor{black}{-}\dfrac{r \sqrt{\ell(\ell+1)}}{\mathsf{F}_{\ell}(r)}\left( \int_0^{\infty} (\partial_rG^{\mathrm{P}\mathrm{P}}_{\ell}(r,s) ) f^m_{\ell}(s)
 \,s^2\,  \mathrm{d} s
\,   + \, \int_0^{\infty}(\partial_r G^{\mathrm{P}\mathrm{B}}_{\ell}(r,s) ) g^m_{\ell}(s)\, s^2\, \mathrm{d}s \right) \\
& \hspace*{0.5cm}\,\textcolor{black}{-}\dfrac{\sqrt{\ell(\ell+1)}}{\mathsf{F}_{\ell}(r)}\left(2- r\dfrac{\alpha_{p_0}(r)}{ \gamma(r)} \right)
\left(
\int_0^{\infty} G^{\mathrm{P}\mathrm{P}}_{\ell}(r,s)\, f^m_{\ell}(s)
 s^2 \mathrm{d} s
 +  \int_0^{\infty} G^{\mathrm{P}\mathrm{B}}_{\ell}(r,s)  g^m_{\ell}(s)\, s^2\, \mathrm{d}s \right)\\
&\hspace*{1cm} \,\textcolor{black}{-} \, \dfrac{r^2}{\mathsf{F}_{\ell}(r)\, \gamma(r)\, p_0(r)} \langle \delta(r - s) ,  g^m_{\ell}(s) \rangle \,.
\end{aligned}
\end{equation}
Thus, we end up with expression \cref{GBP_GPP} for $G^{\mathrm{BP}}_{\ell}$, 
as well as the following one for $G^{\mathrm{BB}}_{\ell}$,
\begin{equation}\label{preGBB}
\begin{aligned}
%%G^{\mathrm{BP}}
%%  &=  \dfrac{r \sqrt{\ell(\ell+1)}}{\mathsf{F}_{\ell}(r)} \partial_rG^{\mathrm{P}\mathrm{P}}_{\ell}(r,s) 
%%   + \sqrt{\ell(\ell+1)}\dfrac{2 \,-\, r\tfrac{\alpha_{p_0}(r)}{ \gamma(r)} }{\mathsf{F}_{\ell}(r)} G^{\mathrm{P}\mathrm{P}}_{\ell}(r,s)\\
%%%
%G^{\mathrm{BP}}_{\ell}&= \textcolor{black}{-}\dfrac{ \sqrt{\ell(\ell+1)}}{\mathsf{F}_{\ell}(r)}
%\left[ r  \partial_rG^{\mathrm{P}\mathrm{P}}_{\ell}(r,s) 
%\, + \,  G^{\mathrm{PP}}_{\ell}(r,s)\left( 2 \,-\, r\dfrac{\alpha_{p_0}(r)}{\gamma(r)}\right) \right] \\
%%
G^{\mathrm{BB}}_{\ell}
&= \textcolor{black}{-}\dfrac{ \sqrt{\ell(\ell+1)}}{\mathsf{F}_{\ell}(r)}
\left[ r  \partial_rG^{\mathrm{PB}}_{\ell}(r,s) 
 + \,  G^{\mathrm{PB}}_{\ell}(r,s)\left( 2 - r\dfrac{\alpha_{p_0}(r)}{\gamma(r)}\right) \right]
 - \,  \dfrac{r^2}{\mathsf{F}_{\ell}(r) \gamma(r) p_0(r)} \dfrac{\delta(r-s)}{s^2}\,.
\end{aligned}
\end{equation}
%-----------------------------------------------------------------

\paragraph{Part 2b: Expression of $G^{\mathrm{BB}}_{\ell}$} We further rewrite \cref{preGBB}.
We need the derivative of $G^{\mathrm{PB}}_{\ell}$:
%\begin{equation}
%\begin{aligned}
%-r\partial_r G^{\mathrm{PB}}_{\ell}
%& = r \partial_r \left( \dfrac{\sqrt{\ell(\ell+1)}}{\mathsf{F}_{\ell}(s)}
%    \left[s\partial_s G^{\mathrm{PP}}_{\ell}\, + \, G^{\mathrm{P}\mathrm{P}}_{\ell} \left( 2  - s\tfrac{\alpha_{p_0}(s)}{ \gamma(s)} \right)    \right]\right)\\
%    %%
%&= \dfrac{\sqrt{\ell(\ell+1)}}{\mathsf{F}_{\ell}(s)}
%    \left[r\partial_r s\partial_s G^{\mathrm{PP}}_{\ell}\, + \,\left( 2  - s\tfrac{\alpha_{p_0}(s)}{ \gamma(s)} \right) r\partial_rG^{\mathrm{P}\mathrm{P}}_{\ell}     \right] \, ,
%\end{aligned}
%\end{equation}
\begin{equation}
\begin{aligned}
-r\partial_r G^{\mathrm{PB}}_{\ell}
= \dfrac{\sqrt{\ell(\ell+1)}}{\mathsf{F}_{\ell}(s)}
    \left[r\partial_r s\partial_s G^{\mathrm{PP}}_{\ell}\, + \,\left( 2  - s\dfrac{\alpha_{p_0}(s)}{ \gamma(s)} \right) r\partial_rG^{\mathrm{P}\mathrm{P}}_{\ell}     \right] \, .
\end{aligned}
\end{equation}
Expression \cref{GBB_GPP} is obtained by noting that the term in the square bracket of \cref{preGBB} is,
\begin{align*}
 &-r\partial_r G^{\mathrm{PB}}_{\ell}
 - G^{\mathrm{PB}}_{\ell}\left(  2 -r \dfrac{\alpha_{p_0}(r)}{\gamma(r)}\right)
 =  \dfrac{\sqrt{\ell(\ell+1)}}{\mathsf{F}_{\ell}(s)}
    \left[r\partial_r s\partial_s G^{\mathrm{PP}}_{\ell}\, + \,\left( 2  - s\dfrac{\alpha_{p_0}(s)}{ \gamma(s)} \right) r\partial_rG^{\mathrm{P}\mathrm{P}}_{\ell}     \right.\\
    & \hspace*{2cm} +  \left. \left(  2 -r \dfrac{\alpha_{p_0}(r)}{\gamma(r)}\right)s\partial_s G^{\mathrm{PP}}_{\ell}\, + \, G^{\mathrm{P}\mathrm{P}}_{\ell} \left(  2 -r \dfrac{\alpha_{p_0}(r)}{\gamma(r)}\right)\left( 2  - s\dfrac{\alpha_{p_0}(s)}{ \gamma(s)} \right)    \right] \, . 
\end{align*}
\end{proof}

From expression  \cref{GPP}  of $G^{\mathrm{PP}}_\ell$ in \cref{coeff_orig::prop}
and \cref{G+_phitdphi:def} for $G^+_\ell$, we can write $G^{\mathrm{PP}}_\ell$ as, 
\begin{equation}\label{GPP_sfG}
G^{\mathrm{PP}}_{\ell}(r,s) =  -\mathsf{G}^+_{\ell}(r,s)   \,\dfrac{ \mathfrak{p}(s)}{s^2} \,, \hspace*{0.3cm}
\text{with} \hspace*{0.2cm} \mathsf{G}^+_{\ell} \text{ defined in } \cref{sfGell::def}\,.
\end{equation}
It remains to compute the derivatives of $G^{\mathrm{PP}}_\ell$ in terms of the quantities in \cref{primker}.
We will need the following lemmas.

%
%\begin{mdframed}
%\begin{lemma}\label{Green3dcoeff::lem}
%\begin{itemize}[leftmargin = *]
%
%\item With  $\phi_\ell$ and $\phi_\ell^+$  both regular on $r > 0$, 
%      the functions $\mathsf{G}^+_{\ell}(r,s)$, $\mathsf{T}^+_{\ell}(r,s)$ and $\mathsf{Q}^+_{\ell}(r,s)$
%      are regular for $r,s  > 0$. 
% \item   Quantities $\mathsf{G}^+_\ell$ and $\mathsf{Q}^+_{\ell}$ are symmetric in $(r, s)$:
%\begin{equation}\mathsf{Q}^+_\ell(r,s) \, = \, \mathsf{Q}^+_\ell(s,r)\,,
%\quad \mathsf{G}^+_\ell(r,s)\, = \, \mathsf{G}^+_\ell(s,r)\,.
%\end{equation}
%On other hand, $\mathsf{T}^+_\ell$ is not with
%\begin{equation}
%\mathsf{T}^+_\ell(s,\textcolor{black}{r})  \, = \,  \mathrm{H}(s-\textcolor{black}{r})\,  \textcolor{black}{\phi_\ell(r)} \, s\partial_s \phi_\ell^+(s) \,\,+\,\, \mathrm{H}(\textcolor{black}{r}-s) \, \, s\phi_\ell'(s) \, \textcolor{black}{\phi_\ell^+(r)}\,.
%\end{equation}
%
%\item Due to the symmetry of $G^{\mathrm{PP}}_{\ell}$, 
%      the quantity $\tfrac{\mathfrak{p}(s)}{s^2}$ is independent of $s$, i.e.
%\begin{equation}\label{dsresult_frakp}
%\partial_s \dfrac{\mathfrak{p}(s)}{s^2} = 0\,.
%\end{equation}
%In addition, $\mathfrak{F}$ of \cref{eq:mathfrakF}
%is regular on $r > 0$, and $\mathcal{W}_\ell^+(s)$ defined in \cref{calW+::def} is 
%regular on $r > 0$.
%\end{itemize}
%\end{lemma}
%\end{mdframed}

%---------------------------------

\begin{lemma}\label{Green3dcoeff::lem}
The quantity $\tfrac{\mathfrak{p}(s)}{s^2}$ defined in \cref{eq:mathfrakp} is independent of $s$, i.e. $\partial_s \tfrac{\mathfrak{p}(s)}{s^2} = 0$. 
%\begin{equation}\label{dsresult_frakp}
%\partial_s \dfrac{\mathfrak{p}(s)}{s^2} = 0\,.
%\end{equation}

\end{lemma}

\begin{proof}
We compute the derivatives of terms appearing in the definition 
of  $\mathfrak{p}$. The derivative of $\mathcal{W}^+_\ell$ defined in \cref{calW+::def} was obtained 
using Abel's identity\footnote{This is derived as follows: starting from its 
definition, $\mathcal{W}_\ell\{\phi_\ell,\phi_\ell^+\}(s)
= \phi_\ell(s) \phi_\ell^{+\prime}(s) - \phi_\ell'(s) \phi_\ell^+(s)$, we have, 
\begin{equation}
\begin{aligned}
 \partial_s \mathcal{W}\{\phi_\ell,\phi_\ell^+\}(s) &= \phi_\ell'(s) \phi_\ell^{+\prime}(s) 
                                                     + \phi_\ell(s) \partial_s^2\phi_\ell^{+}(s) 
 - \phi_\ell''(s) \phi_\ell^+(s)
- \phi_\ell'(s) \phi_\ell^{+\prime}(s)\\
&=-\dfrac{ \phi_\ell(s)}{\tilde{q}_1} ( \tilde{q}_2 \phi_\ell^{+'}(s) + \tilde{q}_3 \phi_\ell^{+}(s)    )
 +\dfrac{ \phi_\ell^+(s)}{\tilde{q}_1} (  \tilde{q}_2 \phi_\ell'(s) + \tilde{q}_3 \phi_\ell(s)    )
 = -\dfrac{\tilde{q}_2(s)}{\tilde{q}_1(s)}  \mathcal{W} \{\phi_\ell,\phi_\ell^+\}(s)\,.
 \end{aligned}
\end{equation}
} and the fact that $\phi_\ell$, $\phi^+_\ell$ are solutions 
to \cref{eq:phi_phi_problem}, here we work with the form of $\Lorigin$ 
given in \cref{oldcoeffs::rmk},
\begin{equation}
\begin{aligned}
 &- \dfrac{\partial_s \mathcal{W} \{\phi_\ell,\phi_\ell^+\}(s)}{  \mathcal{W}_\ell\{\phi_\ell,\phi_\ell^+\}(s)} 
 = \dfrac{\tilde{q}_2(s)}{\tilde{q}_1(s)}
 \,, \hspace*{1cm} \partial_s \dfrac{1}{ c_0^2 \rho_0 s^2}
 = \dfrac{1}{ c_0^2 \rho_0 s^2}\left( \alpha_{c_0^2 \rho_0} - \dfrac{2}{s}\right),\\
 & \hspace*{2cm} \text{and} \hspace*{1cm} \partial_s  \dfrac{\mathsf{F}_\ell}{\mathsf{F}_0}
  = \dfrac{ \mathsf{F}_0'}{ \mathsf{F}_0} - \dfrac{\mathsf{F}_\ell \mathsf{F}_0' }{\mathsf{F}_0^2}
  = \dfrac{\mathsf{F}_\ell}{\mathsf{F}_0}
  \dfrac{  \ell(\ell+1) \mathsf{F}'_0}{\mathsf{F}_{\ell} \, \mathsf{F}_0}\,.
\end{aligned}
\end{equation}
For the third identity, we have used  $\mathsf{F}_{\ell} - \mathsf{F}_0 = -\ell(\ell+1)$.
We thus have
\begin{equation}
\partial_s  \dfrac{\mathfrak{p}(s)}{s^2}
 =  \dfrac{\mathfrak{p}}{s^2}\, \left(\dfrac{\tilde{q}_2(s)}{\tilde{q}_1(s)} + \alpha_{c_0^2 \rho_0}-\dfrac{2}{s}
+   \ell(\ell+1)\dfrac{   \mathsf{F}'_0(s)}{\mathsf{F}_{\ell}(s)\, \mathsf{F}_0(s)}\right) \,.
 \end{equation}
The first statement is obtained by substituting in the definition of the coefficients from \cref{oldcoeffs}, 
and using adiabatic condition \cref{assump:g1},
$c^2_0\rho_0 = \gamma p_0$, 
 \begin{equation}
 \dfrac{\tilde{q}_2(s)}{\tilde{q}_1(s)} =  -\alpha_{c^2_0\rho_0 } + \dfrac{2}{s}
\, - \, \ell(\ell+1) \dfrac{\mathsf{F}_0'(s)}{\mathsf{F}_0(s)\, \mathsf{F}_{\ell}(s)}\,.
\end{equation}
\end{proof}

\begin{lemma}\label{Green3dcoeff::lem2}
 As defined in \cref{primker}, the derivatives of $\mathsf{G}^+_\ell$  are related to
$\mathsf{T}_\ell^+$ and $\mathsf{Q}^+_\ell$ as follows,
\begin{subequations}
\begin{align}
r\partial_r \mathsf{G}^+_{\ell}(r,s)\,  &=\,  \mathsf{T}^+_{\ell}(\textcolor{black}{r},s)
\,, \qquad s\partial_s \mathsf{G}^+_{\ell}(r,s)\,  =\,  \mathsf{T}^+_{\ell}(s,\textcolor{black}{r}) \,;\\[0.3em]
 s\partial_s\,  r\partial_r \mathsf{G}^+_{\ell}(r,s)\,&=\,  r\partial_r \,  s\partial_s \mathsf{G}^+_{\ell}(r,s) \, =\,  \mathsf{Q}^+_\ell(r,s)\,  -\,  r^2 \, \mathcal{W}^+_\ell(r) \, \textcolor{black}{\delta(r-s)}\,.
\end{align}
\end{subequations}

\end{lemma}

\begin{proof}
We next compute 
the derivatives of  $\mathsf{G}^+_{\ell}(r,s)$.
 Using the identity with the distributional derivative of the Heaviside with a smooth function $f$, 
 \begin{equation}(f(x) \mathrm{H}(x-x_0))' \, =\,  f(x_0) \, \delta(x-x_0) \, +\,  f'(x)\,  \mathrm{H}(x-x_0)\,,
 \end{equation}
 we obtain
\begin{equation}
\begin{aligned}
r\partial_r \mathsf{G}^+_{\ell}(r,s)
&= \mathrm{H}(s-r) \, r\partial_r\phi_\ell(r) \,\phi_\ell^+(s) \,\,+ \,\,  \mathrm{H}(r-s) \, \phi_\ell(s)\,  r\partial_r\phi_\ell^{+}(r)\\
& \hspace*{2cm}   -r\, \delta(s-r) \, \phi_\ell(s) \,\phi_\ell^+(s) \,\,+ \,\,  s\, \delta(r-s) \, \phi_\ell(s)\,  \phi_\ell^+(s) \\[0.3em]
&= \mathrm{H}(s-r) \, r\partial_r\phi_\ell(r) \,\phi_\ell^+(s) \,\,+ \,\,  \mathrm{H}(r-s) \, \phi_\ell(s)\,  r\partial_r\phi_\ell^{+}(r) \,.
\end{aligned}
\end{equation}
We next compute the derivatives with respect to $s$ in a similar manner and get, 
\begin{equation}
\begin{aligned}
 s\partial_s \mathsf{G}^+_{\ell}(r,s)
\, &=\,  +\mathrm{H}(s-r) \, \phi_\ell(r) \,s\partial_s\phi_\ell^+(s) \,\,+ \,\,  \mathrm{H}(r-s) \, s\partial_s\phi_\ell(s)\,  \phi_\ell^{+}(r)\,; \\[0.5em]
s\partial_s \, r\partial_r \mathsf{G}^+_{\ell}(r,s)
&= \mathrm{H}(s-r) \, r\partial_r\phi_\ell(r) \,s\partial_s\phi_\ell^+(s) \,\,+ \,\,  \mathrm{H}(r-s) \, s\partial_s\phi_\ell(s)\,  r\partial_r\phi_\ell^{+}(r) \\[0.3em]
&\hspace*{2.5cm} +\underbrace{s^2\delta(s-r) \,\partial_s\phi_\ell(s) \,\phi_\ell^+(s) \,\,- \,\,  s^2\delta(r-s) \, \phi_\ell(s)\,  \partial_s\phi_\ell^{+}(s)}_{- s^2\delta(r-s) \mathcal{W}\{\phi_\ell,\phi_\ell^+\}(s)} \,.
%
%&\hspace*{0.2cm}= \mathrm{H}(s-r) \, r\phi_\ell'(r) \,s\partial_s\phi_\ell^+(s) \,\,+ \,\,  \mathrm{H}(r-s) \, s\partial_s\phi_\ell(s)\,  r\partial_r\phi_\ell^{+}(r)  - s^2\delta(r-s) \mathcal{W}\{\phi_\ell,\phi_\ell^+\}(s)\,.
\end{aligned}
\end{equation}

\end{proof}

As a result of \cref{Green3dcoeff::lem2}, we obtain the derivatives of $G^{\mathrm{PP}}_\ell$.

\begin{corollary}[Derivatives of $G^{\mathrm{PP}}_{\ell}$]\label{derGPP::cor}
\begin{subequations}
\begin{align}
r\partial_r G^{\mathrm{PP}}_{\ell}(r,s) &=  - \dfrac{\mathfrak{p}(s)}{s^2}\,  r\partial_r \mathsf{G}^+_{\ell}(r,s) 
= - \dfrac{\mathfrak{p}(s)}{s^2}\, \mathsf{T}^+_\ell(r,s) \,;\\
s\partial_s G^{\mathrm{PP}}_{\ell}(r,s)
&= - \dfrac{\mathfrak{p}(s)}{s^2}\,  s\partial_s \mathsf{G}^+_{\ell}(r,s)
 = - \dfrac{\mathfrak{p}(s)}{s^2}\,\mathsf{T}^+_\ell(\textcolor{black}{s},r)   \,; \\
s\partial_s \, r\partial_r G^{\mathrm{PP}}_{\ell}(r,s) &= 
 - \dfrac{\mathfrak{p}(s)}{s^2}\,  s\partial_s \, r\partial_r \mathsf{G}^+_{\ell}(r,s) 
  =  \dfrac{\mathfrak{p}(s)}{s^2}\,  \left(-\mathsf{Q}^+_\ell(r,s) +  r^2  \mathcal{W}^+_\ell(r)  \textcolor{black}{\delta(r-s)} \right)\,.
\end{align}
\end{subequations}

\end{corollary}

%-----------------------------------------------------------

\begin{proposition}\label{compoG_reg:prop}
We have the following expressions for the coefficients of the Green's 
tensor in VSH basis in terms of the quantities introduced in \cref{primker},
%
%In terms of the quantities $\mathsf{G}^+_\ell$, $\mathsf{T}^+_\ell$ and $\mathsf{Q}^+_\ell$ \cref{sfGell::def,sfTell::def,sfQell::def}, the components of the 3D Green's kernels are, 
\begin{subequations}
\begin{align}
G^{\mathrm{PP}}_{\ell}(r,s) &=   -\mathsf{G}^+_{\ell}(r,s)\, \dfrac{\mathfrak{p}(s)}{s^2}
\,\, = \,\,  - \mathsf{G}^+_{\ell}(r,s)\, \dfrac{\mathsf{F}_{\ell}(s)}{\mathcal{W}^+_{\ell}(s)\, \mathfrak{F}(s)\, s^2} \,;\label{GPPfinal}\\
G^{\mathrm{BP}}_{\ell}(r,s)
&= \sqrt{\ell(\ell+1)}
 \left[ \mathsf{T}^+_\ell(r,s)
\, + \,  \mathsf{G}^+_\ell(r,s) \left( 2 \,-\, r\dfrac{\alpha_{p_0}(r)}{\gamma(r)}\right) \right] \dfrac{1}{\mathcal{W}^+_\ell(r)\ \mathfrak{F}(r)\, r^2}\,;\label{GPBfinal}\\
%----
G^{\mathrm{P}\mathrm{B}}_{\ell}(r,s) 
  &=  \sqrt{\ell(\ell+1)}
    \left[\mathsf{T}^+_{\ell}(s,r)\, + \, \mathsf{G}^+_{\ell}(r,s) \left( 2  - s\dfrac{\alpha_{p_0}(s)}{ \gamma(s)} \right)    \right]\dfrac{1}{\mathcal{W}^+_{\ell}(s)\, \mathfrak{F}(s)\, s^2}\,;
\end{align}
\end{subequations}

\vspace*{-.4cm}

\begin{subequations}\label{GBBf}
\begin{align}
&\text{and}\hspace*{2cm} G^{\mathrm{BB}}_{\ell} (r,s)= G^{\mathrm{BBreg}}_{\ell}(r,s)\,\,   - \,\, \dfrac{r^2}{\mathfrak{F}(r) \, s^2} \delta(r-s) \,, \hspace*{0.3cm} \text{where} \,,\\
&G^{\mathrm{BBreg}}_{\ell}(r,s)
 =  -\dfrac{\ell(\ell+1)}{\mathsf{F}_{\ell}(r)}
 \left[ \mathsf{Q}^+_{\ell}(r,s) + \left( 2  - s\dfrac{\alpha_{p_0}(s)}{ \gamma(s)} \right) \mathsf{T}^+_{\ell}(r,s)
+ \left(  2 -r \dfrac{\alpha_{p_0}(r)}{\gamma(r)}\right) \mathsf{T}^+_{\ell}(s,r) \right. \nonumber\\
&\left.\hspace*{3cm} + \left(  2 -r \dfrac{\alpha_{p_0}(r)}{\gamma(r)}\right) \left( 2  - s\dfrac{\alpha_{p_0}(s)}{ \gamma(s)} \right)  \mathsf{G}^+_{\ell}(r,s) \right] \dfrac{1}{ \mathcal{W}^+_{\ell}(s) \, \mathfrak{F}(s)\, s^2} \,. \label{GBBregf}
\end{align}
\end{subequations}
Additionally, we have the following symmetry of the kernels in $r$ and $s$,
\begin{equation}
\begin{aligned}
&G^{\mathrm{PP}}_{\ell}(r,s) 
=  G^{\mathrm{PP}}_{\ell}(s,r) 
 \,, \hspace*{0.3cm}  G^{\mathrm{PB}}_{\ell}(r,s) 
 = G^{\mathrm{BP}}_{\ell}(s,r)  ,\\
&\hspace*{1cm} G^{\mathrm{BBreg}}_{\ell}(r,s) = G^{\mathrm{BBreg}}_{\ell}(s,r)\,.
\end{aligned}
\end{equation}

\end{proposition}

\begin{proof}
\textbf{Part 1} We have obtained from \cref{GPP_sfG} the expression for $G^{\mathrm{PP}}_\ell$. For the expression of $G^{\mathrm{BP}}_{\ell}$ and $G^{\mathrm{P}\mathrm{B}}_{\ell}$, we start with
 \cref{GBP_GPP,GPB_GPP} and substitute in the expression for derivatives of $G^{\mathrm{PP}}_\ell$ given in \cref{derGPP::cor}.
 Obtaining the expression for $G^{\mathrm{BP}}_{\ell}$ is a bit less straightforward. We have,
\begin{equation}
\begin{aligned}
G^{\mathrm{BP}}_{\ell}(r,s)
&  =\dfrac{ \sqrt{\ell(\ell+1)}}{\mathsf{F}_{\ell}(r)}
 \left[ \mathsf{T}^+_\ell(r,s)
\, + \,  \mathsf{G}^+_\ell(r,s)\left( 2 \,-\, r\dfrac{\alpha_{p_0}(r)}{\gamma(r)}\right) \right] \dfrac{\mathfrak{p}(s)}{s^2}\\
&= \dfrac{ \sqrt{\ell(\ell+1)}}{\mathsf{F}_{\ell}(r)}
 \left[ \mathsf{T}^+_\ell(r,s)
\, + \,  \mathsf{G}^+_\ell(r,s)\left( 2 \,-\, r\dfrac{\alpha_{p_0}(r)}{\gamma(r)}\right) \right] \dfrac{\mathfrak{p}(r)}{r^2}\,.
% = \text{l.h.s of } \cref{GPBfinal} \,.
\end{aligned}
\end{equation}
To obtain \cref{GPBfinal}, we use the constancy of
$s\mapsto \tfrac{\mathfrak{p}(s)}{s^2}$ from \cref{Green3dcoeff::lem}, i.e.,
$\tfrac{\mathfrak{p}(s)}{s^2}  =  \tfrac{\mathfrak{p}(r)}{r^2}$. 
Then simplification comes from the definition of $\mathfrak{p}$ given in \cref{eq:mathfrakp}.
%\begin{equation}\label{frakp-sfF}
%\dfrac{\mathfrak{p}(s)}{\mathsf{F}_{\ell}(s) \, s^2} 
% = \dfrac{1}{s^2\, \mathfrak{F}(s)\,  \mathcal{W}_\ell^+(s)} \, .
%\end{equation}
%%----
%For $G^{\mathrm{P}\mathrm{B}}_{\ell}$, the derivation is more straightforward:
%\begin{equation}
%\begin{aligned}
%G^{\mathrm{P}\mathrm{B}}_{\ell}(r,s) 
%&= \dfrac{\sqrt{\ell(\ell+1)}}{\mathsf{F}_{\ell}(s)} 
%    \left[\mathsf{T}^+_\ell(s,r)\, + \, \mathsf{G}^+_\ell(r,s) \left( 2  - s\tfrac{\alpha_{p_0}(s)}{ \gamma(s)} \right)    \right]\dfrac{\mathfrak{p}(s)}{s^2} \\
%    %
%  &=  \sqrt{\ell(\ell+1)}
%    \left[\mathsf{T}^+_\ell(s,r)\, + \, \mathsf{G}^+_\ell(r,s) \left( 2  - s\tfrac{\alpha_{p_0}(s)}{ \gamma(s)} \right)    \right]\dfrac{1}{\mathcal{W}^+_{\ell}(s)\, \mathfrak{F}(s)\, s^2}\,.
%\end{aligned}
%\end{equation}

\medskip

\noindent \textbf{Part 2} We next obtain the expression \cref{GBBf} for $G^{\mathrm{BB}}_{\ell}$. We start with the expression for $G^{\mathrm{BB}}_{\ell}$ 
given in \cref{GBB_GPP} and decompose this expression as 
$G^{\mathrm{BB}}_{\ell}
 =  I_1 + I_2 +  I_3$, with
\begin{equation}
\begin{aligned}
I_1&=  \dfrac{ \ell(\ell+1)}{\mathsf{F}_{\ell}(r)\,\mathsf{F}_{\ell}(s)}
    r\partial_r s\partial_s G^{\mathrm{PP}}_{\ell}  \hspace*{0.5cm} \, \textcolor{black}{-} \,  \dfrac{r^2}{\mathsf{F}_{\ell}(r)\, \gamma(r)\, p_0(r)} \dfrac{\delta(r-s)}{s^2}\,;\\
I_2 & = \dfrac{ \ell(\ell+1)}{\mathsf{F}_{\ell}(r)\,\mathsf{F}_{\ell}(s)}
    \left[    \,\left( 2  - s\dfrac{\alpha_{p_0}(s)}{ \gamma(s)} \right) r\partial_rG^{\mathrm{P}\mathrm{P}}_{\ell}     
\,\,+  \,\, \left(  2 -r \dfrac{\alpha_{p_0}(r)}{\gamma(r)}\right)s\partial_s G^{\mathrm{PP}}_{\ell}\right]\,;\\
I_3 & = \dfrac{ \ell(\ell+1)}{\mathsf{F}_{\ell}(r)\,\mathsf{F}_{\ell}(s)}G^{\mathrm{P}\mathrm{P}}_{\ell} \left(  2 -r \dfrac{\alpha_{p_0}(r)}{\gamma(r)}\right)\left( 2  - s\dfrac{\alpha_{p_0}(s)}{ \gamma(s)} \right) \,.
\end{aligned}
\end{equation}
Term $I_3$ and $I_2$ are readily transformed by using \cref{GPPfinal,derGPP::cor} to give the second to last term in square bracket of $G^{\mathrm{BBreg}}_\ell$. To rewrite $I_1$,  we start by substituting in the expression of $r\partial_r s\partial_s G^{\mathrm{PP}}_{\ell} $ from \cref{derGPP::cor},
%\begin{equation}
%\begin{aligned}
%I_1&=\dfrac{ \ell(\ell+1)}{\mathsf{F}_{\ell}(r)\,\mathsf{F}_{\ell}(s)}
%    r\partial_r \,s\partial_s G^{\mathrm{PP}}_{\ell}
% \,\,    \textcolor{black}{-} \,\,  \dfrac{r^2}{\mathsf{F}_{\ell}(r)\, \gamma(r)\, p_0(r)} \dfrac{\delta(r-s)}{s^2}\\
%   %%
%     &=- \dfrac{ \ell(\ell+1)}{\mathsf{F}_{\ell}(r)\,\mathsf{F}_{\ell}(s)}
% \dfrac{\mathfrak{p}(s)}{s^2}\, \,     r\partial_r\,  s\partial_s \mathsf{G}^+_{\ell}
%\,\,     \textcolor{black}{-} \,\,  \dfrac{r^2}{\mathsf{F}_{\ell}(r)\, \gamma(r)\, p_0(r)} \dfrac{\delta(r-s)}{s^2}\\
% %
%&=- \dfrac{ \ell(\ell+1) \mathfrak{p}(s)}{s^2\, \mathsf{F}_{\ell}(r)\,\mathsf{F}_{\ell}(s)}\mathsf{Q}^+_\ell(r,s) \,\, +\,\,
%\ell(\ell+1) \dfrac{  \mathfrak{p}(s)}{s^2\, \mathsf{F}_{\ell}(r)\,\mathsf{F}_{\ell}(s)} s^2\, \delta(r-s)\,  \mathcal{W}_\ell^+(s) \\
%%
% &\hspace*{8cm} \textcolor{black}{-} \, \, \dfrac{r^2}{\mathsf{F}_{\ell}(r)\, \gamma(r)\, p_0(r)} \dfrac{\delta(r-s)}{s^2}\\
% %
% &= - \dfrac{ \ell(\ell+1) \mathfrak{p}(s)}{s^2\, \mathsf{F}_{\ell}(r)\,\mathsf{F}_{\ell}(s)}\mathsf{Q}^+_\ell(r,s) \,\, - \,\,\dfrac{\delta(r-s)}{\gamma(r) p_0(r) \mathsf{F}_0(r)}\,.
%\end{aligned}
%\end{equation}
\begin{equation}
\begin{aligned}
I_1
&=- \dfrac{ \ell(\ell+1) \mathfrak{p}(s)}{s^2\, \mathsf{F}_{\ell}(r)\,\mathsf{F}_{\ell}(s)}\mathsf{Q}^+_\ell(r,s) \, +\,
\dfrac{ \ell(\ell+1)  \mathfrak{p}(s)s^2\,  \mathcal{W}_\ell^+(s)}{s^2\, \mathsf{F}_{\ell}(r)\,\mathsf{F}_{\ell}(s)} \delta(r-s) 
 \textcolor{black}{-} \, \, \dfrac{r^2}{\mathsf{F}_{\ell}(r)\, \gamma(r)\, p_0(r)} \dfrac{\delta(r-s)}{s^2}\\
 &= - \dfrac{ \ell(\ell+1) \mathfrak{p}(s)}{s^2\, \mathsf{F}_{\ell}(r)\,\mathsf{F}_{\ell}(s)}\mathsf{Q}^+_\ell(r,s) \,\, - \,\,\dfrac{\delta(r-s)}{\gamma(r) p_0(r) \mathsf{F}_0(r)}\,.
\end{aligned}
\end{equation}
The  terms involving $\delta(r-s)$ are simplified by
using   $\frak{p} \mathcal{W}_\ell^+ = \tfrac{\mathsf{F}_\ell}{ \gamma p_0 \mathsf{F}_0}$ and $\mathsf{F}_{\ell} = \mathsf{F}_0 - \ell(\ell+1)$, thus
\begin{equation}
\begin{aligned}
& \dfrac{ \ell(\ell+1) \mathfrak{p}(s)\,s^2\,   \mathcal{W}_\ell^+(s) }{\mathsf{F}_{\ell}(r)\,\mathsf{F}_{\ell}(s)}  \delta(r-s) \,\,\textcolor{black}{-} \, \, \dfrac{r^2}{\mathsf{F}_{\ell}(r)\, \gamma(r)\, p_0(r)} \dfrac{\delta(r-s)}{s^2}\\
& \hspace*{1cm} =\,\, \delta(r-s) \dfrac{\ell(\ell+1)  - \mathsf{F }_0(r)}{\gamma(r) \, p_0(r)\, \mathsf{F}_0(r)\, \mathsf{F}_{\ell}(r)} \,\, = \,\, - \dfrac{\delta(r-s)}{\gamma(r) \, p_0(r) \, \mathsf{F}_0(r)}\,.
\end{aligned}
\end{equation}

%
%\medskip
%
%\noindent $\bullet$ The second term $I_2$ is rewritten as follows,
%\begin{equation}
%\begin{aligned}
%I_2 &= \dfrac{ \ell(\ell+1)}{\mathsf{F}_{\ell}(r)\,\mathsf{F}_{\ell}(s)}
% \,\left( 2  - s\dfrac{\alpha_{p_0}(s)}{ \gamma(s)} \right) r\partial_rG^{\mathrm{P}\mathrm{P}}_{\ell}     
%\,\,+  \,\, \dfrac{ \ell(\ell+1)}{\mathsf{F}_{\ell}(r)\,\mathsf{F}_{\ell}(s)}\left(  2 -r \dfrac{\alpha_{p_0}(r)}{\gamma(r)}\right)s\partial_s G^{\mathrm{PP}}_{\ell}\\
%%
%&= \,-\dfrac{ \ell(\ell+1)}{\mathsf{F}_{\ell}(r)\,\mathsf{F}_{\ell}(s)}\left( 2  - s\dfrac{\alpha_{p_0}(s)}{ \gamma(s)} \right)  \dfrac{\mathfrak{p}(s)}{s^2}\,  r\partial_r \mathsf{G}^+_{\ell}  \,\,-  \,\, \dfrac{ \ell(\ell+1)}{\mathsf{F}_{\ell}(r)\,\mathsf{F}_{\ell}(s)}\left(  2 -r \dfrac{\alpha_{p_0}(r)}{\gamma(r)}\right) \dfrac{\mathfrak{p}(s)}{s^2}\,  s\partial_s \mathsf{G}^+_{\ell}\,.
%%
%\end{aligned}
%\end{equation}

\medskip

\noindent  \textbf{Part 3 - Symmetry relation} 
The symmetry of $\mathsf{Q}_{\ell}^+$ and $\mathsf{G}^+_{\ell}$ provide readily the relation for $G^{\mathrm{PB}}_\ell$ and $G^{\mathrm{BP}}_\ell$. Those of $G^{\mathrm{BBreg}}$ and $G^{\mathrm{PP}}$ use additionally the constancy $\tfrac{ \mathfrak{p}(s)}{s^2}$ which gives 
\begin{equation}
\dfrac{1}{\mathsf{F}_{\ell}(r)}
\dfrac{1}{ \mathcal{W}^+_{\ell}(s) \, \mathfrak{F}(s)\, s^2} 
 \,\,=\,\, \dfrac{1}{\mathsf{F}_{\ell}(r)\, \mathsf{F}_{\ell}(s)} \dfrac{\mathfrak{p}(s)}{s^2}\,.
\end{equation}
\end{proof}

 \paragraph{Revisited formal expansion}
 In light of \cref{coeff_orig::prop,compoG_reg:prop}, 
 in separating out the Dirac distribution in $G^{\mathrm{BB}}_\ell$, 
 we can write the solution $\xib$ to $\boldsymbol{\mathcal{L}}\,\xib
 =  \mathbf{f} $ in $ \mathbb{R}^3$ as
 \begin{equation}
 \xib = \Rsun^2 \left< \mathbb{G}^+_\mathrm{SG} , \mathbf{f}\right>\,, \hspace*{0.2cm}\text{ with}\hspace*{0.2cm}
\mathbb{G}^+(\mathbf{x},\mathbf{y}) \, = \, \mathbb{G}^+_{\mathrm{reg}}(\mathbf{x},\mathbf{y})
\, - \,  \dfrac{\lvert \mathbf{x}\rvert^2}{\textcolor{black}{ 
   \mathfrak{F}(\lvert \mathbf{x}\rvert)}}
   \delta(\mathbf{x}-\mathbf{y})\,   \mathrm{P}_\parallel\,, 
   \end{equation}
   where $\mathrm{P}_\parallel $ is the projection, $\mathrm{P}_\parallel 
 = \mathbb{Id} - \mathbf{e}_r(\hat{\mathbf{x}}) \otimes \mathbf{e}_r(\hat{\mathbf{y}})$, and 
 \begin{equation} \mathbb{G}^+_{\mathrm{reg}}  = \mathbb{G}^+_{\perp} + \mathbb{G}^+_{\mathrm{h-reg}}\,,
 \end{equation} having the following formal expansions in VSH tensor basis,
 \begin{equation}
 \begin{aligned}
\mathbb{G}^+_{\perp}(\mathbf{x},\mathbf{y})
&= \sum_{(\ell,m)} 
  \textcolor{black}{G^{\mathrm{PP}}_{\ell}(\lvert\mathbf{x}\rvert,\lvert\mathbf{y}\rvert)} \,
\mathrmb{P}^m_{\ell}(\widehat{\mathbf{x}})  \otimes\overline{\mathrmb{P}^m_{\ell}(\widehat{\mathbf{y}})} 
+ \sideset{}{'}\sum_{(\ell,m)} 
\textcolor{black}{G^{\mathrm{BP}}_{\ell}(\lvert\mathbf{x}\rvert,\lvert\mathbf{y}\rvert)} \,\, 
\mathrmb{B}^m_{\ell}(\widehat{\mathbf{x}})  \otimes\overline{\mathrmb{P}^m_{\ell}(\widehat{\mathbf{y}})}\,;\\ 
%%----------------------------------------
\mathbb{G}^+_{\mathrm{h-reg}}(\mathbf{x},\mathbf{y})&= 
 \sideset{}{'}\sum_{(\ell,m)} 
\textcolor{black}{G^{\mathrm{PB}}_{\ell}(\lvert\mathbf{x}\rvert,\lvert\mathbf{y}\rvert)} \, \,
\mathrmb{P}^m_{\ell}(\widehat{\mathbf{x}})  \otimes\overline{\mathrmb{B}^m_{\ell}(\widehat{\mathbf{y}})}
  + \sideset{}{'}\sum_{(\ell,m)} 
\textcolor{black}{G^{\mathrm{BBreg}}_{\ell}(\lvert\mathbf{x}\rvert,\lvert\mathbf{y}\rvert)} \, 
\mathrmb{B}^m_{\ell}(\widehat{\mathbf{x}})  \otimes\overline{\mathrmb{B}^m_{\ell}(\widehat{\mathbf{y}})} 
\,. 
 \end{aligned}
 \end{equation}

% ----------------------------------------------------------------------------------------
\subsection{Algorithm for constructing Green's tensor using assembling formula}
\label{subsection:assembling-algorithm}

We propose \cref{algorithm:assemble} which implements the 
result of \cref{compoG_reg:prop} to compute the directional 
kernels of the Green's tensor $\mathbb{G}^+$,
\begin{equation}\label{quantocompute_rep}
G^{\mathrm{PP}}_\ell, \quad G^{\mathrm{PB}}_\ell, \quad
G^{\mathrm{BP}}_\ell, \quad G^{\mathrm{BB}}_\ell\,.
\end{equation}
 Recall that \cref{compoG_reg:prop} is developed from the assembling formula \cref{G+_phitdphi:def}.
 For the scalar equation, the Green's kernel is a scalar quantity whose coefficient in harmonic basis 
 is computed by working directly with \cref{G+_phitdphi:def}, cf. \cite[Algorithm 3.2]{barucq2020efficient}. 
 On the other hand, the kernel for the vector equation, $\mathbb{G}$, has four components with different levels of singularity. It is thus here that one fully exploits the benefits of the `assembling' idea, now conveyed by \cref{compoG_reg:prop}.

For numerical resolution, we need the following modification from previous sections. 
We need to `regularize' the singularities of the coefficient of $\Lorigin$, which is 
done by multiplying both sides of equation $\Lorigin w = f$ by $\mathsf{F}_\ell^2$ 
for $\ell>0$ and by $r^2$ for $\ell=0$, the coefficients of the ODE employed 
in \cref{algorithm:assemble} are thus,
\begin{equation}
 q_i =r^2 \hat{q}_i\,, \hspace*{0.2cm}\text{for } \,\, \ell=0 \,, \hspace*{1cm} q_i = \mathsf{F}_\ell^2\hat{q}_i  \,, \hspace*{0.2cm} \text{for } \,\, \ell>0 \,,
 \qquad\qquad i\,=\,\{1,\,2,\,3\}.
\end{equation}
The regular-at-0 boundary condition \cref{regzerocond::def} at $r=0$ is replaced 
by an equivalent one independent of $\ell$, $ru'=0$ at $r=0$, 
cf. \cite[Section 5.1]{barucq:hal-03406855}. 
The outgoing-at-infinity condition \cref{outinfcond::def} is replaced by the ABC 
conditions and placed at a specific height in the low atmosphere. 
The ABC conditions are constructed for the Schr\"odinger operator $\Lconjug$ 
in \cref{section:rbc}, and those for the original operator $\Lorigin$ are obtained 
via the change of variable \cref{conjmodalODErhs} which relates the two formulations.
Specifically, for $u = \frak{J}_\ell\, \tilde{u}$,
\begin{equation} \label{eq:rbc-equivalent:origin-conj}
%%% \text{for } u = \frak{J}_\ell\, \tilde{u} \,: \hspace*{0.9cm} 
u' = \mathcal{Z}_\mathrm{abc} u \hspace*{0.3cm} \Leftrightarrow  \hspace*{0.3cm} \frak{J}_\ell\, ru' = r\,\mathcal{N}_\mathrm{abc}\, u \,, \,\, \text{with } \hspace*{0.1cm} \mathcal{N}_\mathrm{abc}:=
 \mathcal{Z}_\mathrm{abc}\, \frak{J}_\ell  + \frak{J}'_\ell \,.
\end{equation}
Note that \cref{algorithm:assemble} is written with boundary 
condition \cref{eq:rbc-equivalent:origin-conj} at $r=\rmax$, 
however it can work with other types of boundary condition.
We refer to \cite[Section 5]{barucq2020outgoing} for more details on implementation.

\begin{remark}
One can work either with a first-order formulation of the original ODE or with its Schr\"odinger form, cf. \cite[Appendix A]{barucq:hal-03406855} for implementation with the latter. We will investigate the numerical stability of both forms in \cref{section:numerics:original-and-conjugated}.
\end{remark}

\paragraph{Features of the algorithm}
%The advantages of \cref{algorithm:assemble} come  directly from the fact it is based on \cref{compoG_reg:prop} and employs first-order formulations discretized with Hybridizable Discontinuous Galerkin (HDG) method.
 Employing \cref{compoG_reg:prop} to compute the Green's tensor has the following advantages:
\begin{enumerate} \setlength{\itemsep}{-2pt}
  \item The singularity of the Dirac is  avoided, thus the need to employ special techniques such as mesh refinement or singularity extraction to maintain accuracy is removed;
\item Since one works with regular solutions, the singularity at equal height of source and receiver, $r=s$, is exactly described by the Heaviside distribution; %,  \cref{algorithm:assemble} offers a more numerically stable approach;
  \item From just two resolutions (the regular-at-0 solution $\phi$ and the outgoing-at-infinity $\phi^+$), all components of $\mathbb{G}$, i.e., all coefficients  \cref{quantocompute_rep}, are given for any $r,s\in (\epsilon,\rmax)$ with $\epsilon \ll 1$ and $\rmax > \ra$.
\end{enumerate}
 Moreover, using first-order formulation discretized with HDG method further brings the following additional benefits:
\begin{enumerate} \setlength{\itemsep}{-2pt}
  \item Without post-processing, each resolution gives both the primal 
        unknown and its derivative, both of which are needed to compute 
        the quantities in \cref{primker} and in \cref{compoG_reg:prop};
  \item With HDG method, the global discretized system is only in terms of one 
        unknown (despite solving a first-order system).
\end{enumerate}

\begin{algorithm}[ht!]
\caption{Computation of the Green's kernels for any position $r$, $s$ in
         interval $(\ra,\rb)$ at frequency $\omega$ and mode $\ell$ 
         using \cref{compoG_reg:prop}. 
         % Here $\phi_\ell, \phi_\ell^+$ are written as $\phi_1$ and $\phi_2$ respectively.
         }
\label{algorithm:assemble}
\begin{algorithmic}
\STATE \noindent \textbf{1a.} 
       Compute $\phi_{\ell} := \mathfrak{I}(\rb)w$ and $\partial_r \phi_{\ell} \,:=\, \mathfrak{I}(\rb)v/r$,
       where $(w, \, v)$ solve,
       \vspace*{-0.5em}
       \begin{equation} \label{eq:first-order-assemble_1}
       \left\lbrace \begin{aligned}
       r \, q_1(r)  \partial_r v(r) \,+\, q_2(r) \, v(r) \,+\, q_3(r)\, w(r) & \,=\, 0 
                                    \,,\qquad\qquad  r\in(0,\rb), \\
       r \partial_r w \,-\, v & \,=\, 0 \,,\qquad\qquad  r\in(0,\rb), \\
       v\mid_{r=0}   \,= \, 0\, , \qquad \qquad w\mid_{r=\rb} \, =\, 1 \,. 
       \end{aligned} \right. \end{equation}

\STATE \noindent \textbf{1b.} 
       Compute $\phi_\ell^+ := \mathfrak{I}(\ra) \, w$ and 
       $\partial_r \phi_\ell^+ \,:=\,\mathfrak{I}(\ra) v/r$ 
       where $(w, \, v)$ solve,
       \vspace*{-0.5em}
       \begin{equation} \label{eq:first-order-assemble_2}
       \left\lbrace \begin{aligned}
       r \, q_1(r)  \partial_r v(r) \,+\, q_2(r) \, v(r) \,+\, q_3(r)\, w(r) & \,=\, 0 
                                    \,,\qquad\qquad  r\in(\ra,\rmax), \\
       r \partial_r w \,-\, v & \,=\, 0 \,,\qquad\qquad  r\in(\ra,\rmax), \\
       w\mid_{r=\rb} \, =\, 1 \,, \qquad
       \Big(\mathfrak{I}_\ell \, v \,-\, r\,\mathcal{N} \, w \Big)_{r=\rmax} = 0. 
       \end{aligned} \right. \end{equation}

  \STATE\textbf{2.} 
         Compute Wronskian function $\wronskian$ and $\frak{F}$ with $\Fz$ 
         from \cref{sfFell::def},
        \begin{equation}
        \wronskian(s) = \phi_\ell(s) \partial_s\phi_\ell^+(s) - \phi_\ell^+(s)\partial_s\phi_\ell(s)\,,
        \hspace*{0.3cm}
        \mathfrak{F}(r) = c_0^2(r) \rho_0(r) \Fz(r) \,.
        \end{equation}

 \STATE \textbf{3.} Compute 
   $\mathsf{G}^+_{\ell}$, $\mathsf{T}^+_{\ell}$ , $\mathsf{Q}^+_{\ell}$ using 
   \cref{primker} with $\phi_\ell$, $\phi_\ell^+$ and their derivatives.

 \STATE \textbf{4.} Compute kernels $\GPP$, $\GBP$ and $\GBB$ using \cref{compoG_reg:prop} and $\GPB(r,s) = \GBP(s,r)$ .

\end{algorithmic}
\end{algorithm}

\begin{remark}[Dirac-source approach]\label{App1_vector::rmk} 
We highlight here the difficulties encountered if one tries to  compute components \cref{quantocompute_rep} by solving directly the modal equations with a Dirac-type source,  in a similar vein to \cite[Algorithm 3.1]{barucq2020outgoing}. 

\begin{itemize}[leftmargin = *]\setlength{\itemsep}{-1pt}
\item At first glance, one can solve directly for components \cref{quantocompute_rep} 
      by employing \cref{abGPPB} of \cref{coeff_orig::prop}. Specifically, 
for a given source height $s_0 $, 
one solves  $\Lorigin u = \frak{f}$ with $\frak{f}=\frak{f}^P$ and $\frak{f}^B$ respectively, defined as images of mapping \cref{origmodalODErhs} with
 $(f_\ell= \delta(r-s_0), g_\ell=0) \mapsto \frak{f}^\mathrm{P}$
 and $(f_\ell= 0$, $g_\ell=\delta(r-s_0)) \mapsto \frak{f}^\mathrm{B}$.
Working with a first-order formulation, resolution with $\frak{f}^P$ would give $r\mapsto u_\mathrm{P}(r;s_0)$ and its derivative. From the primal solution, one obtains $G^\mathrm{PP}$; together with its derivative, we compute $G^{\mathrm{BP}}$ by using  \cref{intb_green}. In the same manner, resolution with $\frak{f}^B$ would give $G^{\mathrm{PB}}$ and $G^{\mathrm{BB}}$. However,  with $G^{\mathrm{PP}}$ containing a Dirac, cf. \cref{GBB_GPP} of \cref{coeff_orig::prop}, this approach is numerically unstable without singularity extraction.

\item One can attempt to make the above approach viable by first 
  rewriting \cref{GPB_GPP}--\cref{GBB_GPP} (which are currently in terms of $G^\mathrm{PP}_\ell$) in terms of $G^+_\ell$, by using the relation \cref{GPP} between these two kernels. The resulting expressions give $G^\mathrm{PP}_\ell, G^\mathrm{PB}_\ell,G^\mathrm{BP}_\ell, G^\mathrm{BBreg}_\ell$ in terms of $s\partial_s$ , $r\partial_r$ , $r\partial_r s\partial_s$ of $G^+_\ell$.
 Next, for each source height $s_0$, we solve $\Lorigin G^+_\ell=\delta(r-s_0)$ for 
 $r\mapsto G^+_\ell(\cdot,s_0)$ and $ \partial_r G^+_\ell(\cdot,s)$. 
 Its derivatives $s\partial_s$ and $s\partial_s r\partial_r$  could be computed by finite-differences.

\end{itemize}

This approach incurs the inherent disadvantages as pointed in \cite{barucq2020outgoing} for a scalar problem:
(1) each resolution giving   value of $G^+_\ell$ associated with one source height, 
(2) having to deal with singularity of the Dirac source, 
(3) losing accuracy near equal source and receiver height, $r=s$. 
This loss of accuracy now poses a bigger problem for the vector equation
as the derivatives $s\partial_s$ and $s\partial_s r\partial_r$ of a solution which 
approximates a Heaviside 
are computed in a post-processing step (e.g., with finite differentiation). % and more errors risk accumulating.
This problem is remedied in \cref{compoG_reg:prop,algorithm:assemble} 
by prescribing analytically the singularities. %analytically in \cref{Green3dcoeff::lem}.
%We note that this problem is remedied in \cref{compoG_reg:prop} and \cref{algorithm:assemble} by book-keeping analytically these dirac terms which show up due to distributional derivative of Heavisde and are to shown to cancel out
%(with the exception of $G^\mathrm{BB}_\ell$) in \cref{Green3dcoeff::lem}. 
\end{remark}

% ----------------------------------------------------------------------------------------
\section{Absorbing boundary conditions (ABC)}
\label{section:rbc}
% ----------------------------------------------------------------------------------------

In this section, we  adapt the approach of \cite{barucq2018atmospheric,fournier2017atmospheric,barucq2020efficient} 
and apply it to the modal operator $\mathcal{L}_\ell=-\partial^2_r + V_\ell$ 
in \cref{notation_modalop} in order to construct boundary conditions to approximate 
the outgoing-at-infinity condition \cref{outinfcond::def}. These comprise 
of a nonlocal boundary condition and its approximations that are local in 
$\ell(\ell+1)$.
The second group represents conditions that can be implemented readily in 
higher dimensions (and in frequency domain) by replacing $\ell(\ell+1)$ by 
$\Delta_{\mathbb{S}^2}$. Compared to the procedure taken in 
\cite{barucq2018atmospheric,fournier2017atmospheric,barucq2020efficient} for 
the scalar equation, dealing with $V_\ell$ requires more care due to the 
presence of propagative region in the atmosphere for 
$(\omega,\ell)$ below
the Lamb frequency, i.e. $\omega \leq \mathcal{S}_\ell$, cf. \cref{Lamb_freq::def}.  
As a result of this region, unlike for the scalar case, solutions
converge slowly in certain region to the 
oscillatory behavior at infinity prescribed by $e^{\ii \mathsf{k}_a r}$. 
For regions bearing similar features as in the scalar case, the construction of ABC 
remains more technical for the vector case.
For the scalar equation, explicit expression of the modal exterior 
Dirichlet-to-Neumann operator exists (in terms of Whittaker function), 
cf. \cite{barucq2020outgoing,barucq2020efficient}; secondly, the 
construction of local families is facilitated by having a potential 
polynomial in $\ell(\ell+1)$.

We recall the form of $V_\ell$ \cref{Vell::def} in the atmosphere: 
\begin{equation}\label{Vellatmo}
\text{for $r\geq \ra$, } \quad V_\ell(r) = - \mathsf{k}_\mathrm{a}^2  +
\dfrac{\phi_0''(r)}{c_0^2(r)}  - \dfrac{\upeta_\mathrm{a}}{r}  + \dfrac{\upnu^2_\ell(r) - \tfrac{1}{4}}{r^2}\,.
\end{equation}
This is due to the constancy \cref{specialfeaturesatmo} of certain coefficients in this region; specifically, for $r\geq r_\mathrm{a}$, 
\begin{equation}\label{eq:atmospheric-quantities-k-eta}
k_0^2 (r)= k^2_a := \Rsun^2 \dfrac{\sigma^2}{c^2_a}\,, \hspace*{0.5cm} \mathsf{k}^2(r) = \mathsf{k}^2_\mathrm{a}: =  k^2_{a} -\dfrac{\alatmo^2}{4}  \,, \hspace*{0.5cm}  \upeta (r) = \upeta_\mathrm{a} := \dfrac{\alatmo}{\adatmo}(2-\adatmo)\,.
\end{equation}
To highlight the new features of the vector equation, we will compare $V_\ell$ with 
the potential of the modal Schr\"odinger ODE of the scalar wave equation, cf. \cite{barucq2020outgoing,barucq2020efficient},
 \begin{equation}\label{scalarSODE}
\text{scalar problem:}\quad
-\partial_r^2 + V_\ell^\mathrm{scalar}\,, \hspace*{0.2cm} \text{with} \hspace*{0.2cm} V_\ell^\mathrm{scalar}(r)= - \mathsf{k}_{\mathrm{a}}^2 + \dfrac{\alatmo}{r} + \dfrac{\ell(\ell+1)}{r^2} \,, \quad r\geq \ra \,.
\end{equation}   In \cref{signV::subsec}, we will start by investigating the sign of $V_\ell$ in order to choose the correct square root branch used to define the nonlocal BC in \cref{NLBC::def}, 
and secondly construct its square root approximations 
that are local in $\ell(\ell+1)$ in \cref{sqrtABC::subsec}.
We will employ the following square root conventions,
\begin{equation}\label{sqrt::def}
\sqrt{z} = \lvert z\rvert e^{\mathrm{i} \tilde{\Theta}(z)/2}
\, \hspace*{0.1cm} \text{with} \hspace*{0.1cm} \tilde{\Theta}(z) \in [0,2\pi)\,, \hspace*{0.7cm}
 (z)^{1/2} = \lvert z\rvert e^{\mathrm{i}\Theta(z)/2}
\, \hspace*{0.1cm} \text{with} \hspace*{0.1cm}  \Theta(z)\in(-\pi,\pi]\,. 
 \end{equation}
 
% ---------------------------------------------------------------------------
\subsection{Potential for solar model \texttt{S} }\label{signV::subsec}
% ---------------------------------------------------------------------------

\paragraph{Behavior at infinity}
From \cite{barucq2021outgoing}, % it is given by 
the oscillatory behavior as $r\rightarrow \infty$
is given, modulo higher terms, by $e^{\pm\ii\mathsf{k}_\mathrm{a} r}$ 
for all $(\omega,\ell)$ (with $+$ 
designated as $\mathsf{k}_a$-outgoing), since 
\begin{equation}\label{infbehavior}
\text{for fixed }\hspace*{0.2cm} (\omega,\ell): \hspace*{0.5cm}  -V_\ell\rightarrow  \mathsf{k}_a^2, \hspace*{0.3cm} \text{as}  \hspace*{0.2cm}r\rightarrow \infty \,.
\end{equation}
Additionally, the zero of the real part of $\mathsf{k}_a^2$ 
defines the atmospheric cut-off frequency $\omega_t$, 
 \begin{equation}\label{awc::def}
 \mathrm{Re}\, \mathsf{k}^2_\mathrm{a} >  0 \hspace*{0.5cm}\Leftrightarrow \hspace*{0.5cm}  \omega > \omega_t \,, \hspace*{0.2cm} \text{with } \omega_t :=\dfrac{\alatmo c_\mathrm{a}}{2\Rsun} 
 = \lim_{r\rightarrow \infty} \omega_c   \sim 2\pi \times 5.2 \,\si{\milli\Hz}\,.
 \end{equation}
 We also note for the sign of the imaginary part of $\mathsf{k}_a^2$,
 \begin{equation}  \label{signatt::def}
 \mathrm{Im}\, \mathsf{k}_\mathrm{a}^2 = 2\ii \omega \Gamma > 0 \,, \hspace*{0.5cm}\text{and} \hspace*{0.3cm} \mathsf{k}_\mathrm{a}:=\sqrt{\mathsf{k}_\mathrm{a}^2} \,.
 \end{equation}
This explains the choice of square root $\sqrt{\cdot}$ \cref{sqrt::def} 
and that outgoing waves are represented with the $+$ sign.
We also introduce the frequency $N_{\mathrm{rad}}$ 
which is the maximum of the buoyancy frequency $N$ 
\cref{Buoy_freq::def} in the interior of 
the Sun up the tacholine (the transitional region 
between the radiative and convective zone in the interior 
of the Sun),
 \begin{equation}\label{Buoyrad_freq::def}
  N_{\mathrm{rad}} := \max_{r\in [0,0.7]}  N(r) \, , 
  \hspace*{1cm} \dfrac{N_{\mathrm{rad}}}{2\pi}\sim \hspace*{0.2cm}0.47\, \si{\milli\Hz}\,.
 \end{equation}
 
\paragraph{Sign of $V_\ell$}
In \cref{fig:potential:main01}, we show the sign of the real part
of $V_\ell$ and $V_\ell^\mathrm{scalar}$ at attenuation $\Gamma/(2\pi)=20$\si{\micro\Hz} 
for three frequencies: \num{0.2} (below $ N_{\mathrm{rad}} $), \num{2} (below cut-off $\omega_t$) 
and \num{7}\si{\milli\Hz} (above cut-off $\omega_t$),
% with the first two being below atmospheric cut-off frequency $\omega_t$ \cref{awc::def}. 
while in \cref{fig:potential:main02}, 
we fix a height $r=\num{1.001}$ (where ABC will be imposed) and show the sign 
of both real and imaginary part for all $(\ell,\omega)$. 
The two colors in these figures distinguish two behaviors of a wave acting like $\exp(\ii r \,\sqrt{-V_\ell})$:
in blue-colored regions where $V_\ell < 0$, waves are oscillating (propagating) as $r$ increases, 
and in red regions where  $V_\ell > 0$, waves are evanescent. 
Let us first discuss the behaviour of the solution in the atmosphere depending on the frequency:

\medskip

\noindent $\bullet\,\,$ {$\omega > \omega_t$} (defined in \cref{awc::def}): $\Vscalar$ and $ V_\ell$ are both negative and the solution is propagating in the atmosphere (case 7\si{\milli\Hz} in \cref{fig:potential:main01}).  

\medskip

\noindent $\bullet\,\,$ {$N_{\rm{rad}} < \omega < \omega_t$}: $V_\ell$ and $\Vscalar$ also have the same sign but both potentials are now positive, thus the solution is evanescent in the atmosphere (case 2\si{\milli\Hz} in \cref{fig:potential:main01}).

\medskip

\noindent $\bullet\,\,$ {$\omega < N_{\rm{rad}}$}:
this group presents features 
which distinguish the vector equation from the scalar one (case \num{0.2} \si{\milli\Hz}).
In the atmosphere, the sign of $V_\ell$ is not uniform with respect to 
      $\ell$, as seen in  \cref{fig:potential:main01} with a change of sign for  $\ell > 125$ at the surface. 
      This behavior is observed for $(\omega,\ell)$ with
      $\omega \leq \mathcal{S}_\ell(1.001)$ \cref{Lamb_freq::def}, 
as shown in \cref{fig:potential:main02} where 
$\mathrm{Re}(\,V_\ell(1.001;\omega))$ is negative (in blue) and is thus of opposite sign with
 $\mathrm{Re}(\,\Vscalar)$. However, due to \cref{infbehavior}, 
 this region will eventually turn positive for large enough $r$.
In another word, for these $(\omega,\ell)$, waves at height $r=1.001$ have 
not `converged' to the behavior at infinity described by $e^{- \ii \mathsf{k}_a r}$, while this `convergence' is already observed for $\omega > \mathcal{S}_\ell$ in low atmospheric height, see also \cref{inCZ::rmk}.

% ------------------------------
\graphicspath{{figures/potential_map/}}
% ------------------------------
\begin{figure}[ht!] \centering
  % ----------------------------------------------------------------------------------
  % 2D maps ==========================================================================
  \subfloat[][Sign of  $\mathrm{Re}\,(V_\ell)$.]
             {\includegraphics[]{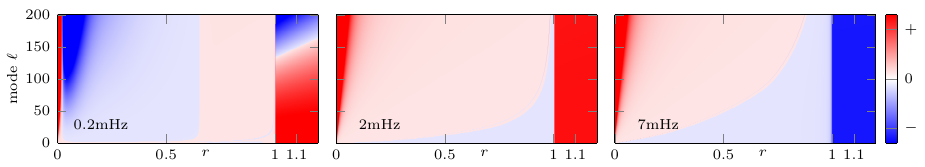} \label{fig:potential:main01a}}
              \vspace*{-1.0em}

  \subfloat[][Sign of  $\mathrm{Re}\,(\Vscalar)$.]
             {\includegraphics[]{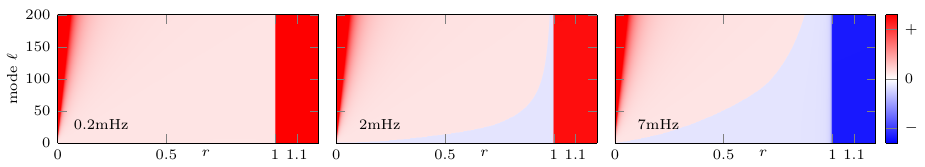}}

  \caption{Comparison between the signs of the real parts of $V_\ell$  \cref{Vell::def}
           and $\Vscalar$ \cref{scalarSODE} as functions of $(\ell,r)$ 
           at attenuation $\Gamma/(2\pi) = 20\si{\micro\Hz}$ for 
           frequencies \num{0.2}, \num{2} and \num{7}\si{\milli\Hz}.}
  \label{fig:potential:main01}
\end{figure}

% ------------------------------
\begin{figure}[ht!] \centering
  % ----------------------------------------------------------------------------------
  \subfloat[][Sign of $\mathrm{Re}\,V_\ell$ and $ \mathrm{Im}\, V_\ell$.]
             {\makebox[.45\linewidth][c]{\includegraphics[]{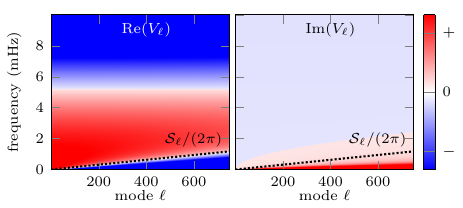}}} \hspace*{1.50em}
  \subfloat[][Sign of $\mathrm{Re}\,\Vscalar$ and $ \mathrm{Im}\, \Vscalar$.]
             {\makebox[.45\linewidth][c]{\includegraphics[]{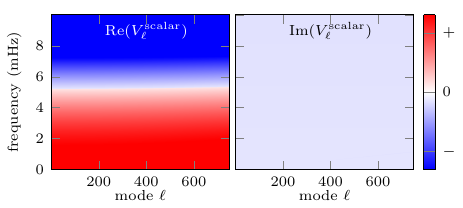}}}
 
  \caption{Comparison between sign of the real and imaginary 
           parts of $V_\ell$  \cref{Vell::def} and $\Vscalar$ \cref{scalarSODE} 
           at $r=\num{1.001}$ as functions of $(\ell, \omega)$
           at attenuation $\Gamma/(2\pi) = 20\si{\micro\Hz}$. 
           This is the height where we imposed the ABC in our numerical experiments. 
           We also show the 
           Lamb frequency $\mathcal{S}_\ell$  of \cref{Lamb_freq::def}.
           The local cut-off frequency $\omega_c/(2\pi)$ \cref{eq:cut-off-frequency}
           here is \num{5.2022}\si{\milli\Hz} for $V_\ell$
           and \num{5.2060}\si{\milli\Hz} for $\Vscalar$.
           }
           
  \label{fig:potential:main02}
\end{figure}

\begin{remark}[Interior propagative regions]\label{Inpropregion::rmk}
%In the interior, $V_\ell$ displays a propagative region (in blue) below the convective zone ($r < 0.7)$ as seen in \cref{fig:potential:main01a}. This region  which is absent for $\Vscalar$  exists for all harmonic modes $\ell$, specifically at frequencies lower than the buoyancy frequency $N_\mathrm{rad}$ \cref{Buoyrad_freq::def}, and gives rise to gravity modes, cf.~\cref{subsection:g-modes}.  
  %These gravity modes (also called g-modes) possess a second behavior being evanescent below the surface layer (indicated by the zone in red). This feature is shared by frequencies in range $(N ,S_\ell)$. For more visualization and discussion ofthis region, we refer to \cite[Section 7.1, Figures 7--10]{barucq:hal-03406855}
  For $N_{\rm rad} < \omega < \omega_t$, the propagative region below the surface observed for both $V_\ell$ and $\Vscalar$ corresponds to the region of the solar acoustic modes (p-modes). The radius where the potential changes sign is lower for low values of $\ell$ indicating that waves with smaller harmonic degree are traveling deeper into the Sun.
  For $\omega < N_{\rm rad}$, $V_\ell$ displays a propagative region below the convective zone ($r < 0.7)$ as seen in \cref{fig:potential:main01a}. This region  which is absent for $\Vscalar$  
exists for all harmonic modes $\ell$ and gives rise to internal-gravity 
waves (g-modes), cf.~\cref{subsection:g-modes}. For more visualization and discussion of this region, we refer to \cite[Section 7.1, Figures 7--10]{barucq:hal-03406855}. 
\end{remark}

\paragraph{Nonlocal boundary condition} Following the approach of Engquist and Majda \cite{engquist1977absorbing}, the construction of the nonlocal BC comes from factoring $\Lconjug$ as $\Lconjug = (\partial_r - \mathrm{i}\sqrt{-V_\ell}) (\partial_r + \mathrm{i} \sqrt{-V_\ell})+$ smoothing operator, and gives,
 \begin{equation}\label{NLBC::def} \partial_r u=  \ZrbcNL u\,, \hspace*{0.3cm} \text{with} \hspace*{0.3cm} \ZrbcNL \,:=\, \ii \sqrt{-V_\ell} \,.
 \end{equation}
% This comes from factoring $\mathcal{L}_\ell$ as $\mathcal{L}_\ell = (\partial_r - \mathrm{i}\sqrt{-V_\ell}) (\partial_r + \mathrm{i} \sqrt{-V_\ell})+$ smoothing operator. 
Since $\mathrm{Im}\,\sqrt{z} >0$ for $z \in \mathbb{C}\setminus\{ 0\}$, imposing 
$\partial_r - \ZrbcNL$ at $r=\rmax$ 
picks out wave decaying as $r$ increases in this neighborhood. In addition, at vanishing (positive) attenuation, i.e. $\Gamma \rightarrow 0^+$, wave either remains evanescent or converges to one behaving like $e^{\ii \mathsf{k}_\mathrm{a} r}$ (also called $\mathsf{k}_a$-outgoing) respectively when
\begin{equation}\label{CZsign}
\mathrm{Re}\,(-V_\ell) < 0  \hspace*{0.2cm} \text{(decaying)}\,, \hspace*{0.3cm} \text{ or }\hspace*{0.3cm}   \{ \mathrm{Re}\,(-V_\ell) >0 \text{ and }
\mathrm{Im}\,(-V_\ell) \geq 0  \} \hspace*{0.2cm} (\mathsf{k}_\mathrm{a}\text{-outgoing)}\,.
\end{equation} 
\begin{remark}\label{inCZ::rmk}
As observed in \cref{fig:potential:main02}, for region $(\omega,\ell)$ with $\omega < \mathcal{S}_\ell$, condition \cref{CZsign} is not observed. The nonlocal BC defined in \cref{NLBC::def} still chooses only decaying solutions (as $r$ increases).
However, in this region, as $\Gamma\rightarrow 0$, $\mathrm{Re}\,(-V_\ell) >0, 
\mathrm{Im}\,(-V_\ell) \rightarrow 0^-$, chosen solutions will behave like $e^{-\ii r\sqrt{\mathrm{Re}\, (-V_\ell)}}$. This is another indication that the behavior here does not yet converge to the behavior at infinity  described by $e^{\ii \mathsf{k}_a r}$, cf. \cref{infbehavior}. 
\end{remark}
%However, in this region, as $\Gamma\rightarrow 0$, chosen solutions will behave like $e^{-\ii \mathsf{k}_a r}$.

% ------------------------------------------------------
\subsection{Approximation of $V_\ell$ in the atmosphere}\label{ApproxVell::subsec}
% -------------------------------------------------------
In comparing between the expression of $V_\ell$ and $\Vscalar$, cf. 
\cref{scalarSODE,Vellatmo}, it is noted that 
the difficulty comes from the function $\nu_\ell(r)$ which is not a polynomial in $\ell(\ell+1)$, 
see \cref{eq::nuell}.  
For this reason, we will work with a potential $Q_{\ell}^\mathrm{G}$ introduced in \cite[Proposition 10]{barucq2021outgoing}
which provides a good  approximation of $-V_\ell$ in high atmosphere (i.e. for $r \gg 1$, fixed $(\ell,\omega)$),
\begin{subequations}
\begin{align}
&\hspace*{1cm} Q_{\ell}^\mathrm{G}(r) = 
 \mathsf{k}_\mathrm{a}^2  + \dfrac{\upeta_\mathrm{a}}{r}  +  \dfrac{2 \Rsun^2G\,  \mathfrak{m}}{c_\mathrm{a}^2 \, r^3}- \dfrac{\mu_{\ell}^2\, -\, \tfrac{1}{4}}{r^2}\label{QGell::def}\,,\\
& \text{with} \hspace*{0.7cm} \mu^2_{\ell} - \dfrac{1}{4} :=  2 + \ell(\ell+1) +\dfrac{\ell(\ell+1)}{\katmo^2} \dfrac{\alatmo}{\adatmo}  
  \left(\dfrac{\alatmo}{ \adatmo} -\alatmo\right) \,,\label{muell::def}\\
& \text{and constant } \hspace*{0.2cm} \frak{m} = 4\pi \int_0^{\ra} s^2\rho_0(s) ds  + 4\pi \rho_0(\ra) 
\dfrac{ (\alatmo \ra )^2 + 2 \alatmo \ra  +  2 }{ (\alatmo)^3  }\,.
 \end{align}
\end{subequations}
In absorbing gravity term into $\mu_\ell$ or $\upeta$, 
we introduce the following form of $Q_\ell^\mathrm{G}$,
\begin{equation}
Q_{\ell}^\mathrm{G}(r) =  \mathsf{k}_a^2  + \dfrac{\upeta^\mathrm{G}_\mathrm{a}(r)}{r} - \dfrac{\mu_{\ell}^2\, -\, \tfrac{1}{4}}{r^2}\\
  = \mathsf{k}_a^2  + \dfrac{\upeta_\mathrm{a}}{r} - \dfrac{(\mu^\mathrm{G}_{\ell}(r))^2\, -\, \tfrac{1}{4}}{r^2} \,,
\end{equation}
with `gravity-modified' coefficients,
\begin{equation}\label{eq:atmospheric-quantities-etaG}
 (\mu^\mathrm{G}_{\ell})^2 - \dfrac{1}{4}:=  \mu_{\ell}^2 - \dfrac{1}{4} -    \dfrac{2\,\Rsun^2 G  \mathfrak{m}}{c_\mathrm{a}^2 }\, \dfrac{1}{r} \,, \hspace*{0.5cm} \upeta^\mathrm{G}_\mathrm{a} (r)  :=   \upeta_\mathrm{a}  +    \dfrac{2\,\Rsun^2 G  \mathfrak{m}}{c_\mathrm{a}^2 } \dfrac{1}{r^2} \,.
\end{equation}
%From \cite[Proposition 9]{barucq2021outgoing}, we have
%\begin{equation}
%V_0(r)  =   -Q_0^\mathrm{G}+\, \mathrm{a.e.d.t} \,, \hspace*{0.4cm} 
%%
%V_\ell(r) = -Q_\ell^{\mathrm{G}} 
%+ \dfrac{\ell(\ell+1)}{\Gamma_a^2} \mathsf{O}(r^{-3})
%+  \dfrac{[\ell(\ell+1)]^2}{(\Gamma_a\, \omega)^2} \mathsf{O}(r^{-6})\,, \hspace*{0.2cm} \ell >0\,.
%\end{equation}
Below, we will investigate numerically whether this potential 
still gives a good representation of $-V_\ell$ in low atmosphere (i.e., for $r \gtrsim 1$). 
Replacing $-V_\ell$ by $ Q_{\ell}^\mathrm{G}$ in \cref{NLBC::def}, we obtain the approximate nonlocal condition, 
 \begin{equation}\label{zNLG::def}
\partial_r u=  \ZrbcNLaG u\,, \hspace*{0.2cm} \text{with} \hspace*{0.2cm} \ZrbcNLaG (r) := \ii \sqrt{ Q^G_{\ell}(r)} \,,
 \end{equation}
which will undergo square-root approximation, to construct local 
BC in \cref{sqrtABC::subsec}.
%\begin{remark}[Comparison with a Whittaker function]
%If we freeze the coefficient of $ \upeta^\mathrm{G}_\mathrm{a}$ on $[\frak{r},\infty)$, for $\frak{r}>\ra$, 
%we obtain a Whittaker potential,\begin{equation}\label{QWell::def}
%Q_{\ell}^\mathrm{W}(r;\frak{r}) 
%  =  \mathsf{k}_a^2  + \dfrac{\upeta^\mathrm{G}_\mathrm{a}(\frak{r})}{r} - \dfrac{\mu_{\ell}^2\, -\, \tfrac{1}{4}}{r^2}\,.
%\end{equation}
%The outgoing-at-infinity solution of operator $-\partial_r^2 + Q_{\ell}^\mathrm{W}(r;\frak{r})$ is a multiple of Whittaker function $\mathrm{W}_{-\chi,\mu_\ell}(-2\ii\,\mathsf{k}_\mathrm{a}\,r)$ and is captured by boundary condition, $ \partial_r u=  \mathcal{Z}_{\mathrm{W},\ell}u\ $ at $r=\frak{r}$ with
%\begin{equation}
%\mathcal{Z}_{\mathrm{W},\ell}(\frak{r}) := -2\ii\, \mathsf{k}_\mathrm{a} \dfrac{\mathrm{W}'_{-\chi,\mu_\ell}(-2\ii\,\mathsf{k}_\mathrm{a}\, \frak{r})}{\mathrm{W}_{-\chi,\mu_\ell}(-2\ii\,\mathsf{k}_\mathrm{a}\, \frak{r})} \,, \quad \text{with } \, \chi\,=\, -\ii \dfrac{\upeta^\mathrm{G}_\mathrm{a}(\frak{r})}{2 \mathsf{k}_\mathrm{a}} \, .
%\end{equation}
%Comparison with such a condition is useful since its solution has known in oscillatory behavior,
%cf. [Equation 2.35-2.37]\cite{barucq2020outgoing},
%\begin{equation}
% \mathrm{W}_{-\chi,\mu_\ell}(-2\ii\,\mathsf{k}_\mathrm{a}\,r) \,\,\sim\,\,   \exp \left(\ii \, \mathsf{k}_\mathrm{a} r - \chi \log ( 2\mathsf{k}_\mathrm{a} r)  +  \ii \tfrac{\pi}{2} \chi \ \right)\,, \hspace*{0.2cm} \text{ as } r\rightarrow \infty \,.
%\end{equation}
%\end{remark}

\paragraph{Numerical comparison in low atmosphere}
% ================================================
In \cref{fig:potential:1d}, we plot, for fixed frequency and mode, 
the point-wise relative error between $-V_\ell$ and $Q_\ell$ defined
by
\begin{equation}\label{eq:error:point-wise}
  \errpt(\,Q_\ell, \, r) \,=\, 
       \dfrac{\vert(- V_\ell(r)) \,-\, Q_\ell(r) \vert}
             {\vert V_\ell(r) \vert} \,.
\end{equation} 
Since we work with positive attenuation 
 $\Gamma>0$, potential $V_\ell$ (appearing in the denominator) is 
 non-zero.
We also compare with an approximation of $-V_\ell$ which ignores gravity,
\begin{equation}\label{Q0::def}
Q^0_{\ell}(r) = \mathsf{k}_\mathrm{a}^2  + \dfrac{\upeta_\mathrm{a}}{r}  - \dfrac{\mu_{\ell}^2 - \tfrac{1}{4}}{r^2} \,.
\end{equation}
We note that results in \cite{barucq2021outgoing}  prove and illustrate these errors  
in very high atmosphere (with $r \rightarrow \infty$) and for large $\ell$ and frequencies
in view of scattering theory.
Here, with construction of ABC in mind, we
focus on the low atmospheric region, namely for $r \in (1.001,\, 1.07)$, and various ranges of frequency and modes.

From \cref{fig:potential:1d}, we observe that the error \cref{eq:error:point-wise} increases as $\omega$ decreases and as $\ell$ increases. 
The approximations are the worst for \num{0.2}\si{\milli\Hz} (rendering the addition of gravity irrelevant).
However, for 2 and 7\si{\milli\Hz}, $Q_{\ell}^\mathrm{G}$ gives a good representation 
of $V_\ell$ on $[1.001,1.07)$. At these frequencies, while the magnitude of error with $Q_\ell^0$ remains consistent 
for both frequencies and both modes, the inclusion of
gravity improves substantially the accuracy, an improvement which is more marked for lower $\ell$ and higher frequencies, 
with up to six orders of magnitude gained between $Q_{\ell}^\mathrm{G}$ and $Q_{\ell}^0$.

% -------------------------------------
\graphicspath{{figures/potential_1d/}}
% -------------------------------------
\begin{figure}[ht!] \centering
  % ----------------------------------------------------------------------------------
  % 1D maps ==========================================================================
  % \subfloat[][Potential and approximations in the atmosphere.]
  %           {\includegraphics[]{standalone_freq_mode_V1d}}
 
  % \subfloat[][Relative error $\errpt$ \cref{eq:error:point-wise} between the potential and approximations.]
             {\includegraphics[]{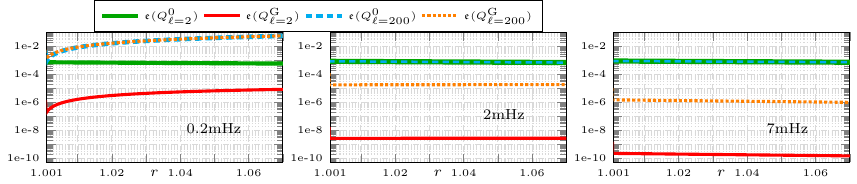}}
             
  \caption{
  Relative error $\errpt$ \cref{eq:error:point-wise} between the potential 
  and approximations $Q_{\ell}^\mathrm{G}$ of \cref{QGell::def} in the low
  atmosphere region ($r \in (1.001,\, 1.07)$) for three frequencies and 
  different modes.}
  \label{fig:potential:1d}
\end{figure}

% ----------------------------------------------------------------
\subsection{Square-root for approximate nonlocal BC} \label{sqrtABC::subsec}
% ----------------------------------------------------------------

Here we  approximate $\ZrbcNLaG$ introduced in \cref{zNLG::def} 
with HFG and SAIG 
families\footnote{In \cite{barucq2018atmospheric} HF stands for 
\emph{high-frequency} and SAI for \emph{small angle of incidence}. 
Here we add letter `G' to the nomenclature convention originally 
proposed in \cite{barucq2018atmospheric}, to indicate the contribution 
of gravity.} which are local in $\ell(\ell+1)$.
We work under the assumption,
\begin{equation}\label{sqrtapp::Ass}
\mathrm{Im}\,\, Q_{\ell}^\mathrm{G} \geq 0\,, \hspace*{0.5cm} \Gamma > 0\,,
\end{equation}
to switch to the principal square root $(\cdot)^{1/2}$ and work with the Taylor expansion 
\begin{equation}\label{sqrtTaylor}
\left(1 + \upepsilon\right)^{1/2}  = 1 + \dfrac{1}{2} \upepsilon - \dfrac{1}{8}\upepsilon^2 + \dfrac{1}{16} \upepsilon^3 -\ldots, \hspace*{0.5cm}\text{for } \hspace*{0.2cm} \upepsilon \in \mathbb{C}\,, \,\,\lvert \upepsilon\rvert < 1\,.
\end{equation}
For the HFG and SAIG, with $\mathsf{k}_\mathrm{a} := \sqrt{\mathsf{k}_\mathrm{a}^2} $, we normalize\footnote{Since $\mathrm{Im}\, \mathsf{k}_a^2 = \mathrm{Im} (\mathsf{k}_a^2+\tfrac{\eta^\mathrm{G}_a}{r}) = 2\omega \Gamma $, assumption $\Gamma > 0$ in \cref{sqrtapp::Ass} allows to write $(\tfrac{a}{b})^{1/2} $ as $ \tfrac{a^{1/2} }{ b^{1/2}}$, for $a,b\in\mathbb{C}^\star$.  } the coefficient inside the square root to make appear 1, respectively as,
\begin{equation}
\begin{aligned}
 \ZrbcNLaG (r)
           &   =     \ii\, \mathsf{k}_{\mathrm{a}}  \left( 1+ \varepsilon(r) \right)^{1/2} 
                 =  \ii
                   \left(  \mathsf{k}_{\mathrm{a}}^2 +  \dfrac{\upeta_{\mathrm{a}}^\mathrm{G}(r) }{ r}
                 \right)^{1/2} \hspace*{-0.5em} \left( 1 + \tilde{\varepsilon}(r)\right)^{1/2}\\
&            = \ii  \left(  \mathsf{k}_{\mathrm{a}}^2 +  \dfrac{\upeta_{\mathrm{a}} }{ r}\, 
                 \right)^{1/2} \left( 1 + \check{\varepsilon}(r)\right)^{1/2}    \,.
\end{aligned}
\end{equation}
Here 
 the small quantities $\varepsilon$, $\tilde{\varepsilon}$, ad $\check{\varepsilon}$ contain varying gravity term and are defined as, 
\begin{equation}
\varepsilon (r)=  \dfrac{\upeta^\mathrm{G}_{\mathrm{a}}(r)}{r\mathsf{k}_{\mathrm{a}}^2 }\, + \, \dfrac{\tfrac{1}{4} \, -\,\mu_{\ell}^2}{r^2\mathsf{k}_{\mathrm{a}}^2 }  \,, \hspace*{0.3cm}
\tilde{\varepsilon}(r) 
                =   \dfrac{\tfrac{1}{4}  \, -\,\mu_{\ell}^2 }{\mathsf{k}_{\mathrm{a}}^2\, r^2 + r\,\upeta_{\mathrm{a}}^\mathrm{G}(r) }\,,\hspace*{0.3cm} \check{\varepsilon}(r) 
                =   \dfrac{\tfrac{1}{4}  \, -\,[\mu_{\ell}^\mathrm{G}(r)]^2 }{\mathsf{k}_{\mathrm{a}}^2\, r^2 + r\,\upeta_{\mathrm{a}}(r) }\,,
\end{equation}
where $\mathsf{k}_{\mathrm{a}}$ and 
$\upeta_{\mathrm{a}}$ are given in  
\cref{eq:atmospheric-quantities-k-eta}, and 
$\upeta_{\mathrm{a}}^\mathrm{G}$ in \cref{eq:atmospheric-quantities-etaG}.
In assuming,
\begin{equation}\label{smallquant::assu}
\lvert \varepsilon(r)\rvert < 1  \hspace*{0.2cm}\text{ (for HFG) }, 
\hspace*{0.5cm} \lvert \tilde{\varepsilon}(r)\rvert < 1 \hspace*{0.2cm}\text{(for SAIG)}, 
\hspace*{0.5cm} \lvert \check{\varepsilon}(r)\rvert < 1  \hspace*{0.2cm}\text{(for } \ZrbcSAIGb),\end{equation}
respectively, and working modulo up to $\mathsf{O}(\varepsilon^2)$, $\mathsf{O}(\tilde{\varepsilon}^2)$ and $\mathsf{O}(\check{\varepsilon}^2)$, we obtain the following coefficients:
\begin{equation} \label{eq:rbc:ZSAI-ZHF}
\begin{aligned}
\ZrbcHFzero &= \mathrm{i}\,\mathsf{k}_\mathrm{a} \,, \hspace*{2cm}
\ZrbcHFGb(r) = \mathrm{i}\,\mathsf{k}_\mathrm{a}\left(  1 +  \dfrac{\upeta_{\mathrm{a}}^\mathrm{G}(r) }{2\,\mathsf{k}_{\mathrm{a}}^2\, r}\, 
               + \dfrac{\tfrac{1}{4}  -\mu_{\ell}^2}{2\,\mathsf{k}_{\mathrm{a}}^2\,r^2} \right)\,, \\
\ZrbcSAIGzero(r) \, & = \, 
     \mathrm{i}\, \mathsf{k}_{\mathrm{a}}\, \left(1    +\dfrac{\upeta_{\mathrm{a}}^\mathrm{G}(r) }{r\,\mathsf{k}_{\mathrm{a}}^2 } \right)^{1/2} \,, \\
  \ZrbcSAIGa(r)\, & = \,  
     \mathrm{i}\, \mathsf{k}_{\mathrm{a}}\, \left(1   +\dfrac{\upeta_{\mathrm{a}}^\mathrm{G}(r) }{r\,\mathsf{k}_{\mathrm{a}}^2 } \right)^{1/2} 
     \left(1 \, +\, \dfrac{\tfrac{1}{4}  \, -\,\mu_{\ell}^2 }{2\left(\mathsf{k}_{\mathrm{a}}^2\, r^2 \,+ \,r\,\upeta_{\mathrm{a}}^\mathrm{G}(r)\right) } \right) \, ,\\
     \ZrbcSAIGb (r)\, & = \,  
     \mathrm{i}\, \mathsf{k}_{\mathrm{a}}\, \left(1   +\dfrac{\upeta_{\mathrm{a}}}{r\,\mathsf{k}_{\mathrm{a}}^2 } \right)^{1/2} 
     \left(1 \, +\, \dfrac{\tfrac{1}{4}  \, -\,(\mu^\mathrm{G}_{\ell}(r))^2 }{2\left(\mathsf{k}_{\mathrm{a}}^2\, r^2 \,+ \,r\,\upeta_{\mathrm{a}}\right) } \right)\,.
 \end{aligned}
\end{equation}
If we approximate the square root in $\mathcal{Z}_{\mathrm{SAIG0}}$, we obtain
\begin{equation}
\ZrbcHFGa(r) = \mathrm{i}\,\mathsf{k}_\mathrm{a}\left(  1 +  \dfrac{\upeta_{\mathrm{a}}^\mathrm{G}(r) }{2\,\mathsf{k}_{\mathrm{a}}^2\, r} \right)\,, \hspace*{0.5cm} \text{ assuming }  \hspace*{0.3cm} \left|\dfrac{\upeta_{\mathrm{a}}^\mathrm{G}(r) }{r\,\mathsf{k}_{\mathrm{a}}^2} \right| \ll 1 \,.
\end{equation}

%\begin{remark}[Radial ABC]
%We refer to these ABC  independent of $\ell$ as radial ABC. In addition those obtained in approximation with $\mathsf{O}(1)$ modulo, since 
%$\mu_\ell^2-\tfrac{1}{4} = 2$ at $\ell=0$ (instead of $0$ in the case for $\Vscalar$),
%radial ABC can be obtained by setting $\ell=0$ in those dependent on $\ell$, specifically in $\ZrbcHFGb$, $\ZrbcSAIGa$, $\ZrbcSAIGb$.
%We thus have the following radial ABCs,
%\begin{equation}
%\ZrbcHFzero, \quad\ZrbcSAIGzero, \quad \ZrbcHFGa, \hspace*{0.8cm} \mathcal{Z}_\mathrm{HFG2}^0 ,\quad \mathcal{Z}_\mathrm{SAIG1}^0 \,, \quad \mathcal{Z}_\mathrm{SAIG1B}^0\,.
%\end{equation}
%\end{remark}

\begin{remark}[Zero-gravity ABC]
Setting $G$ to zero in the above coefficients (equivalently by applying directly the procedure in \cite{barucq:hal-02168467,barucq:hal-02423882} to $Q^0_\ell$ \cref{Q0::def}),  we obtain `usual' HF and SAI family, 
 \begin{equation}
\ZrbcHFa    \,, \,\, 
\ZrbcHFb    \,, \,\, 
\ZrbcSAIzero\,, \,\, 
\ZrbcSAIa   \,.
 \end{equation}  
 
\end{remark}
\section{Accuracy of the numerical solutions}
\label{section:numerics:original-and-conjugated}
% ============================================================

The numerical experiments are carried out with the open-source 
software \texttt{hawen}\footnote{\url{https://ffaucher.gitlab.io/hawen-website/}}, 
\cite{Hawen2020} which implements the Hybridizable Discontinuous Galerkin (HDG) 
method (e.g., \cite{Cockburn2009,Cockburn2023}), that we already used in the scalar case, 
cf.~\cite[Section 3.1]{barucq2020efficient}.
For details on implementation of the first-order formulation with 
both the original operator $\Lorigin$ and the conjugate one $\Lorigin$ in \cref{notation_modalop}, we refer to 
\cite[Section~5 and Appendix A]{barucq:hal-03406855}.
A challenge posed by model \texttt{S-AtmoI} is the rapid decrease of density 
and wave-speed in the convective surface layer, \cite{barucq2020efficient}. 
To deal with this, we exploit the hp-adaptivity of 
the HDG method implemented in \texttt{hawen},  
which means that we refine the space discretization near solar surface, and 
allow for the parameters to vary within a mesh cell (e.g., \cite{Faucher2020adjoint}).
One advantage of the HDG method is to solve the first-order system without
increasing the computational cost (i.e., still working with one unknown which is
the numerical trace). In that way, we readily access to the fields $w$ and $v$ 
that are needed to assemble all of the directional kernels, \cite{barucq:hal-03406855}.
The investigations carried out in this section are twofold:
\vspace*{-0.50em}
\begin{enumerate}\setlength{\itemsep}{-2pt}
  \item We highlight the complementarity aspect of the original and conjugated equations
        in terms of numerical simulations: only the former is able to treat very 
        low attenuation, while only the latter is able to consider very large interval,
        cf.~\cref{subsection:original-conjugated_attenuation,subsection:original-conjugated_rmax}.
  \item We evaluate the performance of the absorbing boundary conditions for
        outgoing solutions, and highlight the importance of including the gravity 
        term, \cref{subsection:numerical-performance-rbc}.        
\end{enumerate}

% ==============================================================================
\subsection{Solutions of the conjugated and original equations with attenuation}
\label{subsection:original-conjugated_attenuation}
% ==============================================================================

We compare the solutions obtained with the first-order formulation 
of the original and conjugated equations, respectively 
denoted $\Lorigin$ and $\Lconjug$ in \cref{notation_modalop}. 
We solve the Dirichlet boundary value problem with $\rmax=\num{1.001}$, that is,
\begin{subequations} \label{eq:numeric:conjug-original_dirichlet}\begin{align}
\Lorigin(w_{\mathrm{org}},v_{\mathrm{org}}) = 0 \, \text{ on (0,$\rmax$)}\,\,\, & \text{;} \quad
           w_{\mathrm{org}}(\rmax)\,=\, 1 \,\,\, \text{;} \quad v_{\mathrm{org}}(0)\,=\, 0 \,. \\
\Lconjug(w_{\mathrm{cjg}},v_{\mathrm{cjg}}) = 0 \, \text{ on (0,$\rmax$)}\,\,\, & \text{;} \quad
           w_{\mathrm{cjg}}(\rmax)\,=\, 1 \,\,\, \text{;} \quad v_{\mathrm{cjg}}(0)\,=\, 0 \,.
\end{align} \end{subequations}
The relation between the two solutions $w$ is given by the change of variable in \cref{conjmodalODErhs},
%\begin{equation}
%  w_{\mathrm{cjg}}(r) \,=\,  \mathcal{K}(r,\rmax) \,w_{\mathrm{org}}(r) \,=\,
%  \dfrac{\mathfrak{I}_{\ell}(\rmax)}{\mathfrak{I}_{\ell}(r)} \,\, w_{\mathrm{org}}(r) \,.
%\end{equation}
\begin{equation}
  w_{\mathrm{cjg}}(r) \,=\,  \mathcal{K}(r,\rmax) \,w_{\mathrm{org}}(r) \,, \hspace*{0.5cm} \text{with} \hspace*{0.3cm}
   \mathcal{K}(r,\rmax) =  {\mathfrak{I}_{\ell}(\rmax)}/{\mathfrak{I}_{\ell}(r)} \,.
\end{equation}
% where the term $\mathfrak{I}_{\ell}(\rmax)$ comes as we use a Dirichlet 
% condition set to 1 for both cases in \cref{eq:numeric:conjug-original_dirichlet}.
In \cref{fig:conjug-original_attenuation-level}, we picture the solutions at 
frequency 7 \si{\milli\Hz} for modes 0 and 20, with different levels of 
attenuation $\Gamma$. 
% We also picture the solution $w_{\mathrm{cjg}}^{\mathrm{scalar}}$ that uses the scalar 
% approximation, see \cref{rk:scalar-approximation}.

% ------------------------------
\graphicspath{{figures/forward_dirichletBC/}}
% ------------------------------
\begin{figure}[ht!] \centering
  % ----------------------------------------------------------------------------------
  % 2D maps ==========================================================================
  \subfloat[][Conjugated solutions at (7\si{\milli\Hz}, $\ell=20$) with different levels of attenuation.]
             {\includegraphics[]{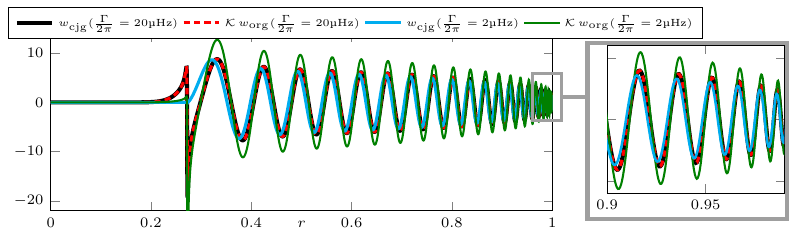}
              \label{fig:conjug-original_attenuation-level_a}              
              }
              \vspace*{-0.50em}
  \subfloat[][Original solutions at (7\si{\milli\Hz}, $\ell=20$) with different levels of attenuation.]
             {\includegraphics[]{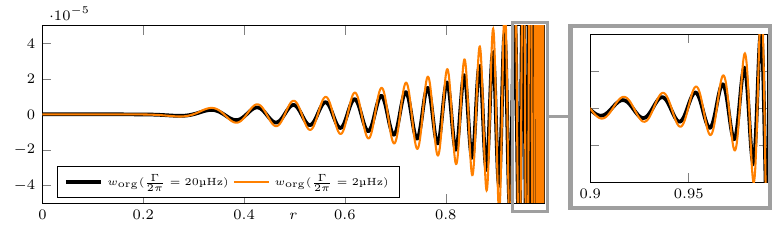}
              \label{fig:conjug-original_attenuation-level_b}}
              \vspace*{-0.50em}
  \subfloat[][Conjugated and original solutions at 
              (7\si{\milli\Hz}, $\ell=0$, $\tfrac{\Gamma}{2\pi}=\num{0.2}\si{\micro\Hz}$).]
             {\includegraphics[]{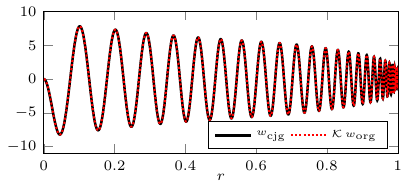}
              \includegraphics[]{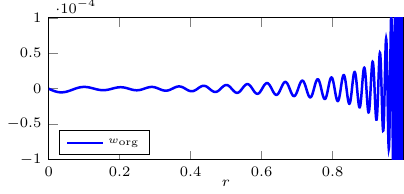}
              }
  \caption{Comparison between solutions of the original ($\Lorigin$)
           and conjugated ($\Lconjug$) equations \cref{eq:numeric:conjug-original_dirichlet}
           at frequency 7 \si{\milli\Hz} with different levels of
           attenuation.
           }
  \label{fig:conjug-original_attenuation-level}
\end{figure}

We have the following observations.
\vspace*{-0.50em}
\begin{itemize} \setlength{\itemsep}{-2pt}
\item
We see that the wavelength decreases as one approaches $r=1$, 
this behaviour is expected due to the rapid decrease of the solar parameters in surface layers.
The solutions of the original equation (\cref{fig:conjug-original_attenuation-level_b})
show large variations in magnitude, more than 
four orders between the surface (where we impose 
Dirichlet condition to 1), and the interior where 
we have forced the scale for visualization.
On the other hand, the conjugated solutions remain in 
the same magnitude for the entire interval.
\item
For relatively high attenuation (\num{20}\si{\micro\Hz}), 
the solutions $w_{\mathrm{cjg}}$ and $\mathcal{K} \,w_{\mathrm{org}}$ 
coincide well. 
For lower attenuation, \num{2}\si{\micro\Hz}, the original solution 
remains stable as only the amplitude is changed between the different
levels of attenuation.
However, the conjugated equation appears unstable 
with phase shift in the signal for $\ell\neq0$ 
in \cref{fig:conjug-original_attenuation-level_a}.
% However, for low attenuation, the conjugated solution becomes unstable
% for modes $\ell \neq 0$ while the original solution is stable.
This is due to the solution of the $\Lorigin$ which is smooth 
at $r^\star_{\ell,\omega}$~\cref{Lambsing::def} 
(having positive indicial roots here) as opposed 
to $\Lconjug$ which has one local negative indicial root, 
cf. \cite[Remark p.25]{barucq:hal-03406855}.
This singularity is due to the zero of $\Fl$ that comes 
from Lamb frequencies, and does not exist for mode $\ell=0$, 
cf. bottom line of \cref{fig:conjug-original_attenuation-level}.
\end{itemize}
Consequently, despite the high variations in magnitude in the solutions,
the original problem is the only candidate able to handle low levels of 
attenuation.

\subsection{Non-local condition with different positions of $\rmax$}
\label{subsection:original-conjugated_rmax}
% ==========================================================================================

We now solve the equations in \cref{eq:numeric:conjug-original_dirichlet}, 
with a Dirac source at height $s=1$: 
\begin{subequations} \label{eq:numeric:conjug-original_dirac}\begin{align}
\Lorigin(w_{og},v_{og}) = \delta(r-1) \, \text{ on (0,$\rmax$)} \,\,\, & \text{;}\hspace*{0.3cm}
           v_{og} - \hat{\Zrbc}\, w_{og} =\! |_{\rmax} \hspace*{0.1cm} 0\, \text{;} & \hspace*{0.3cm}
           v_{og}(0)\,=\, 0 \,. \\
\Lconjug(w_{cj},v_{cj}) = \delta(r-1) \, \text{ on (0,$\rmax$)}\,\,\, & \text{;} \hspace*{0.3cm}
           v_{cj} - \Zrbc\, w_{cj} =\!|_{\rmax}  \hspace*{0.1cm} 0\, \text{;} & \hspace*{0.3cm}v_{cj}(0)\,=\, 0 \,.
           \label{eq:numeric:conjug_dirac}
\end{align} \end{subequations}
The relation between the RBC of the original equation $\hat{\Zrbc}$ and
of the conjugated one $\Zrbc$ is given in \cref{eq:rbc-equivalent:origin-conj}.
To evaluate the performance in the choice of boundary conditions, numerical 
solutions are compared with a reference solution 
$r\mapsto w^\mathrm{ref}(r;\omega,\ell)$ that is computed on a much larger interval.
For each frequency and each mode $(\omega,\ell)$, we introduce the error function
$\err(\omega,\ell)$ which shows the relative difference between the numerical solution 
$r\mapsto w^{\Zrbc}(r;\omega,\ell)$ obtained with ABC $\Zrbc$ 
and the reference solution $r\mapsto w^\mathrm{ref}(r;\omega,\ell)$ 
on the interval of interest $(0,\num{1.001})$, i.e.,
\begin{equation}\label{eq:error:norm}
  \err_{\Zrbc}(\omega,\ell) \,=\, \dfrac{\Vert w^\mathrm{ref}(\cdot;\omega,\ell) \,-\, w^{\Zrbc}(\cdot;\omega,\ell) \Vert_{(0,\num{1.001})}}
                                {\Vert w^\mathrm{ref}(\cdot;\omega,\ell) \Vert_{(0,\num{1.001})}} \,.
                                % \qquad\quad \text{for $r\in(0,\num{1.001})$.}
\end{equation} 
%In \cref{fig:conjug_rmax}, we plot the solutions at frequency 
%7\si{\milli\Hz} and \num{0.2}\si{\milli\Hz} using the nonlocal 
%condition $\ZrbcNL$ imposed at different positions of $\rmax$. 
%We also zoom on near $r=1$ for better visualization.

In \cref{fig:conjug_rmax}, we plot the solutions at frequency 
7\si{\milli\Hz} and \num{0.2}\si{\milli\Hz} using the nonlocal 
condition $\ZrbcNL$ imposed at different positions of $\rmax$,
and we use attenuation $\Gamma/(2\pi)=20$\si{\micro\Hz}. 
While the conjugated equation is able to handle large 
$\rmax$ (e.g., $\rmax=10$ below),
the numerical discretization with $\Lorigin$ is unstable
for $\rmax > 1.1$ (position depending on frequency 
and modes), in which case the matrix becomes singular, or the
solver returns `NaN' values.
In another word, by working with original problem, we are unable to 
handle arbitrarily large value of $\rmax$. 
This could be explained from the high variation of the 
solutions in the atmosphere observed in 
\cref{fig:conjug-original_attenuation-level} leading to 
numerically overwhelming arithmetic when increasing the 
position of $\rmax$. 
% In fact, discrepancies are visible at 7 \si{\milli\Hz} right after
% $r=1$ in \cref{fig:conjug_rmax_original} while it seems that frequencies 
% below cut-off are more stable (right of \cref{fig:conjug_rmax_original}).
% On the other hand, the conjugated equation is stable and all 
% solutions visually match.
To summarize, we can highlight that the two formulations are 
\emph{numerically complementary} and that they each have their range
of applications: \vspace*{-0.6em}
\begin{itemize} \setlength{\itemsep}{-2pt}
  \item The conjugated equation \emph{cannot} handle low attenuation 
        but \emph{can} consider large $\rmax$.
  \item The original equation \emph{can} handle low attenuation but
        \emph{cannot} consider large $\rmax$.
\end{itemize}

% ------------------------------
\graphicspath{{figures/forward_rbc/}}
% ------------------------------
\begin{figure}[ht!] \centering
  % ----------------------------------------------------------------------------------
  % 2D maps ==========================================================================
  \subfloat[][(7\si{\milli\Hz}, $\ell=20$).]
             {\includegraphics[]{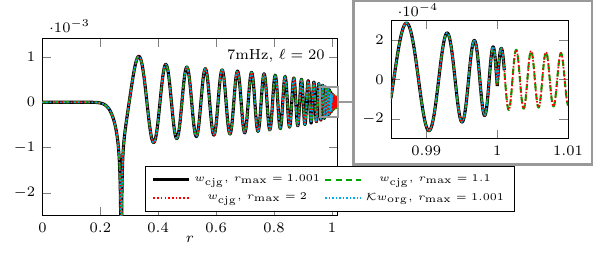}}
  \subfloat[][(0.2\si{\milli\Hz}, $\ell=200$), zoom near surface.]
             {\includegraphics[]{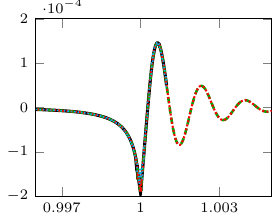}}

  \caption{
           Solutions of the conjugated equations with 20\si{\micro\Hz} attenuation 
           using different positions of $\rmax$ where the nonlocal condition $\ZrbcNL$ 
           is imposed.
           The original equation cannot be computed for large $r$ as it leads 
           to singular matrix or NaN entries.
           }
  \label{fig:conjug_rmax}
\end{figure}

We evaluate the accuracy of the solutions depending on the  position 
of $\rmax$ in \cref{fig:rbc-with-rmax}, for different frequencies and 
modes using only the nonlocal condition $\ZrbcNL$, with the
conjugated equation. 
% For the conjugated equation (\cref{fig:rbc-with-rmax_conjugated}),
The reference solution is computed using $\rmax=\num{10}$ and we compare 
frequencies below and above cut-off $\omega_c$ \cref{awc::def}. 
% For the original equation (\cref{fig:rbc-with-rmax_original}),
% due to above-mentioned instability for large $\rmax$, in particular
% observed for frequencies above cut-off, the reference solution is 
% computed with $\rmax=\num{1.2}$ and we restrict the analysis to 
% frequencies below cut-off, mostly to verify the consistency between 
% the two equations.
In this experiment, we use attenuation 20\si{\micro\Hz}, and evaluate
the relative error $\err$ of \cref{eq:error:norm}, that is, whatever
$\rmax$ is used for the simulations, the error only considers the fields
up to \num{1.001} to have fair comparisons.
We see that below the cut-off frequency, the nonlocal condition 
is accurate for low $\rmax$, except for frequencies below the 
Lamb frequency (i.e., low-frequency/large mode),
where the potential is negative (cf. \cref{fig:potential:main01,fig:potential:main02})
and $\rmax$ needs to be slightly pushed further out.
Above the cut-off frequency, we have a steady decrease of the error with increasing $\rmax$, 
which eventually needs to be relatively far to reach the accuracy obtained below cut-off. 
Nonetheless, using $\rmax=\num{1.001}$  already provides a good accuracy everywhere with the
condition $\ZrbcNL$, as the relative error is always below \num{e-4}. 
% The results obtained with original equation (only below cut-off) are consistent, using a 
% reference solution where $\rmax=\num{1.2}$. 
% In the following experiments, the conjugated equation will be used to ensure the reference
% solution is computed on a large interval, and we maintain the attenuation at \num{20}\si{\micro\Hz}.

% ===================================
\setlength{\plotheight}{2.30cm}
\newcommand{\infofreq}{}
% \pgfplotsset{minor grid style = {dashed, Gray}}
% \pgfplotsset{major grid style = {dotted, Gray}}
% ===================================
\begin{figure}[ht!]\centering
  \renewcommand{\myylabel}{$\err$}
  \renewcommand{\datafile}{figures/forward_rbc/data_error_rmax/data_conjugated_equation_att20uHz_norm.txt}
  \pgfkeys{/pgf/fpu=true} \pgfmathsetmacro{\ymin}{5e-16} \pgfmathsetmacro{\ymax}{1e-3} \pgfkeys{/pgf/fpu=false}
  \pgfkeys{/pgf/fpu=true} \pgfmathsetmacro{\xmin}{0.998}  \pgfmathsetmacro{\xmax}{1.051} \pgfkeys{/pgf/fpu=false}
  \renewcommand{\xtickloc}{1.01,1.1,1.2,1.4}
  \renewcommand{\dataA}{A}\renewcommand{\legendA}{$\ell=0$}
  \renewcommand{\dataB}{B}\renewcommand{\legendB}{$\ell=20$}
  \renewcommand{\dataC}{C}\renewcommand{\legendC}{$\ell=200$}
  \renewcommand{\dataD}{D}\renewcommand{\legendD}{$\ell=500$}
  \renewcommand{\myfreq}{0.2mHz}
  \setlength{\plotwidth} {3.50cm}
  \renewcommand{\infofreq}{\num{0.2}\si{\milli\Hz}}  
  % \subfloat[][Relative error $\err_{\ZrbcNL}$ for frequencies 
  %            \num{0.2}, \num{2}, and \num{7}\si{\milli\Hz}.
  %             The reference solution is computed with the conjugated 
  %             equation using $\rmax=\num{10}$, the numerical ones 
  %             with $\rmax\in(\num{1.001},\num{1.5})$.] 
  {\makebox[2.40cm][l]{ % \input{figures/forward_rbc/skeleton_ploterr_01left} 
                       \includegraphics[scale=1]{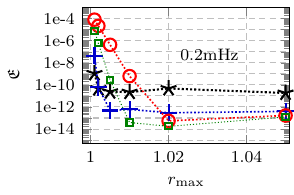}
                      }
   \renewcommand{\myfreq}{2.0mHz}
   \renewcommand{\infofreq}{\num{2}\si{\milli\Hz}}  
   \makebox[5.2cm][l]{ % \input{figures/forward_rbc/skeleton_ploterr_01mid}
                      \includegraphics[scale=1]{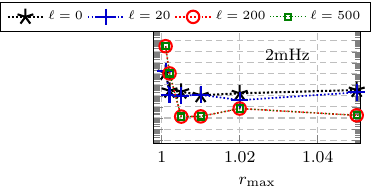}
                     }
   \pgfkeys{/pgf/fpu=true} \pgfmathsetmacro{\xmin}{0.98}  
   \pgfmathsetmacro{\xmax}{1.51} \pgfkeys{/pgf/fpu=false}
   \setlength{\plotwidth} {5.00cm}
   \renewcommand{\myfreq}{7.0mHz}
   \renewcommand{\infofreq}{\num{7}\si{\milli\Hz}}  
   \makebox[5.50cm][l]{% \input{figures/forward_rbc/skeleton_ploterr_01right}
                       \includegraphics[scale=1]{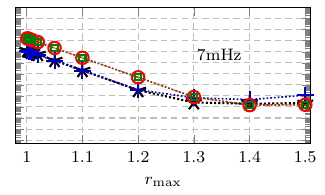} 
                       }
  \label{fig:rbc-with-rmax_conjugated}}

  % ======================================
  % \renewcommand{\datafile}{figures/forward_rbc/data_error_rmax/data_original_equation_att20uHz_w11.txt}
  % \renewcommand{\datafile}{figures/forward_rbc/data_error_rmax/data_original_equation_att20uHz_norm.txt}
  % \pgfkeys{/pgf/fpu=true} \pgfmathsetmacro{\xmin}{0.998}  \pgfmathsetmacro{\xmax}{1.10} \pgfkeys{/pgf/fpu=false}
  % \renewcommand{\xtickloc}{1.01,1.05,1.1}
  % \subfloat[][Relative error $\err_{\ZrbcNL}$ for frequencies 
  %             \num{0.2} and \num{2}\si{\milli\Hz}.
  %             The reference solution is computed with the original
  %             equation using $\rmax=\num{1.2}$, the numerical ones 
  %             with $\rmax\in(\num{1.001},\num{1.1})$.]{
  %  \renewcommand{\myfreq}{0.2mHz}
  %  \renewcommand{\infofreq}{\num{0.2}\si{\milli\Hz}}  
  %  \setlength{\plotwidth}{4.50cm}
  %  \makebox[6.0cm][l]{\input{figures/forward_rbc/skeleton_ploterr_01left}}
  %  \renewcommand{\myfreq}{2.0mHz}
  %  \renewcommand{\infofreq}{\num{2}\si{\milli\Hz}}
  %  \makebox[8.5cm][l]{\input{figures/forward_rbc/skeleton_ploterr_02right}}
  %  \label{fig:rbc-with-rmax_original}}

  \caption{Relative error $\err_{\ZrbcNL}$ of \cref{eq:error:norm}
           as a function of the position where the nonlocal boundary
           condition $\ZrbcNL$ is imposed, for frequencies 
           \num{0.2}, \num{2}, and \num{7}\si{\milli\Hz}. 
           The reference solution is computed with the conjugated 
           equation using $\rmax=\num{10}$, the numerical ones 
           with $\rmax\in(\num{1.001},\num{1.5})$.
           All simulations use an attenuation of \num{20}\si{\micro\Hz}.}
  \label{fig:rbc-with-rmax}
\end{figure}

% ============================================================
\subsection{Performance of the approximate nonlocal boundary conditions}
\label{subsection:numerical-performance-rbc}
% ============================================================

We compare the performance of the boundary conditions 
constructed in \cref{section:rbc}. 
This experiment is carried out with the conjugated equation and
attenuation 20\si{\micro\Hz}.
We employ the following solution as reference solution,
\begin{equation}
w^\mathrm{ref}(r;\omega,\ell) \text{ solves }
 \cref{eq:numeric:conjug_dirac} \text{ with } \rmax = 10 \text{ and } \mathcal{Z} = \ZrbcNL\,.
\end{equation}
In \cref{figure:rbc-error:NL01}, we evaluate the performance of 
conditions in the nonlocal family, that is, conditions 
$\ZrbcNL$, $\ZrbcNLa$, and $\ZrbcNLaG$ of \cref{zNLG::def}. 
We show the relative 
error $\err$ \cref{eq:error:norm} for modes between 0 and 200, 
for three selected frequencies.

% -------------------------------------------
\setlength{\plotwidth} {4.60cm}
\setlength{\plotheight}{2.50cm}
% -------------------------------------------
\begin{figure}[ht!]\centering
  \renewcommand{\myylabel}{$\err$}
  \pgfkeys{/pgf/fpu=true} \pgfmathsetmacro{\ymin}{5e-14}\pgfmathsetmacro{\ymax}{1e-3}  \pgfkeys{/pgf/fpu=false}
  \pgfkeys{/pgf/fpu=true} \pgfmathsetmacro{\xmin}{0}     \pgfmathsetmacro{\xmax}{200}  \pgfkeys{/pgf/fpu=false}
  \renewcommand{\xtickloc}{0,50,100,150}
  \renewcommand{\myfreqA}{0.2mHz}
  \renewcommand{\myfreqB}{2mHz}
  \renewcommand{\myfreqC}{7mHz}
  \renewcommand{\datafile}{figures/forward_rbc_all/data_conjugated/relative-error-meanscaled_ref-Znonlocal10.000_num1.001_w.txt}
  \renewcommand{\dataA}{Znonlocal}        \renewcommand{\legendA}{$\ZrbcNL$,   $\rmax=\num{1.001}$}
  \renewcommand{\dataB}{Znonlocal_approx} \renewcommand{\legendB}{$\ZrbcNLa$,  $\rmax=\num{1.001}$}
  \renewcommand{\dataC}{Znonlocal_approxG}\renewcommand{\legendC}{$\ZrbcNLaG$, $\rmax=\num{1.001}$}
  \renewcommand{\datafileB}{figures/forward_rbc_all/data_conjugated/relative-error-meanscaled_ref-Znonlocal10.000_num1.050_w.txt}
  \renewcommand{\dataD}{Znonlocal} \renewcommand{\legendD}{$\ZrbcNL$, $\rmax=\num{1.05}$}
  \renewcommand{\datafileC}{} \renewcommand{\dataE}{} 
  \includegraphics[scale=1]{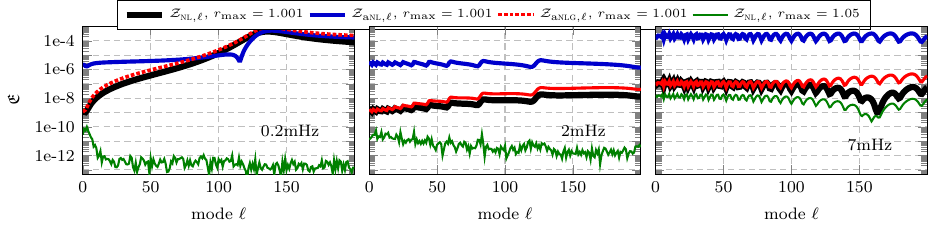}

  % ==================================================================================
  \caption{Relative error $\err_{\Zrbc}$ of \cref{eq:error:norm}
           as a function of $\ell$
           for the nonlocal conditions \cref{zNLG::def} set in 
           $\rmax=\num{1.001}$ (except $\rmax=\num{1.05}$ for the green curve).
           The reference solution uses the nonlocal condition 
           $\ZrbcNL$ in $\rmax=\num{10}$.
           % Simulations uses the conjugated equation with 20\si{\milli\Hz}
           % attenuation and modes between $\ell=0$ and $\ell=200$; 
           We compare
           frequency \num{0.2} (left), \num{2} (middle) 
           and \num{7} \si{\milli\Hz} (right).
           }
  \label{figure:rbc-error:NL01}
\end{figure}
% -----------------------------------------------------

We see that including the gravity term $G$ in the approximation 
of $V_\ell$ (i.e., comparing $\ZrbcNLaG$ with $\ZrbcNLa$) drastically
improves the accuracy, with a more considerable gain at 
7\si{\milli\Hz} in \cref{figure:rbc-error:NL01}, with an improve in 
accuracy of three to four orders of magnitude.
The relative error remains relatively stable with the modes
for 2 and 7\si{\milli\Hz} while at \num{0.2}\si{\milli\Hz}, 
the error increases with increasing modes, and is the largest
in the low frequency and high mode regions.
In \cref{figure:rbc-error:all01,figure:rbc-error:all02}, 
we now evaluate the performance of the conditions in the HF 
and SAI families of \cref{eq:rbc:ZSAI-ZHF}.
In particular \cref{figure:rbc-error:all02} shows a cartography
of the relative error for modes between 0 to 200, and frequencies from \num{0.01} to 10 \si{\milli\Hz}.

% -------------------------------------------
\begin{figure}[ht!]\centering
  \renewcommand{\myylabel}{$\err$}
  \renewcommand{\datafile}{figures/forward_rbc_all/data_conjugated/relative-error-meanscaled_ref-Znonlocal10.000_num1.001_w.txt}
  \pgfkeys{/pgf/fpu=true} \pgfmathsetmacro{\ymin}{1e-14} \pgfmathsetmacro{\ymax}{2.0} \pgfkeys{/pgf/fpu=false}
  \pgfkeys{/pgf/fpu=true} \pgfmathsetmacro{\xmin}{0}     \pgfmathsetmacro{\xmax}{200} \pgfkeys{/pgf/fpu=false}
  \renewcommand{\xtickloc}{0,50,100,150}
  \renewcommand{\myfreqA}{0.2mHz}
  \renewcommand{\myfreqB}{2mHz}
  \renewcommand{\myfreqC}{7mHz}
  \renewcommand{\dataA}{ZSHF0} \renewcommand{\legendA}{$\ZrbcHFzero$}
  \renewcommand{\dataB}{ZSHF1} \renewcommand{\legendB}{$\ZrbcHFa$}
  \renewcommand{\dataC}{ZSHF2} \renewcommand{\legendC}{$\ZrbcHFb$}
  \renewcommand{\dataD}{ZSHFG1}\renewcommand{\legendD}{$\ZrbcHFGa$}
  \renewcommand{\dataE}{ZSHFG2}\renewcommand{\legendE}{$\ZrbcHFGb$}
  \renewcommand{\dataF}{}
  \renewcommand{\dataG}{}
  \renewcommand{\datafileB}{figures/forward_rbc_all/data_conjugated/relative-error-meanscaled_ref-Znonlocal10.000_num1.010_w.txt}
  \renewcommand{\dataH}{ZSHFG2} \renewcommand{\legendH}{$\ZrbcHFGb\,\rmax=\num{1.01}$}
  \subfloat[][HF family.]{%\input{figures/forward_rbc_all/skeleton_plot3freq_03b} 
                          \includegraphics[scale=1]{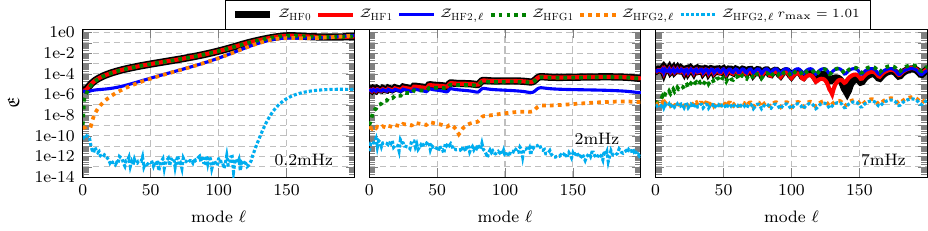}
                         }

  \renewcommand{\dataA}{ZSAI0}   \renewcommand{\legendA}{$\ZrbcSAIzero$}
  \renewcommand{\dataB}{ZSAI1}   \renewcommand{\legendB}{$\ZrbcSAIa$}
  \renewcommand{\dataD}{ZSAIG0}  \renewcommand{\legendD}{$\ZrbcSAIGzero$}
  \renewcommand{\dataE}{ZSAIG1}  \renewcommand{\legendE}{$\ZrbcSAIGa$}
  \renewcommand{\dataC}{}        \renewcommand{\legendC}{}
  \renewcommand{\dataF}{ZSAIG1b} \renewcommand{\legendF}{$\ZrbcSAIGb$}
  \renewcommand{\dataG}{}        \renewcommand{\legendG}{}
  \renewcommand{\datafileB}{figures/forward_rbc_all/data_conjugated/relative-error-meanscaled_ref-Znonlocal10.000_num1.010_w.txt}
  \renewcommand{\dataH}{ZSAIG1} \renewcommand{\legendH}{$\ZrbcSAIGa\,\rmax=\num{1.01}$}
  \subfloat[][SAI family.]{
                          \includegraphics[scale=1]{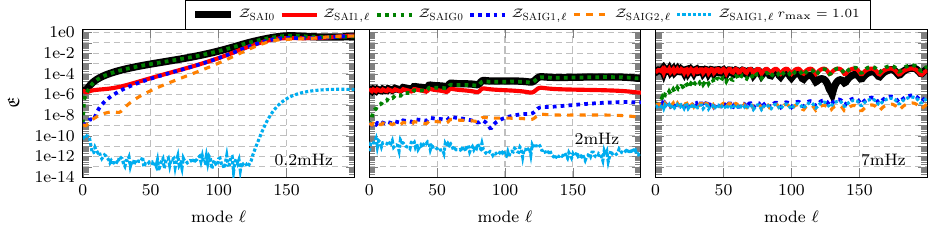}
                          }

  % ==================================================================================
  \caption{Relative error $\err_{\Zrbc}$ of \cref{eq:error:norm}
           for the HF and SAI conditions \cref{eq:rbc:ZSAI-ZHF}
           set at $\rmax=\num{1.001}$ except for the blue curve ($\rmax=\num{1.01}$).
           The reference solution uses the nonlocal condition 
           $\ZrbcNL$ in $\rmax=\num{10}$.
           % Simulations uses the conjugated equation with 20\si{\milli\Hz}
           % attenuation and modes between $\ell=0$ and $\ell=200$; 
           We compare
           frequency \num{0.2} (left), \num{2} (middle) 
           and \num{7} \si{\milli\Hz} (right).
           }
  \label{figure:rbc-error:all01}
\end{figure}
% -----------------------------------------------------

% -----------------------------------------------------
\setlength{\modelwidth} {4.20cm}
\setlength{\modelheight}{4.00cm}
% -----------------------------------------------------
\graphicspath{{figures/forward_rbc_all/maps_conjugated_0.1-to-10mHz_mode0-to-200_att20uHz_refZnlc-10.000_rmax1.001_scale-9to-1/}}
\begin{figure}[ht!]\centering

  \renewcommand{\cbtitle}{$\err$} 
  \pgfkeys{/pgf/fpu=true} 
  \pgfmathsetmacro{\cmin}{1e-9} \pgfmathsetmacro{\cmax}{1e-1} 
  \pgfkeys{/pgf/fpu=false}
  \renewcommand{\modelfile}{norml2err_ZSAI0}
  \subfloat[][$\ZrbcSAIzero$]{% \input{figures/forward_rbc_all/skeleton_map_left}
                              \includegraphics[scale=1]{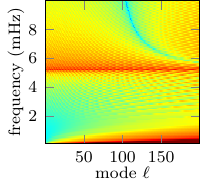}
                              } \hspace*{-1em}
  \renewcommand{\modelfile}{norml2err_ZSAIG0}
  \subfloat[][$\ZrbcSAIGzero$]{% \input{figures/forward_rbc_all/skeleton_map_mid}
                               \includegraphics[scale=1]{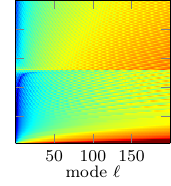}}  \hfill
  \renewcommand{\modelfile}{norml2err_ZSAI2}
  \subfloat[][$\ZrbcSAIa$]{% \input{figures/forward_rbc_all/skeleton_map_mid}
                           \includegraphics[scale=1]{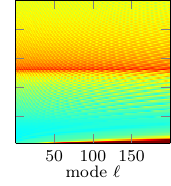}} \hspace*{-1em}
  \renewcommand{\modelfile}{norml2err_ZSAIG1b}
  \subfloat[][$\ZrbcSAIGb$]{% \input{figures/forward_rbc_all/skeleton_map_mid}
                            \includegraphics[scale=1]{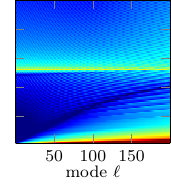}} \hspace*{-1em}
  \renewcommand{\modelfile}{norml2err_ZSHFG2}
  \subfloat[][$\ZrbcHFGb$]{% \input{figures/forward_rbc_all/skeleton_map_right}
                           \includegraphics[scale=1]{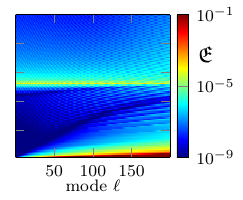}}

  % ==================================================================================
  \caption{Relative error $\err_{\Zrbc}$ of \cref{eq:error:norm}
           for the HF and SAI conditions \cref{eq:rbc:ZSAI-ZHF}
           set at $\rmax=\num{1.001}$.
           The reference solution uses the nonlocal condition 
           $\ZrbcNL$ in $\rmax=\num{10}$.
           % Simulations uses the conjugated equation with 20\si{\milli\Hz}
           % attenuation and modes between $\ell=0$ and $\ell=200$; 
           We show the cartography for
           frequencies between \num{0.1} and \num{10} \si{\milli\Hz}, 
           and mode between \num{0} and \num{200}.
           }
  \label{figure:rbc-error:all02}
\end{figure}
% -----------------------------------------------------

\medskip

\noindent We have the following observations.
\vspace*{-0.50em}
\begin{itemize}\setlength{\itemsep}{-2pt}

  \item The improvement of including the gravity in the condition
        is maintained in the HF and SAI family, especially for 
        conditions that are mode-dependent, e.g., comparing 
        $\ZrbcHFb$ with $\ZrbcHFGb$ or comparing 
        $\ZrbcSAIa$ with $\ZrbcSAIGa$.
        For the conditions that do not depend on $\ell$, we still
        see an improvement, in particular for low modes, e.g., 
        comparing $\ZrbcSAIzero$ and $\ZrbcSAIGzero$ in 
        \cref{figure:rbc-error:all02}.

   \item Comparing the HF and SAI families, the two are 
        very close in terms of accuracy, with the best
        choice being condition $\ZrbcSAIGb$.
        
  \item The approximate conditions in the HF and SAI provide an 
        acceptable relative error (below \num{e-4}) except for 
        frequencies below the Lamb frequency $\mathcal{S}_\ell$ 
        (low frequencies/high modes zone),
        see \cref{figure:rbc-error:all01}. 
        This corresponds to regions where $V_\ell$
        is poorly approximated by $Q_{\ell}^\mathrm{G}$, cf. \cref{ApproxVell::subsec},
        and where condition \cref{smallquant::assu} for square-root approximation fails.
        Here, the error largely increases and these 
        conditions should only be used if $\rmax$ is 
        sufficiently large (e.g., in \num{1.01} in 
        \cref{figure:rbc-error:all01}).
\end{itemize}

% ============================================================
\section{Numerical experiments of solar observables}
\label{section:observables}
% =============================================================

%In this section we carry out numerical experiments of solar 
%observables by first relating the directional Green's kernels
%to the power spectrum that can be observed by helioseismology in
%\cref{subsection:connection-observables}. Then we compute
% power spectra to highlight the surface-gravity modes and the pressure modes in \cref{subsection:solar-greens-kernels}  and the internal gravity modes in \cref{subsection:g-modes}. To emphasize the importance of gravity, we compare this spectra to one the case where gravity is ignored which consists in setting $\phi_0' = 0$ and $E_{\rm he} = 0$ in \cref{eq:auxiliary-sch-main}. This is equivalent to considering only the first two terms in \cref{main_eqn_Galbrun} which was done to derive the scalar wave equation of \cite{gizon2017computational}. However the scalar equation from \cite{gizon2017computational} (also denoted scalar in this work) has an unknown related to the divergence of the displacement while in the framework of this paper the unknown is the radial displacement $\xi_r$. To ease comparison, we choose to compute power spectra on the vector equation with and without gravity so that they both represent the same quantity.

% where we also illustrate the importance of accurate boundary
% conditions.
% Contrary to the scalar approximation \cite{barucq2020efficient},
% including the gravity allows to access the g-modes, which we
% investigate in \cref{subsection:g-modes}.

In this section, to demonstrate the utility of our results towards helioseismology,
 we show that synthetic solar 
observables in form of power spectrum produced with our methods capture correctly the physical modes. We compare the locations of maximum power of the ridges of our power spectrum to the eigenvalues obtained by the eigenvalue solver \texttt{gyre} \cite{Townsend2013} and to the measurements from \cite{Larson2015}\footnote{The measured frequencies are available as supplementary electronic material from \cite{Larson2015} and can be downloaded at \url{http://sun.stanford.edu/~schou/anavw72z/av}. The first column corresponds to the harmonic degrees $\ell$ and the third one to the frequencies in $\mu$Hz.} for low-degree modes ($\ell < 300$) and from \cite{Korzennik2013} for high-degree modes.
%by fitting the observed power spectrum computed from MDI dopplergrams \cite{Scherrer1995}.
We start by relating the directional  Green's kernels $G^\mathrm{PP}_\ell$
to the power spectrum that can be observed by helioseismology in
\cref{subsection:connection-observables}. We then implement \cref{algorithm:assemble} to compute all three 
directional kernels for all position of source and receivers in the Sun in \cref{subsection:solar-greens-kernels}. 
Employing relation \cref{synPS}, synthetic power spectra are extracted from the full kernel and are shown to display surface-gravity ridges and  pressure ridges in \cref{subsection:solar-greens-kernels},  and  internal gravity ridges in \cref{subsection:g-modes}. To emphasize the importance of gravity, we compare these spectra to those obtained in zero-gravity approximation. This is the equation obtained from \cref{notation_modalop} by setting
\begin{equation}\label{eq:zero-gravity_approximation}
  \phi_0'=0, \,\, \Ehe=0 \quad \text{ in \cref{sfFell::def}--\cref{eq:auxiliary-sch-main} } \hspace*{1cm}  \substack{\text{zero-gravity}\\\text{approximation}}\,.
\end{equation}

\begin{remark} The approximation in \cref{eq:zero-gravity_approximation} 
is equivalent to considering only the first two terms in \cref{main_eqn_Galbrun} 
which was done to derive the scalar wave equation of \cite{gizon2017computational}, 
see also \cref{scalarSODE}.
However, the unknown of the scalar equation is related to the divergence 
of the displacement $\xib$, while in the framework of this paper the unknown 
is the radial displacement $\xi_r$. For this reason,
we compare spectra associated with the unknown of the original ODE \cref{notation_modalop} 
and its version in approximation \cref{eq:zero-gravity_approximation} instead of the solution of \cref{scalarSODE}.
\end{remark}

% =============================================================
\subsection{Connection to the observed solar power spectrum}
\label{subsection:connection-observables}
% =============================================================

We suppose that the observations at the points $\hat{\bx}$ 
on the surface of the Sun
correspond to the line-of-sight velocity 
$\psi(\hat{\bx},\omega)$ at a specific height $r_{\rm obs}$, 
\begin{equation}
  \psi(\hat{\bx},\omega) = -\ii \omega \hat{\mathbf{l}}(\hat{\bx}) \cdot \xib(r_{\rm obs}, \hat{\bx},\omega)\,.
\end{equation}
Here, we will employ the simple but commonly used assumption that the line-of-sight $\hat{\mathbf{l}}$ is purely radial, i.e.
 \begin{equation}\hat{\mathbf{l}}(\hat{\bx}) = \mathbf{e}_r \,. \label{eq:los_radial}
 \end{equation}
The power spectrum $\mathcal{P}_l^m$ is defined from the spherical harmonic coefficients of $\psi$ with
\begin{equation}
  \mathcal{P}_l^m = \mathbb{E}\left[ |\psi_l^m|^2 \right], \hspace*{0.5cm} 
  \text{where} \hspace*{0.3cm}  \psi_l^m = \int \psi(\hat{\bx}) \overline{Y_l^m}(\hat{\bx}) \mathrm{d}\bx = -\ii \omega a_l^m\, ,
\end{equation}
and coefficients $a_l^m$ are the modal solution from \cref{coeffVHS_f_xi}.
The power spectrum can then be computed as
\begin{align}
 \mathcal{P}_l^m = \omega^2 \Rsun^2 \int \int \overline{G_l^{PP}(r,s)} G_l^{PP}(r,s') \mathbb{E}[\overline{f_l^m(s)} f_l^m(s')]  s^2 {s'}^2 ds ds' \nonumber \\
 + \omega^2 \Rsun^2 \int \int \overline{G_l^{PB}(r,s)} G_l^{PB}(r,s') \mathbb{E}[\overline{g_l^m(s)} g_l^m(s')]  s^2 {s'}^2 ds ds'.
\end{align}
It remains to define the expected values of the 
source of excitation of the waves, 
$\mathbb{E}[\overline{f_l^m(s)} f_l^m(s')]$ and $\mathbb{E}[\overline{g_l^m(s)} g_l^m(s')]$. 
We suppose that the sources are uncorrelated, purely radial and coming from a single 
depth $r_s$ (see, e.g., \cite{boning2016sensitivity}), so that
\begin{equation}
 \mathbb{E}[\overline{f_l^m(s)} f_l^m(s')] = \frac{1}{s^2} \delta(s-s') \delta(s-r_s), \hspace*{0.5cm} \mathbb{E}[\overline{g_l^m(s)} g_l^m(s')] = 0.
\end{equation}
In this case, the power spectrum is directly linked 
to the Green's kernel $\GPP$: 
\begin{equation}\label{synPS}
 \mathcal{P}_l^m(r_{\rm obs}) = \omega^2 \Rsun^2 \vert \GPP(r_{\rm obs}, r_s) \vert^2.
\end{equation}

\begin{remark}
The kernel $\GPB$ needs to be used if sources in the horizontal 
direction are kept, and the kernels $\GBB$ and $\GBP$ appear if 
one uses a general expression for the line-of-sight instead of the purely radial case \cref{eq:los_radial}.
\end{remark}

% =============================================================
\subsection{Surface gravity and pressure modes in the solar power spectrum}
\label{subsection:solar-greens-kernels}
\graphicspath{{figures/kernels/}}
% =============================================================

For each mode and frequency $(\ell,\,\omega)$,  
following \cref{algorithm:assemble}, we perform two simulations
from which all the directional Green's kernels can be obtained, 
for any position of source $s$ and receiver $r$. 
% Following \cref{algorithm:assemble}, we compute the solar Green's kernels 
% with two simulations for each choice of mode and frequency. 
From these two simulations, we can then assemble any of the directional 
Green's kernel for any position of source and receiver $s$ and $r$. We use the solar background model \texttt{S-AtmoI} with a constant attenuation of 2\si{\micro\Hz} and solve the original equation as it leads to better numerical stability at relatively low attenuation. 
We compute the kernels for source and receiver at 1. In the framework of the previous section, the source position corresponds to the depth $r_s$ which is usually assumed a few hundred kilometers below the surface while the receiver corresponds to the observation height located slighly above the surface. In \cref{figure:spectrum_modelS_r1s1_vector}, we show the different kernels for modes $\ell$ between 0 and 500, and
for frequencies $\omega/(2\pi)$ 
between \num{1}\si{\micro\Hz} and 10\si{\milli\Hz} with steps 
\num{1}\si{\micro\Hz}. This amounts to a total of \num{501e4} 
values on each kernel. The different kernels show ridges of stronger power corresponding to the surface gravity mode (ridge with the lowest frequencies) and pressure modes. The frequencies of these modes have been measured in \cite{Larson2015} and agree well with the location of maximum power in our directional kernel, see \cref{figure:spectrum_modelS_r1s1_vector_with_obs}.
%which are also visible in the solar observations.
%The computations use attenuation of 2\si{\micro\Hz} to obtain sharper ridges in
%the slice at fixed mode in \cref{figure:spectrum_modelS_r1s1_vector-scalar-modes},
%therefore we solve the original equation.

% ------------------------------------------------------------
\setlength{\modelwidth} {4.70cm}
\setlength{\modelheight}{4.25cm}
% ------------------------------------------------------------
\begin{figure}[ht!]\centering
  \renewcommand{\modelfileA}{maps_original-eq/Vector_modelS_s1r1_att02uHz_mode0to500_freq1to10000mHz_GPPabs_Znonlocal-rmax1.0010_logscale7to13}
  \renewcommand{\modelfileB}{maps_original-eq/Vector_modelS_s1r1_att02uHz_mode0to500_freq1to10000mHz_GBPabs_Znonlocal-rmax1.0010_logscale5to13}
  \renewcommand{\modelfileC}{maps_original-eq/Vector_modelS_s1r1_att02uHz_mode0to500_freq1to10000mHz_GBBabs_Znonlocal-rmax1.0010_logscale1to12}
  \pgfkeys{/pgf/fpu=true}
  \pgfmathsetmacro{\cminA}{1e7} \pgfmathsetmacro{\cmaxA}{1e13}
  \pgfmathsetmacro{\cminB}{1e5} \pgfmathsetmacro{\cmaxB}{1e13}
  \pgfmathsetmacro{\cminC}{1e1} \pgfmathsetmacro{\cmaxC}{1e12}
  \pgfkeys{/pgf/fpu=false}
  \renewcommand{\cbtitleA}{$\vert\GPP(1,1)\vert$}
  \renewcommand{\cbtitleB}{$\vert\GBP(1,1)\vert$}
  \renewcommand{\cbtitleC}{$\vert\GBB(1,1)\vert$}

\subfloat[][Directional kernels on a logarithmic scale.]
      {\makebox[.95\linewidth][c]{
      \includegraphics[scale=1]{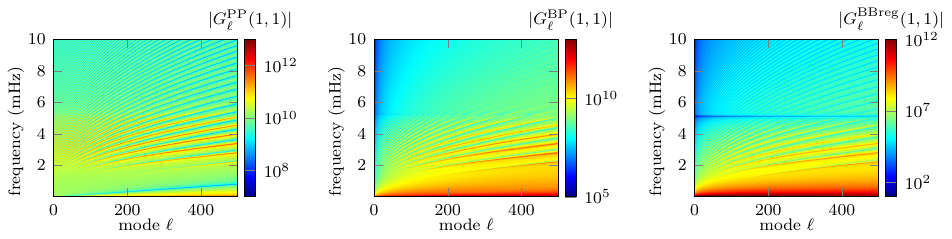}}}

\vspace*{-5.60em}
\begin{tikzpicture}
  \node[] at (0,0) (ghost) {};
  \draw[gray!80!white,densely dotted,opacity=1.0,line width=2pt] 
       (-14.2,0) -- (-13.3,0) -- (-13.3,0.6) -- (-14.2,0.6) -- cycle; 
  \draw[gray!80!white,densely dotted,line width=1.50pt,->] 
       (-13.5,0) -- (-13.5,-1.70); 
\end{tikzpicture}

\vspace*{-1.70em}

\subfloat[][$\GPP$ for selected interval ($\omega$,$\ell$)
            using linear color scale (left) and with the observed mode frequencies
            on top (white dots, right).]
           {\includegraphics[scale=1]{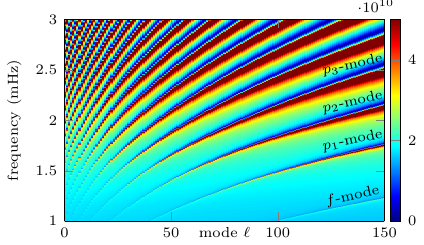}
            \hspace*{2em}
            \includegraphics[scale=1]{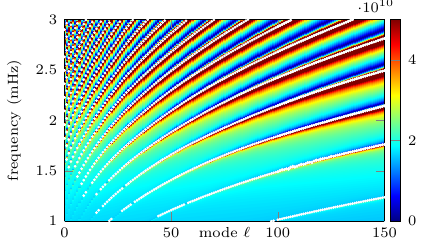}
            \label{figure:spectrum_modelS_r1s1_vector_with_obs}}

\caption{Directional kernels computed from \cref{algorithm:assemble} solving 
         the original equation with boundary condition 
         $\ZrbcNL$ set in $\rmax=\num{1.001}$; 
         comparison with the solar observed mode frequencies \cite{Larson2015} (bottom). }
\label{figure:spectrum_modelS_r1s1_vector}
\end{figure}

\begin{figure}[ht!]\centering
%  \subfloat[][Mode $\ell=2$.]
%               {\includegraphics[scale=1]{figures/kernels/standalone-data1d_s1r1_att02uHz_mode002_with_modes_and_obs_solo}} %\hspace*{.50em}
%  \subfloat[][Mode $\ell=400$.]
%             {\includegraphics[scale=1]{figures/kernels/standalone-data1d_s1r1_att02uHz_mode400_with_modes_and_obs_solo}}
 \includegraphics[scale=1]{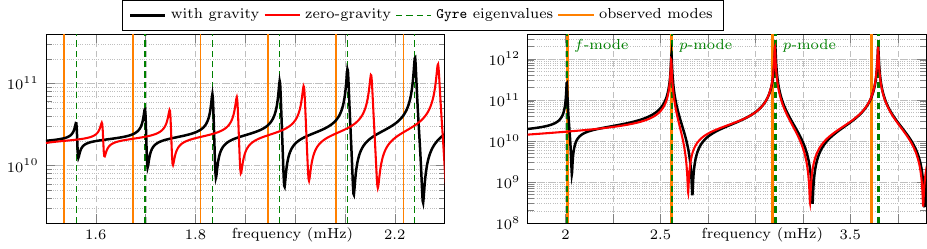}
 \vspace*{-1.50em}

 \subfloat[][mode $\ell=2$.]{\makebox[.50\linewidth][c]{}} \hfill
 \subfloat[][mode $\ell=\num{400}$.]{\makebox[.40\linewidth][c]{}}
\caption{Absolute value of the kernel $\GPP$
         using background solar model \texttt{S-AtmoI}
         as a function of frequency for modes $\ell = 2$ and \num{400}.
         % and different choices of boundary conditions set in $\rmax=\num{1.001}$.
         The modes/eigenfrequencies computed with the software
         \texttt{gyre} \cite{Townsend2013} are indicated with the vertical
         dashed lines, while the observed mode frequencies \cite{Larson2015} are
         represented with vertical solid lines.
         The zero-gravity computations correspond to \cref{eq:zero-gravity_approximation}.}
\label{figure:spectrum_modelS_r1s1_vector-scalar-modes}
\end{figure}

For a more quantitative comparison we show cuts through the 
power spectrum for $\ell=2$ and $\ell=400$ in \cref{figure:spectrum_modelS_r1s1_vector-scalar-modes} 
using the full problem and the zero-gravity approximation \cref{eq:zero-gravity_approximation}. 
We overplot the observed frequencies measured in \cite{Larson2015} (for $\ell=2$) and \cite{Korzennik2013} (for $\ell=400$), as well as
the eigenfrequencies computed numerically with the software \texttt{gyre} \cite{Townsend2013}. 
In \texttt{gyre}, we also use the standard solar model-\texttt{S} with an isothermal boundary 
condition and solve the equations of stellar oscillations under the Cowling approximation as 
in our framework. 
We found a good agreement between the frequency of maximum power of our kernel 
and the eigenfrequencies from \texttt{gyre}. It corresponds also to the measured 
frequencies with a small deviation compared to the numerical values. 
For high values of $\ell$, this is due to the surface effect \cite{Rosenthal1999} 
(simplified treatment of the convection in the surface layers) while for small values 
of $\ell$, it is due to the Cowling approximation which becomes insufficient. 
\subsection{Internal gravity modes in the solar spectrum}
\label{subsection:g-modes}
% ----------------------------------------------------------------------------------------

The propagative region that appears in the interior when including
the gravity (cf.~\cref{Inpropregion::rmk,fig:potential:main01}) 
corresponds to the region of internal gravity modes ($g$-modes). 
These modes exist for frequencies below the buoyancy frequency $N$, 
that is, for frequencies below $0.5$\si{\milli\Hz}.
However, to be visible at the surface, these modes need to travel 
from the radiative interior (below $r=0.7$) to the 
surface and are thus strongly damped in the convective zone where 
the potential is positive. 
Contrary to pressure modes, a clear detection of internal gravity modes 
in the observations has yet to be achieved due to their very small expected 
amplitude at the surface \cite{Belkacem2022}.
 
To highlight them in the simulations, we compute the power spectrum for a 
source located at $s=0.7$ and keeping the observation height at the surface. 
\cref{figure:spectrum_modelS_r1s07_vector_modes} shows the resulting power 
spectrum for $\ell=2$. As expected the vectorial equation that includes 
gravity is necessary to see the gravity modes and the zero-gravity approximation 
displays only the (shifted) pressure modes. To validate our code, we overplot the 
eigenfrequencies from \texttt{gyre} which are in excellent agreement with 
the locations of maximum power in our power spectrum. 
Finally, note that while the pressure modes are present when neglecting gravity, 
their frequencies are shifted, underlying the importance of the gravity terms 
for low-degree modes.

% When the source is close to the surface, they are barely visible
% in the Green's kernels as they come from the interior
% (cf.~\cref{figure:spectrum_modelS_r1s1_vector-scalar-modes}
% where the source is in $s=1$).
% To highlight them, the source is now considered below
% the convection zone, with $s=0.7$, while receiver
% remains at the surface. \newline
% \flo{to do: explain g-modes propagate in the interior} \newline
% \flo{to do: explain g-modes are only low frequency $<0.5$\si{\milli\Hz}}    \newline
% \flo{to do: explain g-modes are only low modes?} \newline
% We picture the corresponding kernels at mode $\ell=2$
% in \cref{figure:spectrum_modelS_r1s07_vector_modes} using a log-log
% scale focusing on low frequencies (below 1 \si{\milli\Hz}),
% and comparing with the eigenfrequencies computed with
% software \texttt{gyre} \creff.

% ------------------------------------------------------------
\setlength{\modelwidth} {4.75cm}
\setlength{\modelheight}{3.80cm}
% ------------------------------------------------------------
%\begin{figure}[ht!]\centering
%  \renewcommand{\modelfileA}{maps_original-eq/Vector_modelS_s07r1_att02uHz_mode0to500_freq1to10000mHz_GPPabs_Znonlocal-rmax1.0010_logscale-80to12}
%  \renewcommand{\modelfileB}{maps_original-eq/Vector_modelS_s07r1_att02uHz_mode0to500_freq1to10000mHz_GBPabs_Znonlocal-rmax1.0010_logscale-80to12}
%  \renewcommand{\modelfileC}{maps_original-eq/Vector_modelS_s07r1_att02uHz_mode0to500_freq1to10000mHz_GBBabs_Znonlocal-rmax1.0010_logscale-80to12}
%  \pgfkeys{/pgf/fpu=true}
%  \pgfmathsetmacro{\cminA}{1e-80} \pgfmathsetmacro{\cmaxA}{1e12}
%  \pgfmathsetmacro{\cminB}{1e-80} \pgfmathsetmacro{\cmaxB}{1e12}
%  \pgfmathsetmacro{\cminC}{1e-80} \pgfmathsetmacro{\cmaxC}{1e12}
%  \pgfkeys{/pgf/fpu=false}
%  \renewcommand{\cbtitleA}{$\vert\GPP(1,1)\vert$}
%  \renewcommand{\cbtitleB}{$\vert\GBP(1,1)\vert$}
%  \renewcommand{\cbtitleC}{$\vert\GBB(1,1)\vert$}
% {\makebox[.95\linewidth][c]{\input{figures/kernels/skeleton_map_groupplotx3_log}}}
%      
%\caption{Absolute value of the directional kernels
%         for a source in \num{0.7} 
%         using solar background model \texttt{S-AtmoI}.
%         The computation is performed with \cref{algorithm:assemble}
%         solving the original equation with
%         2\si{\micro\Hz} attenuation with boundary 
%         condition $\ZrbcNLa$ set in $\rmax=\num{1.001}$.}
%\label{figure:spectrum_modelS_r1s07_vector}
%\end{figure}

% ------------------------------------------------------------
\begin{figure}[ht!]\centering

\includegraphics[scale=1]{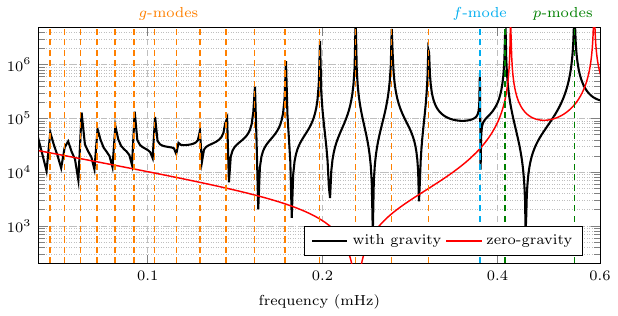}
        
\caption{Absolute value of the kernel $\GPP(0.7,1)$
         computed with small attenuation
         $\Gamma/(2\pi)=0.1$\si{\micro\Hz} to
         obtain disjoint ridges.
         The eigenfrequencies computed with the software \texttt{gyre} \cite{Townsend2013} are indicated
         with  vertical dashed lines.}
\label{figure:spectrum_modelS_r1s07_vector_modes}
\end{figure}

% As expected, we see that the scalar approximation is unable to capture
% the $g$-modes and the $f$-mode.
% The Green's kernel of the vector problem shows peaks that correspond
% perfectly to the eigenfrequencies computed.
% Therefore, to study the $g$-modes, one has to consider an interior
% source so that they are clearly captured.
% \flo{Obervations many g-modes very close to each other, comments? }

% =====================================================================
\section{Conclusion}
\label{section:conclusion}

For the Galbrun's equation under spherical symmetry, we
develop theoretical and numerical tools 
to compute efficiently and accurately 
the coefficients of the Green's tensor in VSH basis.
% We have constructed ABC  and investigated their accuracy 
% in approximating outgoing solution.
With \cref{coeff_orig::prop,compoG_reg:prop}, we characterize
explicitly the structure of the directional 
kernels, and propose \cref{algorithm:assemble} to 
compute their values for any height of sources and receivers 
from two simulations. 
Additionally, with the singularity of the directional kernels
prescribed analytically, \cref{algorithm:assemble} avoids 
dealing with Dirac-type sources, and provides a numerically 
stable and accurate way for their computation.

Regarding numerical resolution, we observe that 
 the original and Schr\"odinger modal ODE 
behave in a complementary manner.
The original form is stable for low attenuation but unable 
to handle large computational domains, while
the Schr\"odinger form is stable for arbitrarily
large intervals but unstable at very low attenuation. 
This means that for a resolution  up to low 
atmosphere ($\rmax \sim \num{1.001}$), the original form 
should be used to provide flexibility with low or
vanishing attenuation in parts of the interior region.
On the other hand, 
the Schr\"odinger form is important to obtain ABCs and 
is recommended if the background models extend high in the atmosphere.

To approximate outgoing-at-infinity solutions, we have constructed nonlocal and local ABCs which
contain a gravity term, 
and we have numerically compared their efficiency.
In the presence of attenuation, our nonlocal condition $\ZrbcNL$ 
can be employed at scaled height $r=1.001$.
For numerical resolution in 2.5D and 3D, one can employ 
$\ZrbcSAIGzero$  or $\ZrbcHFGa$ as a low-order condition, and 
$\ZrbcSAIGb$ or $\ZrbcHFGb$ which contain tangential derivatives.
% {\color{red}{For frequencies above $N_{\mathrm{rad}}$ \cref{Buoyrad_freq::def}, 
% its local approximation can be used they are particular needed 
% for resolution of the original equation in 2.5D and 3D,
% for which purpose we recommend : ...  as a low-order condition,
% and ... which contains one degree of tangential derivatives.}}
Finally,  \cref{algorithm:assemble} is used
to compute solar power spectra with background 
model \texttt{S}-\Atmo, which display
ridges in agreement with computed eigenvalues and observed modes, 
and make appear gravity modes missing with the scalar 
approximations.

\section*{Acknowledgments}
  % BUTTERFLY
  % Numpex
  % correct adastra.

% We thank the reviewers who have helped improve the manuscript.
  This work was partially supported by the INRIA associated-team
  \textsc{Ants} (Advanced Numerical meThods for helioSeismology)
  and the ANR-DFG project \textsc{Butterfly}, grant number ANR-23-CE46-0009-01. 
  DF and LG acknowledge funding by the Deutsche
  Forschungsgemeinschaft (DFG, German Research
  Foundation) -- Project-ID 432680300 -- SFB 1456 (project C04).
  This work was partially supported by the EXAMA (Methods and Algorithms at Exascale) 
  project under grant ANR-22-EXNU-0002.
  This project was provided with computer and storage resources by GENCI at CINES 
  thanks to the grant gda2306 on the supercomputer Adastra's GENOA partition. 
  FF acknowledges funding by the European Union 
  with ERC Project \textsc{Incorwave} -- grant 101116288. 
  Views and opinions expressed are however those of the authors 
  only and do not necessarily reflect those 
  of the European Union or the European Research Council 
  Executive Agency (ERCEA). Neither the European Union nor the
  granting authority can be held responsible for them.
  
%\appendix
%\input{sections/appendix_flo} 

% ========== Bibliography
%\footnotesize 
\bibliographystyle{siamplain}
\bibliography{sections/bibliography}

\end{document}